\documentclass[sigconf, screen, nonacm]{acmart}

\usepackage{algorithm}
\usepackage{algpseudocode}
\usepackage{enumitem}
\usepackage{multirow}
\usepackage{subcaption}
\usepackage{soul}
\usepackage[disable]{todonotes}
\usepackage{xspace}
\usepackage{xurl}
\usepackage{tikz}

\theoremstyle{definition}
\newtheorem{definition}{Definition}

\theoremstyle{plain}
\newtheorem{theorem}{Theorem}
\newtheorem{corollary}{Corollary}

\newcommand{\F}{\mathcal{F}}
\newcommand{\Sim}{\mathcal{S}}

\newcommand{\OSort}{\mathsf{OSort}}

\newcommand{\OCmpSet}{\mathsf{OCmpSet}}
\newcommand{\OHT}{\mathsf{OHT}}
\newcommand{\PADSIZE}{\mathsf{PADSIZE}}

\def\system{\texttt{TENNOR}\xspace}
\def\rpname{TENNOR\xspace}
\def\ptitle{\rpname: Trustworthy Execution for Neural Networks through Obliviousness and Retrievals}
\def\pkeywords{neural networks, obliviousness, trusted execution environments}

\hypersetup{
  pdftitle = {\ptitle},
  pdfkeywords = {\pkeywords}
}

\newcommand*\circledbl[1]{\tikz[baseline=(char.base)]{
\node[shape=circle,draw,fill=black,inner sep=1pt,text=white] (char) {#1};}}

\newcommand*\circled[1]{\tikz[baseline=(char.base)]{
\node[shape=circle,draw,inner sep=1pt] (char) {#1};}}

\def\addauthnote#1#2{%
	\expandafter\def\csname#1\endcsname##1{%
		\todo[inline,size=\footnotesize,color=#2]
			{\textbf{\underline{\texttt{#1}}:} ##1}\xspace}
}

\addauthnote{vpk}{blue!25}
\addauthnote{giuseppe}{green!25}
\addauthnote{evgenios}{orange!25}
\addauthnote{TODO}{yellow!25}
\addauthnote{modify}{purple!25}
\addauthnote{question}{fuchsia!25}
\addauthnote{note}{red!25}

\algtext*{EndFor}
\algtext*{EndProcedure}
\algtext*{EndFunction}
\algtext*{EndIf}

\graphicspath{{img/}}

\AtBeginDocument{%
  }

\setcopyright{none}
\copyrightyear{2026}
\acmYear{2026}
\acmDOI{10.1145/3830454.3846750}
\acmConference[CCS '26]{Proceedings of the 2026 ACM SIGSAC Conference on Computer and Communications Security}{November 15--19, 2026}{The Hague, Netherlands}
\acmBooktitle{Proceedings of the 2026 ACM SIGSAC Conference on Computer and Communications Security (CCS '26), November 15--19, 2026, The Hague, Netherlands}
\acmISBN{979-8-4007-2871-6/2026/11}

\begin{document}

\author{Zifan Qu}
\orcid{0009-0001-2049-6416}
\affiliation{
 \institution{George Mason University}
 \city{Fairfax}
 \state{VA}
 \country{USA}
}
\email{zqu4@gmu.edu}

\author{Vasileios P. Kemerlis}
\orcid{0000-0002-6528-437X}
\affiliation{
 \institution{Brown University}
 \city{Providence}
 \state{RI}
 \country{USA}
}
\email{vpk@cs.brown.edu}

\author{Giuseppe Ateniese}
\orcid{0000-0002-0848-878X}
\affiliation{
 \institution{George Mason University}
 \city{Fairfax}
 \state{VA}
 \country{USA}
}
\email{ateniese@gmu.edu}

\author{Evgenios M. Kornaropoulos}
\orcid{0000-0002-9982-6536}
\affiliation{
 \institution{George Mason University}
 \city{Fairfax}
 \state{VA}
 \country{USA}
}
\email{evgenios@gmu.edu}

\title{\ptitle}
\titlenote{This is the extended version of the paper accepted at ACM CCS 2026.
Please cite the conference version: \url{https://doi.org/10.1145/3830454.3846750}.}

\begin{abstract}
Training neural networks on sensitive data in untrusted cloud environments
requires simultaneously achieving computational efficiency and rigorous privacy
guarantees. Sparsification techniques, essential for scalable training of wide
layers, expose input-dependent memory-access patterns that are
visible and can be exploited by a host OS/\-hypervisor, even when computation
is protected by a Trusted Execution Environment. 

We present \system, a system that resolves this tension by co-designing the
neural network training pipeline with doubly oblivious primitives, eliminating
value-dependent neural network
access-pattern leakage while also utilizing adaptive sparsification. \system
recasts sparse neuron activation as a locality-sensitive hashing (LSH)
retrieval problem, reducing secure sparsification to doubly oblivious accesses
over an LSH data structure. To eliminate the prohibitive storage cost of
``multi-table'' LSH, we introduce Multi-Probe Winner-Take-All (MP-WTA): the
first multi-probe scheme for rank-based LSH, achieving a $50\times$ reduction
in (hash table) memory while preserving model accuracy. We evaluate \system on
multi-label classification benchmarks, with output layers of up to
hundreds of thousands of neurons under different batch sizes inside an Intel TDX Trusted Domain,
achieving
per-batch training
speedups of
$13\times$--$470\times$ over a Path ORAM baseline and reducing a 208-hour run
to about 26 minutes.

\end{abstract}



\settopmatter{printacmref=false, printccs=false, printfolios=false}
\renewcommand\footnotetextcopyrightpermission[1]{}

\maketitle

\pagestyle{empty}

\section{Introduction}
\label{sec:introduction}
\renewcommand{\algorithmicrequire}{\textbf{Input:}}
\renewcommand{\algorithmicensure}{\textbf{Output:}}

The success of neural networks (NN) has led the community to develop a range of architectures, each tailored to different tasks and presenting their own computational challenges. 
Whether the pipeline is built around convolutional networks for vision, transformers for language, or recommendation models for personalization~\cite{he2016deep, vaswani2017attention, naumov2019deep, covington2016deep}, the choice of architecture fundamentally shapes how computation must be organized, scheduled, and optimized. 
Among these, a particularly demanding class consists of networks containing \emph{extremely wide layers}\footnote{The width
 of a network is defined to be the maximal number of nodes in a layer.}: i.e.,~layers whose dimensionality reaches into the hundreds of thousands, or even millions, of neurons. 
Such layers arise naturally in the output layers of extreme multi-label classification and recommendation systems~\cite{babbar2017dismec, yu2022pecos, medini2019extreme, schultheis2023towards, covington2016deep}.

\textbf{Efficient and Adaptive Access via Sparsification.\xspace} Wide layers are computationally demanding even in standard, non-privacy-preserving settings. 
When a single layer contains tens of thousands of neurons, the cost of organizing, storing, and accessing its parameters becomes a bottleneck. 
This state of affairs has motivated a substantial body of work on \emph{sparsification}~\cite{hoefler2021sparsity}. 
That is, rather than activating all neurons in the wide layer for every input, only a small subset participates, dramatically reducing computation per access. 
Sparsification comes in two flavors: \emph{non-adaptive}~\cite{han2015learning, 
frankle2018lottery}, in which the network is pruned to permanently remove neurons that do not contribute to model performance; and \emph{adaptive}~\cite{bengio2015conditional, 
shazeer2017outrageously, spring2017scalable}, in which sparsity is applied dynamically, at access time, depending on the input. 
In this work, we focus on the latter. 
The resulting techniques for making wide-layer networks scalable have become a standard in modern systems~\cite{shazeer2017outrageously, fedus2022switch, chen2020mongoose}.

While inference accesses the wide layer once, training requires repeated accesses. 
This repetition makes wide-layer training the dominant computational bottleneck, far exceeding the cost of inference in any realistic deployment scenario~\cite{spring2017scalable, chen2020slide, chen2020mongoose, schultheis2023towards}. 
Crucially, adaptive sparsification transforms wide-layer computation from dense
matrix operations into sparse, irregular memory accesses, which is a workload
that aligns poorly with GPU architectures that are optimized for regular data
parallelism. CPUs, with their deep cache hierarchies and flexible control flow,
are a natural fit for such access patterns, and \emph{prior work has shown that
sparse training on CPUs can match or outperform GPUs on wide-layer
networks~\cite{chen2020slide, daghaghi2021tale}}.

\textbf{Trustworthy NN Training.\xspace} In this work, we introduce a \emph{trustworthiness} dimension to this setting by requiring that model parameters remain confidential via a Trusted Execution Environment (TEE), where an untrusted OS/\-hypervisor can observe memory access patterns even when data is encrypted. 
Unfortunately, sparsification, which is essential for scalability, becomes a liability in adversarial settings~\cite{kato2023olive, ohrimenko2016oblivious}. 
The pattern of which a subset of neurons is adaptively activated for a given input is precisely the kind of access pattern that must not be revealed, since it leaks information about the input values. 
A growing body of attacks has demonstrated that \emph{access patterns over encrypted data can be exploited} to recover the underlying plaintext~\cite{response-hiding,k-nn-attack,state-of-the-uniform,10.1109/SP.2019.00030,DBLP:conf/ccs/CashGPR15}. 
Thus, naively applying sparsification to a wide layer exposes an attack vector, as memory accesses are revealed. 
This tension between (1)~the efficiency gains from sparsification and (2)~the access-pattern leakage motivates the privacy-preserving system we propose in this work. 

\system (which stands for \underline{T}rustworthy \underline{E}xecution for \underline{N}eural \underline{N}etworks through \underline{O}bliviousness and \underline{R}etrievals) is a system that is designed around CPUs and synthesizes \emph{doubly-oblivious} primitives to enable private training on wide neural networks without revealing sparsification access patterns to the host OS/\-hypervisor. 
Central to our efficiency is the observation that sparsification can be cast as a retrieval task over a locality-sensitive hashing (LSH) data structure~\cite{indyk1998approximate, shrivastava2014asymmetric, spring2017scalable, chen2020slide, chen2020mongoose}. 
This framing allows us to adaptively identify active neurons without scanning the full layer, and to do so obliviously. 
We show that by co-designing the obliviousness mechanism with a sparse-activation structure, we achieve secure training on wide-layer networks at orders-of-magnitude lower cost than existing oblivious approaches, while providing rigorous guarantees against access-pattern leakage.

We make the following contributions:

\noindent\textbf{(1) Doubly-oblivious Wide-layer Training Pipeline.\xspace} We design and
implement \system, a doubly-oblivious training pipeline for wide
neural networks within TEEs that eliminates value-dependent access-pattern leakage beyond the public input support.  \system
carefully composes oblivious hash tables, oblivious sorting, and branch-free
primitives into a provably secure adaptive sparsification pipeline, with
security formally established via a simulation-based proof under a standard TEE
adversarial model.

\noindent\textbf{(2) Multi-Probe Winner-Take-All LSH.\xspace} At the heart of our design is a
newly introduced multi-probe scheme for rank-based LSH (dubbed MP-WTA), which
defines a principled notion of neighboring buckets under WTA's ordinal hashing
framework.  MP-WTA achieves retrieval quality equivalent to that of $50$
independent hash tables, but it uses only a single table, reducing LSH storage
by $50\times$ with negligible accuracy loss.

\noindent\textbf{(3) Orders-of-magnitude Performance over Baseline.\xspace} On extreme
multi-label classification benchmarks, with hundreds of thousands of output labels, \system
achieves $13\times$--$470\times$ speedup over a Path ORAM baseline running
inside Intel TDX.  For this performance, we introduce three targeted
optimizations (i.e., deferred dummy payload instantiation, parallel oblivious
LSH scanning, and OHT layout reuse) that collectively eliminate the dominant
bottlenecks while preserving double obliviousness.

\section{Background}
\label{sec:background}

\subsection{Trusted Execution Environments}
\label{sec:TEE}

Trusted Execution Environments (TEEs) (e.g.,~Intel SGX~\cite{costan2016intel},
ARM TrustZone~\cite{pinto2019demystifying}, AMD
SEV~\cite{sev2020strengthening}) are
\emph{hardware-enforced, isolated runtime environments} that protect certain
code and data from the rest of the system, including privileged software like
the OS or a hypervisor. In VM-based TEEs, such as Intel
TDX~\cite{cheng2024intel}, the protected workload (including a guest OS, if
present) runs inside a \emph{Trust Domain} (TD), while the host and hypervisor
remain outside the trust boundary. Beyond isolation, TEEs provide
\emph{attestation}, enabling remote parties to verify that a workload is
genuinely executing within a secure enclave. Together, these guarantees
underpin cloud \emph{confidential computing}, such as secure machine learning
(ML) services, allowing users to outsource sensitive data and computations without
exposing their contents to an infrastructure provider.

Despite these protections, TEEs do not defend against most classes of
\emph{side-channel attacks}. Prior work has demonstrated attacks
compromising \emph{confidentiality} by exploiting leakage from secret-dependent
memory accesses or control-flow behavior~\cite{brasser2017software, lee2020off,
lee2017inferring, moghimi2017cachezoom, moghimi2020copycat,
shinde2016preventing, van2017telling, xu2015controlled}. Complementing
architectural~\cite{costan2016sanctum, bourgeat2019mi6} and software
defenses~\cite{shih2017t, chen2017detecting}, \emph{data-oblivious programming}
offers an alternative mitigation strategy: programs are written such that their
memory access patterns are independent of any secret inputs
(\S\ref{sec:threat-model}).

\subsection{Oblivious Computation}
\label{sec:obliviousness}

At a high level, a computation is \emph{oblivious} if its memory access pattern
depends only on public information (e.g.,~input length and fixed parameters)
and not on sensitive values. More formally, a computation is oblivious iff any
two input sequences of the same length induce memory access patterns that are
indistinguishable to a computationally-bounded adversary. Obliviousness further
requires that branches, loop bounds, and memory indices in the executed code do
not depend on secret data.

\textbf{Oblivious RAM.\xspace}
\emph{Oblivious RAM} (ORAM), introduced by Goldreich and
Ostrovsky~\cite{goldreich1996software}, is a general-purpose technique for
achieving obliviousness. At a high level, ORAM conceals which logical block is
accessed by performing \emph{additional} reads and writes and periodically
shuffling data, so that the observed access pattern reveals only the number of
operations, not their targets. ORAM has been substantially improved over the
past few decades~\cite{stefanov2018path, wang2015circuit, patel2018panorama,
asharov2020optorama, asharov2023futorama}, yet it remains a generic
construction. \emph{Application-specific oblivious designs have been shown to
achieve significantly better performance, as we also demonstrate in this work.}

\textbf{Application-tailored Obliviousness.\xspace}
Beyond generic ORAM, the research community has developed \emph{oblivious data
structures} (ODS)~\cite{wang2014oblivious} and oblivious algorithms tailored to
specific tasks. Prior work has produced oblivious counterparts for a range of
algorithms and data structures, including transactional and map-like
systems~\cite{crooks2018obladi, tinoco2023enigmap}, search indexes and database
query processing~\cite{mishra2018oblix, eskandarian2017oblidb},
hashing~\cite{chan2017oblivious}, and sorting~\cite{batcher1968sorting,
asharov2020bucket, sasy2022fast, ngai2024distributed}. At the primitive level,
basic operations such as \emph{comparison}, \emph{swap}, and \emph{copy} must
also be oblivious. Several works~\cite{sasy2018zerotrace, dauterman2021snoopy,
sasy2022fast, ngai2024distributed, tinoco2023enigmap} employ branch-free
primitives, such as \texttt{CMOV} (x86) instructions and bitwise masking, to
ensure that access patterns are independent of the input.

\textbf{Double Obliviousness.\xspace}
\label{sec:double-oblivious}
In TEEs, access-pattern leakage can arise from three sources: (1)~memory
accesses inside the TEE, (2)~accesses to encrypted memory outside the TEE, and
(3)~contention on shared microarchitectural resources. Early work largely
focused on concealing access patterns at the boundary between the TEE and
external memory---i.e.,~point~(2). This scope was broadened by
Oblix~\cite{mishra2018oblix}, which introduced the notion of \emph{double
obliviousness}, requiring that access patterns within the TEE are also
protected---i.e.,~point~(1). Subsequent works~\cite{tinoco2023enigmap,
chamani2023graphos, dauterman2021snoopy} adopt this definition, ensuring that
operations inside the TEE are themselves oblivious.

\subsection{CPU-friendly Neural Networks}
\label{sec:SLIDE}

\emph{Deep neural networks} (DNNs) are commonly trained on GPUs and specialized
accelerators, as dense layers map naturally to large matrix multiplications.
In this work, we target CPU-based TEEs, making CPU-efficient DNN designs
particularly relevant; more broadly, \emph{CPUs are the platform of choice when
GPUs are unavailable (e.g.,~at the edge) or when power and cost constraints
dominate}~\cite{mittal2021survey}.

Running DNN inference on general-purpose CPUs is challenging for two reasons.
First, arithmetic throughput is limited: CPUs have fewer cores and narrower
SIMD units, making naive dense-layer execution orders of magnitude slower.
Second, performance is often memory-bound: CPUs rely heavily on caching to
compensate for slower RAM access, so models whose working sets exceed the
cache, or whose access patterns are irregular, suffer significant slowdowns.

\textbf{Sub-Linear Deep Learning Engine.\xspace}
Chen~et~al.~\cite{chen2020slide} introduced the \emph{Sub-Linear Deep Learning
Engine} (SLIDE): a CPU-oriented framework for training and inference that
replaces dense computation with adaptive sparsity. Given an input, SLIDE uses
\emph{locality-sensitive hashing} (LSH) to identify a small, input-dependent
active subset of neurons---in wide layers---and restricts forward and backward
passes to this subset. This reduces dense matrix multiplication to a small
number of hash-table lookups and sparse dot products, better exploiting CPU
strengths. However, because the active set is only an approximation of the true
top-activated neurons, accuracy can degrade if the active set is too small.
More critically for our setting, \emph{the accessed hash buckets and updated
neurons depend directly on the input, making naive SLIDE-style execution highly
access-pattern-sensitive and therefore insecure inside a TEE}.

\subsection{Locality-sensitive Hashing}
\label{sec:LSH}

\emph{Locality-sensitive hashing} (LSH) is a family of randomized hash
functions with the property that similar items are more likely to
\emph{collide} (i.e.,~land in the same bucket) than dissimilar
ones~\cite{indyk1998approximate}. Unlike traditional hashing, collisions in LSH
are intentional: they provide a fast mechanism for retrieving a short list of
candidates without scanning the full dataset.

Several LSH families have been proposed for different similarity measures,
including MinHash for Jaccard similarity~\cite{broder1998min}, SimHash for
cosine similarity~\cite{charikar2002simhash}, and Winner-Take-All hashing for
rank-based comparisons~\cite{yagnik2011power}. LSH has found application across
a range of domains, including approximate nearest-neighbor (ANN)
search~\cite{andoni2015practical}, recommendation
systems~\cite{shrivastava2014asymmetric}, 
and accelerating
ML~\cite{kitaev2020reformer}.
\emph{SLIDE
uses LSH for fast network sparsification during training} (\S\ref{sec:SLIDE}).

More formally, let $d$ be the \emph{distance} measure on the domain $S$.

\textbf{Definition 2.1 (LSH family).\xspace} \textit{A function family $H$ is called
$(r, c\cdot r, p_1, p_2)$-sensitive for $d$, if for any $p,q \in S$:}
\begin{itemize}
	\item \textit{If $d(p,q) \leq r$ , then $\Pr_{h\leftarrow H}[h(p) = h(q)] \geq p_1$,}
	\item \textit{If $d(p,q) > c\cdot r$, then $\Pr_{h\leftarrow H}[h(p)=h(q)] \leq p_2$}
\end{itemize}

For ANN search, we require $c > 1$ and $p_1 > p_2$: points within distance $r$
collide with probability at least $p_1$, while points farther than $cr$ collide
with probability at most $p_2$. In practice, this gap is amplified by
concatenating multiple hash functions and/\-or maintaining multiple hash tables,
then probing the corresponding buckets to retrieve candidates. We next describe
Winner-Take-All hashing, the rank-based LSH family used by both SLIDE and
\system.

\subsection{Winner-Take-All Hashing}
\label{sec:WTA}

Winner-Take-All (WTA)~\cite{yagnik2011power} is an LSH family for rank-based
similarity that captures \emph{ordinal relationships} between features rather
than metric distances. Unlike traditional LSH, WTA is invariant to positive
rescaling of features, making it particularly suitable for high-dimensional
representations common in deep learning.

\textbf{LSH for Rank Correlation.\xspace}
Most LSH families target metric distances, such as Euclidean or angular
distance. WTA instead targets \emph{rank correlation}, which depends only on
the relative ordering of feature values rather than their absolute magnitudes.
Under this notion, two vectors are similar when their features share the same
ranking, even if the actual values differ. Rank correlation is quantified by
counting pairwise order agreements between two rankings as:
\begin{equation}
    \begin{aligned}
    PO(X,Y) = \sum_{i=1}^{n} \sum_{j < i} T\left((x_i - x_j)(y_i - y_j)\right),
    \end{aligned}
    \label{equation:PO-prelim}
\end{equation}

$x_i$ and $y_i$ are the values of the $i^{th}$ feature of $X, Y \in
\mathbb{R}^n$ and $T(z)$ is a threshold function that takes value $1$, if
$z>0$, or value $0$, if $z\leq 0$. Informally, this counts the number of
feature pairs on which $X$ and $Y$ agree in ordering (i.e.,~either both $x_i >
x_j$ and $y_i
> y_j$, or vice versa). A high $PO$ value indicates that the two vectors induce
very similar rankings over their $n$ features. \emph{Rank correlation is
particularly well-suited for domains with noisy, heterogeneously scaled, or
sparse feature representations}.

\textbf{WTA Hashing.\xspace}
While rank correlation is theoretically well-founded, exact computation
requires $\mathcal{O}(n^{2})$ comparisons per vector and does not scale. WTA
approximates rank similarity by sampling a random subset of features and
recording only the winner---i.e.,~the index of the maximum value---within that
subset. Concretely, each hash function $h_i(\cdot)$ is parameterized by a fixed
random permutation and a window size $M$; given a vector $X$, it outputs the
index of the maximum among the $M$ sampled features. Applying $K$ such
functions yields a \emph{signature} $\mathbf{h}(X) = (h_1(X), \ldots, h_K(X))$.
\emph{Two vectors with similar rankings tend to agree on the winner across many
sampled windows, so the Hamming distance between their signatures serves as a
proxy for rank-based similarity.}

\section{Threat Model}
\label{sec:threat-model}

We assume a TEE-based threat model, aligned with prior
work~\cite{chamani2023graphos, mishra2018oblix, tinoco2023enigmap}, which
\emph{combines protected execution with double obliviousness}. In our
prototype, the TEE is instantiated as an Intel TDX TD/\-VM. We assume a
\emph{powerful adversary} $\mathcal{A}$ that: (1)~controls the host software
stack (e.g.,~the host OS and/\-or hypervisor); and (2)~can observe memory access
patterns generated by the TD during execution.

In particular, $\mathcal{A}$ can observe memory accesses at both page and cache-line granularity.
However, $\mathcal{A}$ cannot access processor-internal secrets within the TEE, including the hardware root key (used for memory encryption, attestation, key derivation, etc.), nor can it inspect CPU register contents during TEE execution.

\emph{Our guarantee applies at the algorithmic access pattern
level}.
Additionally, we consider power-analysis~\cite{lipp2021platypus,
kogler2023collide+}, rollback~\cite{matetic2017rote, zhang2024cachewarp},
timing-based~\cite{wang2022hertzbleed, wilke2024tdxdown}, and
denial-of-service~\cite{cheng2024intel} attacks out of scope. We use
Snoopy's~\cite{dauterman2021snoopy} code primitives (i.e.,~copy, sort), but we
do not claim whole-program, binary-level obliviousness after compiler
optimizations. Our goal is to eliminate value-dependent NN access-pattern leakage under the
aforementioned adversarial observation capabilities.

Our system runs inside a Trusted Domain (TD) provided by Intel TDX on an
untrusted cloud server. The user's client establishes a secure channel to the
remote TD/\-TEE using remote attestation, derives a session key, and transmits
sensitive data over this channel. The code executes within the TD, and all
sensitive data are processed within the TD. Unlike an SGX-style enclave
setting, where one often distinguishes enclave memory from memory
\emph{outside} the enclave, the relevant boundary here is the TD/\-VM  vs.
entities \emph{external} to the TD/\-VM (i.e.,~OS, VMM).
Although our prototype focuses on Intel TDX, the threat model considered relies
only on standard TEE guarantees (e.g.,~protected execution, memory encryption)
and therefore the same high-level objectives are also relevant to other TEEs,
such as Intel SGX and AMD SEV---though the concrete trust boundary and system
instantiation may differ.

\textbf{End-to-end Obliviousness.\xspace}
\system treats the number and positions of non-zero input features as public
parameters, while the corresponding feature values remain private.
Specifically, \system protects value-dependent access patterns within the neural network but does not hide which input features are non-zero.
The reported measurements in this work exclude the additional cost of hiding which input features are non-zero,
but we discuss a potential front-end that can make \system end-to-end oblivious.

One straightforward way to create a front-end that hides which input features are non-zero is to ``touch'' and compute on every dimension in the original input, including dimensions whose values are zero. 
Towards analyzing the complexity of this approach, let $D$ denote the input dimensions and $N_0$ the number of neurons in the first layer.
For each batch of $B$ input vectors, the front end scans all $D$ dimensions in a fixed order and performs the same multiply-accumulate operations for each of the $N_0$ first-layer neurons, regardless of the input value.
This task can be parallelized by assigning fixed ranges of features or neurons to $p$ threads, resulting in $O(BN_0D)$ total work parallelized in $O(\frac{BN_0D}{p})$ time.
Across all workloads evaluated in Section~\ref{sec:eval-performance}, this front-end adds at most $0.07\%$ to \system's per-batch runtime for 44 threads.
Note that this overhead is small because the sequential scan of high-dimensional inputs exhibits strong memory locality.
In terms of computation, every input feature (zero or non-zero) needs to perform a multiplication and a write; thus, the end result is doubly oblivious. 
Appendix~\ref{app:input-support} provides the per-dataset measurements.

\begin{figure*}[t]
\centering
\includegraphics[width=0.95\linewidth]{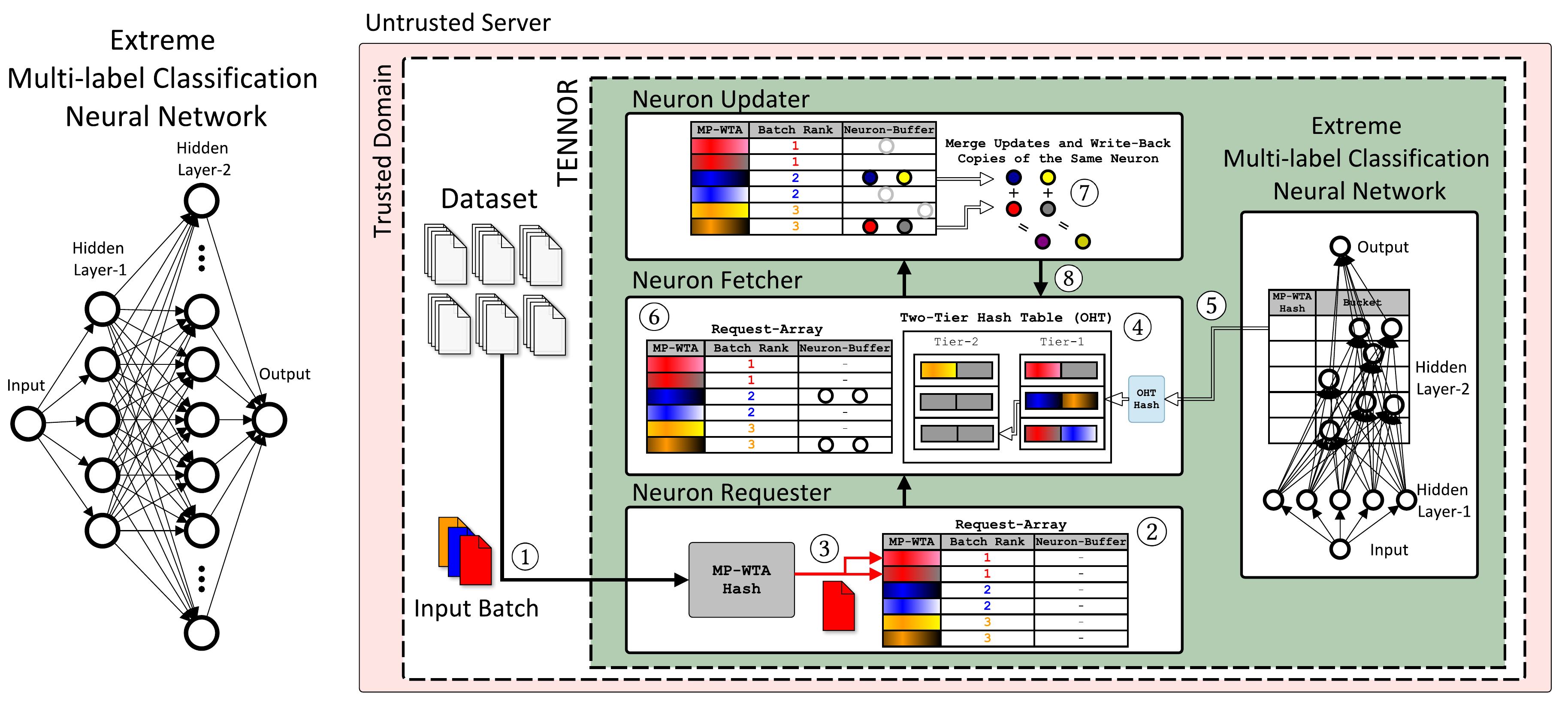}
\caption{Design of \system. The components \texttt{Neuron Requester},
	\texttt{Neuron Fetcher}, and \texttt{Neuron Updater} follow a
	doubly-oblivious design when accessing or updating the wide NN stored in
	Locality-Sensitive Hashing (LSH) table(s), shown on the
	right-hand side. The LSH tables use our newly introduced Multi-Probe
	Winner-Take-All (MP-WTA) hashing scheme (see~\S\ref{sec:MP-WTA}).}
\label{fig:overview}
\end{figure*}

\section{Overview}
\label{sec:overview}

\subsection{Key Insights}
\label{sec:insights}

Secure Neural Network training on an untrusted cloud requires \emph{double
obliviousness} at the algorithmic access-pattern level, i.e.,~data-dependent access patterns, across the memory hierarchy, must be avoided. 
The key challenge is to introduce new techniques and carefully compose existing ones,
tailored to the wide NN training pipeline so as to produce an efficient
doubly-oblivious system.

\circledbl{$1$}~\textbf{NN-adaptive Sparsification via Data Structure
Retrieval.\xspace} Textbook NN training processes all neurons for each input.
More advanced approaches apply static pruning to sparsify the network, reducing
both computation and storage. Instead, \emph{\system dynamically selects a
small subset of dominant neurons, at runtime, via adaptive sparsification}.
Specifically, \system transforms dense matrix multiplication into an LSH
table-retrieval task, a method originally introduced by Chen~et
al.~\cite{chen2020slide}. By casting sparsification as a retrieval problem,
\system avoids computing the subset of activated neurons via MPC or HE.
Instead, \system performs lookups in an LSH table that groups neurons during
initialization to facilitate sparsification. \emph{Consequently, secure
sparsification reduces to performing doubly-oblivious accesses to the LSH table(s)}.

\circledbl{$2$}~\textbf{Reduced LSH Tables via Multi-Probe WTA.\xspace} Given
that LSH is a data structure with probabilistic guarantees, a standard approach
to amplify the probability of success is to construct multiple LSH tables and
fetch a bucket from each. Instead, \emph{\system introduces a novel LSH scheme
for WTA that retrieves multiple buckets from a single table}, rather than
retrieving a single bucket from multiple tables (as SLIDE~\cite{chen2020slide}
does). This \emph{Multi-Probe WTA} scheme, which is of independent interest,
preserves NN training accuracy (see~\S\ref{sec:acc-eval}) while reducing space
usage by $50\times$.

\circledbl{$3$}~\textbf{Tailored Doubly-oblivious Design.\xspace} With these
design choices, wide NN training within \system reduces adaptive sparsification to mapping training data
points to neurons. The key efficiency question is which collection of
``objects'' (e.g.,~training data points or neurons) should serve as the target
of oblivious fetching. Rather than fixing a training point and obliviously
scanning across all neurons, \emph{\system reverses this approach completely:
for each fixed neuron, it obliviously fetches the training points that require
it}. Naively applying this idea to the entire training set would be costly, due
to the scalability limitations of oblivious computation over such a large
collection (of training data). To address this, \system partitions the training data into batches
and performs oblivious fetching per batch. Overall, the linear scan of the LSH
table that stores neurons is doubly-oblivious, while the request entries within
each batch are organized in a doubly-oblivious hash
table~\cite{chan2017oblivious}, providing the overall obliviousness guarantees.
Hence, under cache-line-level leakage, \system provides double obliviousness.

\subsection{Approach}
\label{sec:approach}

Figure~\ref{fig:overview} shows the high-level workflow and components
of \system. The system operates entirely within an Intel TDX Trusted Domain
(TD) and is designed to provide double obliviousness.

\textbf{NN Preprocessing.\xspace} During initialization,  \system samples hash functions for the wide layer of the NN and stores that
layer in the LSH key-value store shown on the right-hand side of
Figure~\ref{fig:overview}. Importantly, each LSH key-value pair stores the full
information associated with a neuron (e.g.,~weights, gradients, bias), rather
than a pointer that must be dereferenced to retrieve neuron state from another
data structure. The training data is then partitioned into batches.
\circled{1}~The workflow begins when a batch of training inputs is transmitted
to \system, and execution proceeds through three doubly-oblivious components (outlined below).

\textbf{Neuron Requester.\xspace} Each input in the batch ``requests'' a subset
of active neurons for training. Multiple inputs may require the same neuron,
which can lead to redundant fetches. \circled{2}~To handle this efficiently, we
organize neuron requests from all inputs into a \texttt{Request-Array} that
stores the triplet: ($i$)~the probed \texttt{LSHbucketID}, ($ii$)~the input's
position in the batch, and ($iii$)~a placeholder buffer for the neuron data
(to be populated later). \circled{3}~A single
training data point probes multiple locations in the LSH table, generating
multiple entries in the \texttt{Request-Array}.

\textbf{Neuron Fetcher.\xspace} The \texttt{Request-Array} does not group
requests by neuron, so determining which training data points in a batch
trigger a particular neuron requires additional bookkeeping. \circled{4}~To
address this, \system transforms the \texttt{Request-Array} into an Oblivious
Hash Table (OHT)~\cite{chan2017oblivious}, which acts as an oblivious key-value store where the key is
an \texttt{LSHbucketID} and the value consists of the corresponding \texttt{Request-Array}
triplet(s). At this stage, the triplets do not yet store any neurons.
\circled{5}~Next, \system performs a linear scan over the LSH table; at each
iteration, the current \texttt{LSHbucketID} is used as a key to the OHT to locate the
associated triplet(s) and populate them with the neurons stored in the
corresponding bucket. \circled{6}~After the scan completes, the OHT is
transformed back into a \emph{populated} \texttt{Request-Array} containing the
requested neurons. If multiple triplets request the same \texttt{LSHbucketID},
\system creates a separate copy of the padded bucket contents for each triplet.

\textbf{Neuron Updater.\xspace} Once the \texttt{Request-Array} has been
populated, \system performs feedforward and backpropagation only on the neurons
contained in it. \circled{7}~Because the same neuron may appear in multiple
triplets, \system obliviously aggregates the gradient updates from all copies
to compute a single consolidated update as part of the backpropagation. \circled{8}~Finally, \system
obliviously writes the updated neurons back to the LSH table by invoking Neuron
Fetcher again, this time with the updated neurons. 
Unlike the first invocation
of Neuron Fetcher (which takes place during feedforward) that transfers neurons from the LSH table to the OHT, this
step (now in backpropagation) writes the updated neurons from the OHT back to the LSH table.

Although our prototype targets Intel TDX, its modular design allows it to be
adapted to other TEE implementations, such as Intel SGX and AMD SEV, with
minimal (if any) changes.

\section{Design and Implementation of \rpname}
\label{sec:design-implementation}

\subsection{LSH Storage Reduction}
\label{sec:MP-WTA}

Recall that amplifying the probability of success in LSH typically requires
constructing multiple hash tables and querying the corresponding bucket in each
table. An alternative approach from the information retrieval community, called
\emph{Multi-Probe LSH} (MP-LSH)~\cite{lv2007multi}, addresses this requirement
by probing multiple buckets within a \emph{single} hash table. The key
intuition is that if a stored item is ``close'' to the query item $q$, but does
not fall into the same bucket, it is likely to reside in a ``nearby''
bucket---i.e.,~a bucket whose hash value differs only slightly from that of
$q$. Accordingly, MP-LSH probes a sequence $\texttt{Seq}$ of neighboring
buckets around the original hash $\mathbf{h}(q)$ \emph{within the same
table}:~$\texttt{Seq}(q) = (\mathbf{h}(q), \mathbf{h}_1(q), \mathbf{h}_2(q),
\ldots, \mathbf{h}_z(q)).$ Each bucket $\mathbf{h}_i(q)$ differs from the original
$\mathbf{h}(q)$ by a \emph{perturbation} $\delta_i$. This strategy achieves
search quality comparable to standard LSH while requiring significantly fewer
hash tables.

\textbf{Multi-Probe WTA.\xspace} To the best of our knowledge, no multi-probe scheme has been proposed for WTA in the literature. The main challenge lies in
defining what constitutes a ``neighboring'' bucket within WTA's rank-based
hashing framework. In the original MP-LSH~\cite{lv2007multi} work,
perturbations are typically $\pm 1$. 
That is, for a given dimension of $\mathbf{h}(q)$,
the two neighboring probe buckets are obtained by incrementing or decrementing
that dimension's integer value by one. In contrast, WTA hashing operates
differently; it samples a subset of elements from the query vector and returns
the index of the maximum element within that subset. Consequently, unlike LSH
for Euclidean distance, there is no obvious notion of a ``neighbor'' LSH bucket in
this ranking-based setting.

\textbf{Our Approach.\xspace} Suppose the WTA signature of a query (or vector) $q
\in \mathbb{R}^m$ is $\mathbf{h}(q) = \bigl(h_1(q), \ldots, h_K(q)\bigr)$,
where each hash $h_i(q)$ returns the index of the maximum feature value among the
$M$ sampled features used in the $i^{\text{th}}$ WTA function.
If multiple sampled features have the same value, ties are broken
deterministically according to the fixed sampled order of that hash function.
Our approach defines a neighboring LSH bucket by considering the index of the \emph{second-largest} feature value within the corresponding sampled subset; this index is denoted by $h^2_i(\cdot)$. 
We generate a perturbed hash by selecting a signature dimension $i$ and replacing $h_i(q)$ with $h^2_i(q)$. 
Consider two vectors with nearly identical feature values except for a single feature that appears as a large outlier in one of the vectors; if this outlier is sampled during the computation of $h_i$, its index will dominate that position of the signature. 
\emph{By falling back to
the second-largest value, we obtain a more robust representation that avoids
penalizing otherwise similar vectors that differ only in the identity of the
maximum feature}.

\textbf{Further Relaxation.\xspace} The above approach perturbs a single
position of $\mathbf{h}(q)$ using one specific alternative. We extend this idea along two axes. First, instead of considering only the
second-largest value, we allow the perturbed position to use the index of the
$r$-th largest feature value, where $2 \le r \le M$. We define $h^r_i(q)$ as
the function that returns the index of the $r$-th largest element within the
sampled subset associated with the $i$-th WTA hash function.

Second, we allow simultaneous perturbations of up to $N$ positions of signature 
$\mathbf{h}(q)$, where $1 \le N \le K$. If the number of perturbed positions is
$i$, where $1 \le i \le N$, then there are $\binom{K}{i}$ ways to choose the
$i$ perturbed positions among the $K$ signature dimensions. For each chosen
position, there are $(r-1)$ alternative hash values available, namely from the
second-largest $h^2_*(q)$ through the $r$-th largest $h^r_*(q)$. Thus, there
are $(r-1)^{\,i}$ possible assignments of alternative hash values across the
$i$ perturbed positions.

Given the above, the total length of the probing sequence is $\mathsf{lenSeq} =
1 + \sum_{i=1}^{N} \binom{K}{i}(r-1)^{\,i}$, where the first term corresponds
to the original LSH bucket $\mathbf{h}(q)$. In our implementation, we set $r=3$ and
$N=3$. Under these parameters, position $i$ of the signature can be replaced by
either $h^2_i$ or $h^3_i$, and up to three positions may be perturbed
simultaneously. Algorithm~\ref{alg:oblivious_mpwta_mp} describes our
\texttt{MP-WTA}.

\begin{algorithm}[tb]
  \caption{\textsc{MP-WTA}$_{r=3,\,N=3}$}
  \footnotesize
  \label{alg:oblivious_mpwta_mp}
  \begin{algorithmic}[1]
    \Require Query vector $q$, Sequence of WTA hash functions $h_1,\ldots,h_K$
    \Ensure Probing sequence $\texttt{Seq}$ of length $\mathsf{lenSeq}=1+\sum_{i=1}^{N}\binom{K}{i}(r-1)^{i}$, with $r=3$ and $N=3$
    \State Compute $\mathbf{h}(q)=\bigl(h_1(q),\ldots,h_K(q)\bigr)$ and $\texttt{Seq} \gets \emptyset$
    \For{$i\in\{1,\ldots,K\}$} compute $h^2_i(q),\;h^3_i(q)$ \Comment{$2^{nd}$ and $3^{rd}$ maximum for all $h_i$}
    \EndFor
    \State $\texttt{Seq}\gets\bigl(\mathbf{h}(q)\bigr)$
    \For{$n\in\{1,2,3\}$} \Comment{number of positions to perturb}
      \For{\textbf{each} subset $S\subseteq\{1,\dots,K\}$ of positions s.t. $|S|=n$} 
        \For{\textbf{each} tuple $(r_j)_{j\in S}$ with values $r_j\in\{2,3\}$}
          \State $\tilde{\mathbf{h}} \gets \mathbf{h}(q)$
          \Comment{copy original hash as the base for perturbation}
          \For{$j\in S$} $\tilde{h}_j \gets h^{r_j}_j(q)$
          \Comment{replace with $2$nd or $3$rd maximum}
          \EndFor
          \State Append $\tilde{\mathbf{h}}$ to $\texttt{Seq}$
        \EndFor
      \EndFor
    \EndFor
    \State \Return $\texttt{Seq}$
  \end{algorithmic}
\end{algorithm}

Note that the parameters are fixed in advance, independently of the training
data, yielding an oblivious multi-probe schedule.
The full exposition of $\mathbf{h}(\cdot), h^2_*(\cdot), h^3_*(\cdot)$ 's
oblivious algorithms is available in Appendix~\ref{app:mp-wta}.

\subsection{\rpname Components}
\label{sec:components}

Before describing the components, we fix the public parameters used by \system.
Let $P$ denote the collection of public parameters known to the adversary.
These include the batch size, the NN architecture and layer dimensions, the layer(s) to which LSH is applied, the public LSH parameters, the number of LSH buckets, the public probe-sequence length, the public padding parameter \texttt{PADSIZE}, the public parameters of the \texttt{OHT}, such as the number of bins and the bin capacity, the LSH rebuild period, the optimizer hyperparameters, and the current training-step index.
For the first layer, we also treat the non-zero feature indices of the input as public.
We treat the public number of training steps or batch iterations as part of $P$.
These parameters determine the shapes of intermediate arrays, the bounds of all loops, and the schedule of memory traversals used by the components.
All other information is private, including the values of the training data, the values of the model parameters, optimizer state, intermediate activations, gradients, and the concrete mapping from requests to buckets and from buckets to fetched neurons. (Refer to Appendix~\ref{sec:public-private} for additional details.)

\subsubsection{Neuron Requester}
\label{sec:neuron-requester}

In our prototype, \system applies LSH-based adaptive sparsification only to the wide layer, which is wide by nature in the extreme multilabel classification setting and thus dominates the training cost (see~\S\ref{sec:overview}). 
We note that other parameterizations of \system may apply LSH to multiple (even all) layers. 
The term \emph{input} to the layer being processed by the Neuron Requester refers to the activation vector fed into that layer. 
Algorithm~\ref{alg:neuron-requester} processes the batch's fetch requests for LSH that correspond to NN layer~$\ell$ and creates an intermediate data structure \texttt{Request-Array}. 
For a batch of $B$ inputs, Neuron Requester invokes \texttt{MP-WTA} on each input $q_i$ to obtain a probe sequence $\texttt{Seq}_i$ of length $\mathsf{lenSeq}_\ell$. 

\textbf{\textit{Obliviousness.}\xspace}
To preserve obliviousness, the Neuron Requester uses a statically pre-allocated array $\mathsf{ReqArray}_\ell$ of $T = B \cdot \mathsf{lenSeq}_\ell$ entries instead of variable-sized data structures, whose dynamic resizing pattern would leak information about the data.
The probe/request with local index $t$ for input $i$ is placed in the $(i-1)\cdot\mathsf{lenSeq}_\ell+t$ slot of the \texttt{Request-Array}. 
Each entry stores a triplet consisting of ($i$)~the signature \texttt{LSHbucketID}, ($ii$)~\texttt{Batch-Rank}, and ($iii$)~a \texttt{Neuron-Buffer} containing \texttt{PADSIZE} dummies. 
This way \texttt{Request-Array} has a fixed structure for the subsequent fetch stage, and its layout does not leak information about the input. 
Although the LSH bucket signature itself depends on the input values, Algorithm~\ref{alg:oblivious_mpwta_mp} computes the signature in a data-oblivious manner. 
We note that public system parameters $P$ fully determine the observable quantities such as the size of $\mathsf{ReqArray}_\ell$, the number of probes issued per input, and the write schedule. 
Thus, the observable layout and write schedule of the Neuron Requester reveal nothing beyond what is already known via $P$.
Notice that lines 6--12 of Algorithm~\ref{alg:neuron-requester} can be further optimized by populating each entry of the \texttt{Request-Array} in parallel since there is no dependence across the content or the computation of \texttt{Request-Array}'s entries.

\begin{algorithm}[tb]
    \caption{\textsc{Neuron Requester for layer $\ell$}}
    \label{alg:neuron-requester}
    \footnotesize
    \begin{algorithmic}[1]
        \Require Layer index $\ell$, batch of input data of size $B$, hash functions $h_1,\ldots,h_K$, and the maximum number of neurons for each \texttt{Neuron-Buffer}, denoted by $\mathsf{PADSIZE}$
    \Ensure Initialized $\mathsf{ReqArray}_\ell$
    
    \State $K \gets$ \# of hash functions at layer $\ell$
    \State $\mathsf{lenSeq}_\ell \gets$ length of probing sequence \texttt{Seq} for the LSH of layer $\ell$
    \State $T \gets B \cdot \mathsf{lenSeq}_\ell$ \Comment{total number of LSH probes for batch $B$}
    \State Allocate empty \texttt{Request-Array} $\mathsf{ReqArray}_\ell[0..T-1]$
    
    \State \textit{// Run MP–WTA and construct $\mathsf{ReqArray}$}
    \For{each input $i = 1$ to $B$}
        \State $\texttt{Seq}_i \gets$ MP-WTA$_{r=3,\,N=3}$($q_i, (h_1,\ldots,h_K$)), where $q_i$ is an entry of the batch
         \For{$t = 0$ to $\mathsf{lenSeq}_\ell-1$}
            \State Index of $\mathsf{ReqArray}$ to populate is $\texttt{ind}\gets(i-1)\cdot\mathsf{lenSeq}_\ell+t$
            \State $\mathsf{ReqArray}_\ell[\texttt{ind}].\texttt{LSHbucketID} \gets \texttt{Seq}_i[t]$
            \State $\mathsf{ReqArray}_\ell[\texttt{ind}].\texttt{Batch-Rank} \gets i$
            \State $\mathsf{ReqArray}_\ell[\texttt{ind}].\texttt{Neuron-Buffer}\gets$initialize with $\mathsf{PADSIZE}$ dummies
    \EndFor
\EndFor
    \State \Return $\mathsf{ReqArray}_\ell$

  \end{algorithmic}
\end{algorithm}

\subsubsection{Neuron Fetcher}
\label{sec:neuron-fetcher}

The main design goal of Neuron Fetcher is to populate the \texttt{Neuron-Buffer} for each \texttt{Request-Array} entry. 
To do that, buckets from the LSH table need to be retrieved and copied to the memory of \texttt{Neuron-Buffer}. 
Specifically, Neuron Fetcher's retrieval step must hide two data-dependent pieces of information: (1)~which LSH buckets are accessed during retrieval, and (2)~the number of neurons residing in each queried bucket for a given request. 
\system addresses the latter by padding every LSH bucket to a fixed public \texttt{PADSIZE} capacity with dummy neurons, so that each bucket retrieval fetches the same number of neurons.
Hiding which buckets are accessed requires a different mechanism.

As described in Section~\ref{sec:approach}, for efficiency reasons, \system adopts a reverse-mapping strategy; it iterates over LSH buckets and for each of them it identifies all requests that target it.
To achieve the desired security guarantees, we need a data structure that can \emph{obliviously} map LSH signatures to requests. 
This is achieved by reorganizing $\mathsf{ReqArray}_\ell$ into a two-tier Oblivious Hash Table (\texttt{OHT}), as proposed by Chan~et al.~\cite{chan2017oblivious}. 
\texttt{OHT} stores one entry per request. The key is \texttt{LSHbucketID}; the value is a triplet of \texttt{LSHbucketID}, \texttt{Batch-Rank}, and \texttt{Neuron-Buffer}. 
Thus, multiple requests in the batch may appear as distinct \texttt{OHT} entries with the same \texttt{LSHbucketID}. 

Because the \texttt{OHT} maps keys to bins using a hash function, entries
(requests) with different \texttt{LSHbucketID}s may be placed in the same bin
due to hash collisions. To allow the fetcher to identify the correct entries
when scanning a bin, each stored triplet also records its \texttt{LSHbucketID}
as a field, in addition to the \texttt{OHT} lookup key.

In more detail (see Algorithm~\ref{alg:fetcher-main}), Neuron Fetcher linearly scans all LSH buckets, computes the corresponding tier-1 and -2 \texttt{OHT} bins for the current bucket, and linearly inspects the entries in those bins. 
For each candidate entry, Neuron Fetcher uses: \texttt{OblCompare}~\cite{dauterman2021snoopy}, a \emph{data-oblivious comparison primitive}, to test whether its \texttt{LSHbucket\-ID} matches the current LSH bucket; and \texttt{OblChoose}~\cite{dauterman2021snoopy}, a \emph{data-oblivious conditional selection+assignment primitive}, to copy that bucket's padded contents into the corresponding \texttt{Neuron-Buffer}(s). (Refer to Appendix~\ref{app:obl-primitive} for details.)
After all buckets have been processed, Neuron Fetcher obliviously transforms the \texttt{OHT} back to \texttt{Request-Array}, whose \texttt{Neuron-Buffer}s now hold the retrieved neuron contents. 
Appendix~\ref{app:neuron-fetcher} details this procedure. 
The design of Neuron Fetcher is \emph{doubly-oblivious}. To build the
\texttt{OHT}, we instantiate \texttt{OblSort} with Batcher's bitonic sorting
network, whose compare and exchange schedule is fixed and
data-independent~\cite{batcher1968sorting}.

\begin{algorithm}[tb]
\caption{Neuron Fetcher for layer $\ell$ (Read mode; Write mode reverses the transfer direction)}
\label{alg:fetcher-main}
\footnotesize
\begin{algorithmic}[1]
\Require $\mathsf{ReqArray}_\ell$ of $T = B\cdot\mathrm{lenSeq}$ entries; LSH table $\mathsf{LSH}_\ell$ with $\mathrm{numLshBucs}$ buckets, each padded to PADSIZE; public OHT parameters $\mathrm{numBin}, \mathrm{binCap}$
\Ensure each OHT entry's \texttt{Neuron-Buffer} populated with its corresponding LSH bucket's neurons
\State \textbf{// Phase 1: build OHT from Request-Array, keyed by LSHbucketID}
\State $\mathsf{OHT}.\Call{build}{\mathsf{ReqArray}_\ell, \mathrm{numBin}, \mathrm{binCap}}$
\State \textbf{// Phase 2: linear scan over all LSH buckets}
\For{$b = 0$ to $\mathrm{numLshBucs}-1$}
  \For{$\mathsf{tier} \in \{1,2\}$}
    \For{each of the $\mathrm{binCap}$ entries $e$ in bin $\mathsf{OHT}.\Call{lookup}{tier, b}$}
      \State $\mathit{match} \gets \Call{OblCompare}{e.\texttt{LSHbucketID}, b}$
      \State $e.\texttt{Neuron-Buffer} \gets \Call{OblChoose}{\mathit{match}, \mathsf{LSH}_\ell[b], e.\texttt{Neuron-Buffer}}$
    \EndFor
  \EndFor
\EndFor
\State \textbf{// Phase 3: restore fixed-length Request-Array}
\State collect all OHT entries into flat array $A$
\State $\Call{OblSort}{A, \texttt{key}=(\texttt{isDummy}, \texttt{Batch-Rank})}$
\State \Return $A[0..T-1]$ \Comment{populated $\mathsf{ReqArray}_\ell$}
\end{algorithmic}
\end{algorithm}

\textbf{\textit{Obliviousness.}\xspace}
Algorithm~\ref{alg:fetcher-main} summarizes the Fetcher's fixed three-phase schedule: oblivious OHT construction, a complete scan of all LSH buckets and corresponding fixed-capacity OHT bins, and oblivious recovery of the \texttt{Request-Array}. 

As established by Chan~et al.~\cite{chan2017oblivious}, double obliviousness
requires that each key of the \texttt{OHT} be queried at most once. This
requirement is satisfied by our construction; since each padded LSH
bucket is scanned exactly once during the linear pass, no key is ever accessed
more than once. Finally, to convert the \texttt{OHT} back into a
\emph{populated} \texttt{Request-Array}, we again apply oblivious sorting.

\textbf{Optimization \#1 (O1): Deferred Dummy Payload.\xspace}
\label{opt:neuronf-fetcher}\label{opt:payload-elision}
Recall that during the feedforward phase, the \texttt{Request-Array} must be transformed into an \texttt{OHT} via an oblivious sort operation. 
At this point, however, the \texttt{Neuron-Buffer} field of each \texttt{Request-Array}'s triplet is occupied \emph{solely by dummy values}. 
Moving these placeholders through the sort operation, i.e., rewriting them alongside the meaningful fields of the triplet, introduces significant overhead (see~\S\ref{sec:opt-eval}) without any security benefit.

\textbf{\textit{Obliviousness.}\xspace}
\system optimizes performance based on the above observation by \emph{deferring dummy payload instantiation}; during construction, the oblivious sort operates only on the fields \texttt{LSHbucketID} and \texttt{Batch-Rank} of \texttt{Request-Array}. 
Once all request items have been placed into their target \texttt{OHT} bins, \system initializes every \texttt{Neuron-Buffer} with dummy neurons, thereby restoring the fixed-capacity layout required by the subsequent fetch stage. 
Both the oblivious sorting and the subsequent dummy-initialization process operate on a fixed number of entries in a fixed order.
Thus, this optimization preserves obliviousness. 
\system avoids redundant oblivious movement of dummy \texttt{Neuron-Buffer}s without affecting the obliviousness of feedforward \texttt{OHT} construction.

\textbf{Optimization \#2 (O2): Oblivious Parallelization.\xspace}
\label{opt:parallel}
Our experiments in Section~\ref{sec:opt-eval} show that the linear scan of the LSH table that populates the \texttt{OHT} is a performance bottleneck.
In this design, each LSH bucket encountered during the scan immediately triggers its corresponding Tier-1 and Tier-2 \texttt{OHT} lookups, forcing LSH bucket iteration and \texttt{OHT} updates to proceed in lockstep.
\system breaks this lockstep by deferring all \texttt{OHT} updates to a separate phase after the scan completes, enabling parallel updates across threads.

Naively parallelizing across LSH buckets would introduce write contention, since multiple buckets may want to write to the same \texttt{OHT} bin. \system eliminates this contention through a two-phase approach. In the first phase, \system iterates over the LSH buckets and initializes a \emph{scheduler} (in fact, one per \texttt{OHT} tier) that records which LSH buckets map to which bins (as determined by the \texttt{OHT} hash functions). 
In the second phase, once the LSH iteration completes, \system partitions the \texttt{OHT} bins into disjoint regions and assigns each region together with its associated LSH buckets to a dedicated thread, guaranteeing that no two threads write to the same bin.

\textbf{\textit{Obliviousness.}\xspace}
This optimization preserves obliviousness.
The scheduler is derived from the fixed sequence of LSH bucket indices and the \texttt{OHT} hash functions, 
so its construction has no data-dependent leakage.
The resulting partition assigns OHT-bin regions to threads without leakage, beyond the public parameters, and does not alter the set of buckets, bins, or entries processed.
Notice that this parallel Neuron Fetcher performs \emph{the same} LSH table iteration, updates \emph{the same} \texttt{OHT} bins, and applies \emph{the same} oblivious operations as the sequential design.
The only difference is the order of \texttt{OHT} updates, not which memory locations are accessed.

\textbf{Optimization \#3 (O3): OHT Layout Reuse.\xspace}
\label{opt:reuse-oht}
Both feedforward and backpropagation construct an \texttt{OHT} from scratch over a \texttt{Request\-Array}, each incurring multiple rounds of oblivious sorting that our experiments identify as another bottleneck (see~\S\ref{sec:opt-eval}).
Since the \texttt{OHT} construction in both phases operates on the same set of requests, the request-to-\texttt{OHT}-bin assignment in \texttt{OHT} established during feedforward remains valid for backpropagation and need not be recomputed.
To support reuse, during the feedforward \texttt{OHT} construction, \system records each request's \texttt{OHT} bin assignment and its position within that bin. 
During backpropagation, the aforementioned \texttt{OHT} assignment reduces the multi-round construction to a single oblivious sort that places each request in its predetermined location.
The scheduler for the parallel Neuron Fetcher is also retained, as the bin-to-\{LSH-bucket, thread\} mappings are unchanged.

\textbf{\textit{Obliviousness.}\xspace}
This optimization preserves obliviousness. 
The recorded layout information is consumed through the same fixed-size oblivious sorting procedure, so it does not introduce data-dependent branches or direct data-dependent addressing.
The oblivious sorting operation moves the same requests into each bin as the original construction, and the subsequent Neuron Fetcher traverses the same LSH buckets and \texttt{OHT} bins under the same scheduler.
The only difference is the direction of data movement; neurons \emph{transfer into} the \texttt{OHT} during the feedforward pass and \emph{move back out} during backpropagation.

\subsubsection{Neuron Updater}
\label{sec:neuron-updater}
With the \texttt{Request-Array} fully populated by Neuron Fetcher, the Neuron Updater performs feedforward and backpropagation on the fixed-size \texttt{Request-Array}.
Because requests are arranged in an FCFS order within \texttt{Request-Array}, the Neuron Updater processes them by linear scan rather than by a data-dependent traversal. 
Specifically, for each layer input $i$ in a batch and layer number $\ell$, \system preallocates two arrays to store each fetched neuron's network-level index (or NeuronID) and its activation value, namely $\mathsf{ActiveNodes}[i][\ell]$ and $\mathsf{ActiveVals}[i][\ell]$.
Both arrays have a fixed capacity of $R_\ell = \mathsf{lenSeq}_\ell \cdot \mathsf{PADSIZE}$, avoiding any data-dependent resizing or compaction during training.

During feedforward, the Neuron Updater linearly scans $\mathsf{ReqArray}_\ell$.
For each neuron in each request's \texttt{Neuron-Buffer}, Neuron Updater computes the neuron's activation from the previous layer's active neurons, applies the activation function, and records the NeuronID and activation into $\mathsf{ActiveNodes}[i][\ell]$ and $\mathsf{ActiveVals}[i][\ell]$.
During backpropagation, the updater traverses the same order---i.e., it computes deltas for each neuron in each request's \texttt{Neuron-Buffer}---, then runs nested loops whose bounds are determined solely by public parameters $P$ to propagate errors and accumulate per-edge weight updates.
See Appendix~\ref{app:neuron-updater} for more details.

\textbf{\textit{Obliviousness.}\xspace}
The Neuron Updater is oblivious because (1)~all loop bounds depend solely on the public parameters $P$ and (2)~it linearly scans and updates the copies of neurons within each individual \texttt{Request-Array} entry.
We use oblivious primitives for comparisons and conditional updates as opposed to data-dependent branches or memory accesses.
Both feedforward and backpropagation follow a fixed, data-independent access pattern.

\subsubsection{Oblivious Neuron Merge and LSH Rebuild.}
\label{sec:neuron-merge}

Before write-back, multiple requests may hold duplicate copies of the same
neuron with updates from different inputs in the batch. To group these
duplicate copies obliviously, \system first applies \texttt{OblSort}, an
\emph{oblivious sorting primitive} implemented with a fixed data-independent
compare and exchange pattern (e.g.,~bitonic sort~\cite{batcher1968sorting}).
\system resolves these duplicates by obliviously grouping requests by
\texttt{LSHbucketID} and then linearly scanning and consolidating to a single
neuron copy action within each group (retaining the redundant entries as fake
writes to preserve security). The merged results are then used to rebuild
the \texttt{OHT}, which initiates a write-back pass that traverses the same LSH
buckets and \texttt{OHT} bins as the feedforward fetch, with only the direction
of neuron transfer reversed.

\system also periodically rebuilds the LSH table by recomputing hash signatures for all neurons and relocating them in the new LSH table.
During initialization, \system adds \texttt{PADSIZE} dummy neurons per LSH bucket.
Periodic rebuilds then reuse this fixed dummy reservoir while recomputing LSH signatures. 
If the sum of dummies and real neurons exceeds the capacity, oblivious compaction sends neurons to an (LSH) \emph{overflow region}.
The same refresh routine is used during initialization, with fresh dummies added in.
The overflow region may hold real neurons, not just dummies, when more than \texttt{PADSIZE} hash to the same LSH bucket. 
See Appendix~\ref{app:obl-merge} and Appendix~\ref{app:lsh-refresh} for the merge and refresh routines, respectively, with the latter also covering the empirical choice of \texttt{PADSIZE}.

\textbf{Overflow Handling and Security.}\xspace

The overflow region has a fixed size (known via public parameters) and is present regardless of whether it contains real neurons or only dummy entries. 
During every LSH initialization or refresh, \system obliviously sorts all the neurons and executes a linear scan to obliviously mark the overflow neurons\footnote{Notice, though, that no neurons are permanently evicted from the table; thus, they have the chance to be considered after a rebuild.}. As a next step, \system fills up every LSH bucket with real/\-dummy neurons (or both) and moves a fixed number of overflowing neurons (known via public parameters) to the overflow region (see Alg~\ref{alg:lsh-refresh-obl}, Alg~\ref{alg:enforce-capacity} and \S\ref{app:lsh-refresh}). 
During training, \system accesses only the fixed-size main region (LSH table) and ignores the overflow content. 
The overflowed neurons become available again during the next scheduled refresh. 
Overall, the overflow phenomenon may affect which neurons are updated between refreshes, but the handling of moving to and out of the overflow region is done obliviously. 
Thus, security is unaffected by actions in the overflow region.

Appendix~\ref{app:padsize-accuracy} empirically evaluates the effect of \texttt{PADSIZE}  on the training accuracy of the neural network.

\subsection{Implementation}
\label{sec:implementation}

We implement \system in approximately 5,000 LoC in C++. The entire
system runs in user space and requires no modifications to the OS kernel, TEE
runtime, or system libraries. The implementation covers the full oblivious
training pipeline, including MP-WTA hashing, the \texttt{Request-Array}, the
OHT-based Neuron Fetcher, the fixed-shape Neuron Updater, and the oblivious
merge and LSH refresh routines. We reuse only the network initialization and
data-loading routines from SLIDE's open-source codebase~\cite{chen2020slide}.

For low-level oblivious primitives, we follow the design principles of
Snoopy~\cite{dauterman2021snoopy}, which provides data-oblivious scalar and
byte operations for TEEs. We extend this approach with oblivious routines for
\system's specific composite objects, such as request entries and neuron
buffers. These routines are implemented within \system's own codebase rather
than by modifying Snoopy's library. Refer to Appendix~\ref{app:obl-primitive}
for additional details.

\section{Security Guarantees}
\label{sec:security-proof}

Recall that $P$ denotes the public parameter set (\S\ref{sec:components}, \S\ref{sec:public-private}).
We model \system as a \emph{two-party computation} between a User and an OS/\-Hypervisor adversary $\mathcal{H}$. 
The User holds input $(P,D)$, where $D$ is the private training data; the Hypervisor holds input $P$.

\textbf{Real World.\xspace} The parties execute protocol $\pi_{\system}$ in the TD. 
The User sends $(P,D)$ to the TD, which runs \system and returns the trained model $M$ to the User.
The Hypervisor is passive.
    Its view consists of the public input $P$ together with the memory-access trace $\tau$ observed during execution.
We write:
$
\mathsf{view}_{\mathsf H}^{\pi_{\system}}(P,D)=\{P,\tau\}.
$

\textbf{Ideal World.\xspace} The parties interact only through the functionality $\F_{\mathsf{train}}\bigl[(P, D),\; P\bigr] = \bigl[M,\; \bot\bigr]$. 
On input $((P,D),P)$, the functionality $\F_{\mathsf{train}}$ returns the trained model $M$ to the User and returns $\bot$ to the Hypervisor.
We say that $\pi_{\system}$ securely realizes $\F_{\mathsf{train}}$ if there exists a PPT simulator $\Sim$ such that, for all $P$ and all $D$ consistent with $P$, the joint output distribution in the ideal world is computationally indistinguishable from that in the real world:
\[
    \bigl\{\mathbf{Ideal}_{\F_{\mathsf{train}},\Sim}(\kappa;\, (P,D),\, P)\bigr\}
    \;\approx\;
    \bigl\{\mathbf{Real}_{\pi_{\system},\mathcal{H}}(\kappa;\, (P,D),\, P)\bigr\}
\]

In our setting, the Hypervisor output is $\bot$ in both worlds, and the security argument focuses on the information leaked through the Hypervisor's view, namely $\tau$.
The proof shows that the real Hypervisor view can be simulated from the public parameters $P$ alone---i.e., $\tau$ does not reveal information beyond the leakage explicitly included in $P$.
Equivalently, \system securely realizes $\F_{\mathsf{train}}$ if there exists a PPT simulator $\Sim$ that, given only $P$, produces a simulated Hypervisor view such that:
\[
\{\Sim(P)\}\approx_c \{\mathsf{view}_{\mathsf H}^{\pi_{\system}}(P,D)\}.
\]

We prove Theorem~\ref{thm:security} in the \emph{hybrid model}~\cite{canetti2000security}, where self-contained primitives (i.e.,~oblivious sort, comparison, and hash table; see Appendix~\ref{app:obl-primitive}) are abstracted as ideal functionalities. 
Specifically, we prove that \system is secure given \emph{access to these functionalities}; each is then replaced by its corresponding protocol, which implements the functionality and whose security follows from prior work.

\begin{theorem}
\label{thm:security}[Informal]
Assume that each underlying oblivious primitive securely realizes its corresponding functionality. 
\system then securely realizes $\F_{\mathsf{train}}$ against the Hypervisor in the hybrid model~\cite{canetti2000security}.
\end{theorem}

\begin{proof}[Proof Sketch]
The proof is modular. 
We first analyze \system in a primitive hybrid model where the low-level oblivious primitives are replaced by functionalities.
In this model, each main component of \system, namely Neuron Requester, Neuron Fetcher, Neuron Updater, and the LSH initialization and refresh routines, is simulatable from the public parameters $P$ alone.
By the security of the underlying primitives, replacing these functionalities with their real protocols preserves security.
Hence, each real component securely realizes its corresponding functionality.  
Modular composition then gives a secure batch step.
Sequential composition over the public training schedule, including initialization, batch steps, and refresh events, yields the security of $\pi_{\system}$.
The optimized Fetcher schedules O1--O3 preserve the same leakage claim because they add only public object-size, parallelization, and layout-reuse metadata.
The full proof appears in Appendix~\ref{app:proof}.
\end{proof}

\section{Evaluation}
\label{sec:evaluation}

We evaluate \system to: (1)~experimentally demonstrate the benefits of our
approach and quantify the overall speedup over a baseline SLIDE+ORAM, which
runs SLIDE~\cite{chen2020slide} ``as is'' on top of a generic ORAM
compiler~\cite{stefanov2018path} (\S\ref{sec:eval-performance});
(2)~attribute performance gains to individual optimizations (O1--O3,
\S\ref{sec:neuron-fetcher}) via progressive ablations and runtime breakdowns
(\S\ref{sec:opt-eval}); and (3)~measure the accuracy of \texttt{MP-WTA}
compared to standard LSH (\S\ref{sec:acc-eval}). The neural networks we
train have a wide layer with 30K to 325K neurons.

\textbf{Testbed.\xspace}
All experiments were conducted on one of two servers. For the performance and
ablation experiments~(1)--(2), we use a Google Cloud Confidential C3 instance
equipped with 2$\times$22-core/44-thread CPUs (Intel Sapphire Rapids) and
352~GB of RAM, running Ubuntu 24.04 LTS as a confidential VM image with Intel
TDX as the TEE, and supporting AVX, AVX2, and AVX-512 instructions. For the
accuracy evaluation~(3), we use a bare-metal x86-64 server (64-core AMD EPYC
7763 Zen 3) without a TEE, as accuracy evaluation does not require secure
execution. \system is implemented in C++ (see~\S\ref{sec:implementation}) and
compiled with GCC v13.3, with OpenMP enabled for multi-threaded parallelism.
In addition, we enable AVX2/\-FMA support when available. 
Our oblivious primitives use inline assembly that is preserved at any
optimization level. We compile without optimization flags (\texttt{-O0} under
GCC) as an additional safeguard against compiler-introduced data-dependent
branches for all tests.

\begin{table}[t]
  \centering
  \caption{Overview of Evaluation Datasets.}
  \label{tab:performance-datasets}
  \begin{tabular}{lrrrr}
    \toprule
    Dataset & Feature‑dim. & \# Labels & \#Train & \#Test \\
    \midrule

    Wiki10‑31K & 101,938 & 30,938 & 14,146 & 6,616 \\
    Amz-Sub  & 135,909 & 92,814 & 42,438 & 14,146 \\
    Wiki‑325K  & 1,617,899 & 325,056 & 1,778,351 & 587,084 \\
    
\cmidrule(lr){1-5}
    Amz‑670K   & 135,909 & 670,091 & 490,449 & 153,025 \\

    \bottomrule
  \end{tabular}
\end{table}

\textbf{Datasets.\xspace}
We use datasets from the Extreme Multi-Label Classification (XC)
repository~\cite{Bhatia16}, as established in prior work~\cite{chen2020slide}.
For large-scale evaluation, we use WikiLSHTC-325K (Wiki-325K), a
Wikipedia-derived dataset with 325K label categories. We also include
Wiki10-31K, our smallest dataset, to evaluate \system under a more constrained
label space. Finally, we construct Amazon-Sub, a curated subset of Amazon-670K,
with label cardinality falling between that of Wiki10-31K and Wiki-325K,
serving as a medium-scale dataset. Together, the three datasets enable the
evaluation of \system across small, medium, and large label regimes that require a wide layer for the classification. Key
statistics are in Table~\ref{tab:performance-datasets}.

\subsection{\rpname vs. SLIDE+ORAM}
\label{sec:eval-performance}

\textbf{Baseline.\xspace}
We compare \system against \texttt{SLIDE+ORAM}, which uses Path ORAM~\cite{stefanov2018path}
to store and access all neurons and LSH buckets used by SLIDE.
We adopt Path ORAM as our baseline because it is an extremely simple tree-based ORAM
construction, and has been widely used in secure-processor and
enclave systems both in academia~\cite{stefanov2018path,ren2013design,sasy2018zerotrace,tinoco2023enigmap} and industry~\cite{gconnell2022oram}.
In this baseline, neurons and LSH buckets are
stored in a uniform representation on the ORAM server. Whenever SLIDE needs to
read or write a neuron or an LSH bucket, it issues a corresponding ORAM access
request through the ORAM client.

As discussed in Section~\ref{sec:double-oblivious}, the ORAM client's memory
accesses to the stash may still be observable via side-channels, even when
running inside a TEE. To avoid the leakage from the stash accesses, the \texttt{SLIDE+ORAM} baseline performs a full
linear scan over the stash on every access, avoiding data-dependent access
patterns---a standard technique proposed by
Mishra~et~al.~\cite{mishra2018oblix}. Additionally, the baseline pads the
number of neuron accesses per input to a fixed upper bound (i.e.,~the LSH
bucket size) to conceal the distribution of retrieved neurons.

We emphasize that this baseline is \emph{not} doubly-oblivious (and thus has
the potential to be faster), yet it represents a state-of-the-practice
implementation that does not consider \emph{eliminating} access-pattern
leakage.  Specifically, while it provides obliviousness for ORAM-backed storage
accesses, it does not make the SLIDE computation itself oblivious: after
retrieving neurons or buckets via the ORAM client, SLIDE's internal control
flow and memory accesses still depend on the input, leaking information about
the training data.

The baseline uses the same parameters as \system (i.e.,~neural network
architecture, LSH hash function count, and batch size), except that the number
of LSH tables is 50, whereas \system uses the newly proposed \texttt{MP-WTA} with only 1 hash
table. All components---the neural network, Path ORAM, and SLIDE---run inside
TDX.

\textbf{Parameters.\xspace}
For all three datasets, we adopt the fully connected neural network with one wide hidden layer, similarly to the work of Chen~et~al.~\cite{chen2020slide, daghaghi2021tale}. 
For the chosen learning task, the wide layer is the output layer (i.e., where LSH is applied), which contains 30K, 90K, and 325K neurons for the Wiki-10, Amazon-Sub, and Wiki-325 datasets, respectively.
As for the hidden layer, we chose a width of 128 for Amazon-Sub and Wiki-325K, matching the configuration used in prior work~\cite{chen2020slide, daghaghi2021tale}.
Since Wiki10-31K is an additional dataset not evaluated in that work, we use a slightly larger hidden layer (256) to obtain a reasonable model capacity on this smaller label space while keeping the overall architecture consistent.

We fix the batch size to 32 across all datasets. This is dictated by our testbed
constraints: Wiki-325K is our largest workload, and larger batch sizes exceed
the available RAM when accounting for model state and associated data
structures. Fixing the batch size also keeps comparisons consistent across
datasets.

For LSH, we use WTA for the \texttt{SLIDE+ORAM} baseline and \texttt{MP-WTA} for \system, applied
to the output layer only, which is the primary computational
bottleneck~\cite{chen2020slide}. We select $K$ via accuracy-driven tuning over
a small set of candidates, yielding $K = 3$, $4$, and $5$ for Wiki10-31K,
Amazon-Sub, and Wiki-325K, respectively, and use these values throughout all
performance experiments.

To keep the Google Cloud infrastructure costs manageable, we benchmark a fixed window of 1,280 training
inputs per dataset and configuration, reporting mean per-batch runtime over
this window. We estimate end-to-end training time by scaling this mean by the
total number of batches in each dataset's training schedule, determined by
dataset size and epoch count. We evaluate \system with 8, 22, and 44 threads,
corresponding to real Google Cloud confidential computing configurations with
varying core counts.

\textbf{Performance Results.\xspace}
Table~\ref{tab:eval-runtime-32} compares the mean per-batch runtime of \system
against the baseline (\texttt{SLIDE+ORAM}) at batch size 32. Across all
datasets, \system achieves 1--3 orders-of-magnitude speedup, with
gains increasing with additional CPU resources, reflecting effective
parallelization of the dominant components. The baseline, by contrast, relies
on a standard Path ORAM implementation with serialized accesses, so additional
vCPUs offer limited benefit; we hence report a single baseline measurement
per dataset.

\begin{table}[t]
\centering
\caption{Per-Batch Runtime (avg. over 1,280 inputs, batch size = 32).
Speedup (SU) is relative to \texttt{SLIDE+ORAM} on the same dataset. Epoch (Ep.) time
is extrapolated from mean batch time and epoch batch count. The baseline uses a
sequential Path ORAM that does not benefit from multiple vCPUs.}
\label{tab:eval-runtime-32}

\begin{tabular}{cccccc}
\toprule
Dataset & Method & \#Thd & Batch(s/\-h) & SU & Ep.(h) \\
\midrule

\multirow{4}{*}{Wiki10}
& \texttt{SLIDE+ORAM}               & 1   & 2.6\,h   & 1.0$\times$ & 1,149   \\
\cmidrule(lr){2-6}
& \multirow{3}{*}{\system} & 8  & 681.3\,s   &  13.7$\times$ & 83.6\\
&                          & 22 & 338.8\,s   &  27.6$\times$ & 41.6\\
&                          & 44 & 248.3\,s   &  37$\times$   & 30.5\\
\midrule

\multirow{4}{*}{Amz-Sub}
& \texttt{SLIDE+ORAM}               & 1   & 16\,h   & 1.0$\times$  &   21,216\\
\cmidrule(lr){2-6}
& \multirow{3}{*}{\system} & 8  & 2,085.5\,s & 27$\times$  &  768.2\\
&                          & 22 & 1,007.8\,s  & 57.1$\times$ &  371.2 \\
&                          & 44 & 695.7\,s  & 82.7$\times$  & 256.5\\
\midrule

\multirow{4}{*}{Wiki325}
& \texttt{SLIDE+ORAM}               & 1   &  208\,h & 1.0 $\times$   & 1.15x$10^7$\\
\cmidrule(lr){2-6}
& \multirow{3}{*}{\system} & 8  & 4,717.8\,s & 157.2$\times$ &  72,828\\
&                          & 22 & 2,255.4\,s & 328.8$\times$ &  34,816  \\
&                          & 44 & 1,578.3\,s & 469.9$\times$ &  24,364\\

\bottomrule
\end{tabular}
\end{table}

On Wiki10-31K, the baseline takes 2.6\,h, while \system reduces runtime to
681\,s, 338\,s, and 248\,s with 8, 22, and 44 vCPUs, corresponding to speedups
of 13.7$\times$, 27.6$\times$, and 37$\times$, respectively. On Amazon-Sub,
\system cuts the baseline runtime of 16\,h down to 2,085\,s, 1,007\,s, and
695\,s (8, 22, and 44 vCPUs), achieving speedups of 27$\times$--82.7$\times$.
Although scaling is sublinear due to residual serial work and contention on
shared, memory-intensive data structures, the absolute reduction in wall-clock
time remains substantial across all core counts. On the larger Wiki-325K
workload, the baseline requires 208\,h, whereas \system completes the same run
in 4,717\,s, 2,255\,s, and 1,578\,s with 8, 22, and 44 vCPUs, yielding speedups
of 157.2$\times$, 328.8$\times$, and 469.9$\times$.

\emph{These results demonstrate that \system makes training more feasible for
substantially larger XC problems~\cite{Bhatia16} where a Path ORAM-based approach becomes
prohibitively slow.}

\begin{figure}[tb]
  \centering
  \begin{subfigure}[b]{0.49\columnwidth}
    \centering
    \includegraphics[width=\columnwidth]{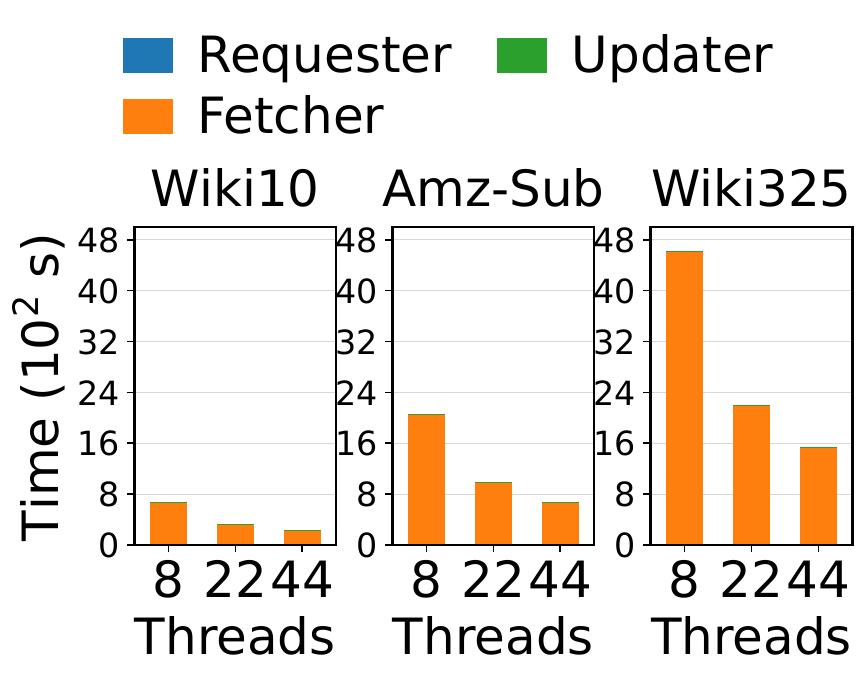}
    \caption{Total}
    \label{fig:32-sub-a}
  \end{subfigure}\hfill
  \begin{subfigure}[b]{0.49\columnwidth}
    \centering
    \includegraphics[width=\columnwidth]{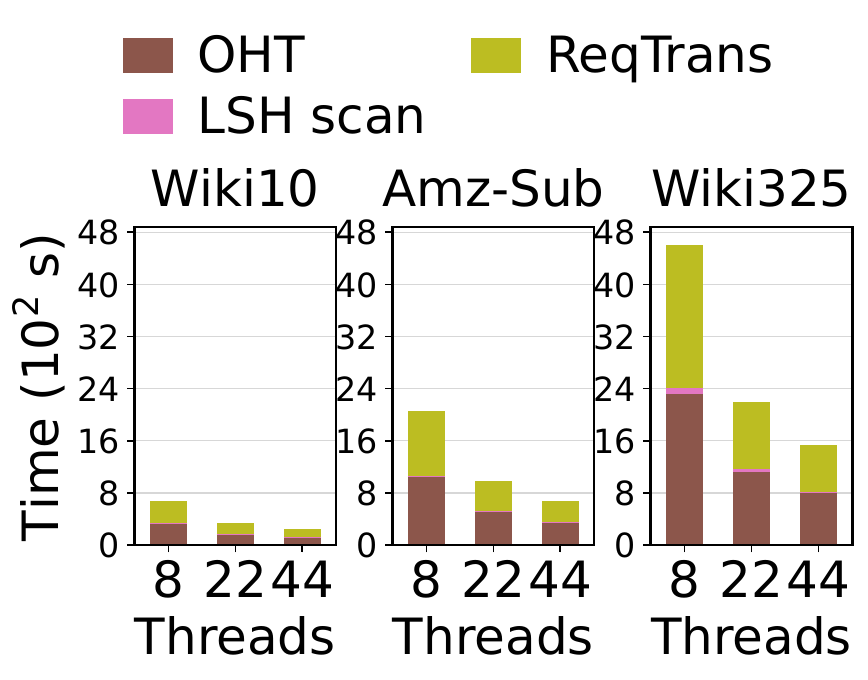}
    \caption{Fetcher}
    \label{fig:32-sub-b}
  \end{subfigure}
  \caption{Breakdown of Runtime (Total and Fetcher) across datasets and thread counts with batch size = 32.}
  \label{fig:breakdown-bs32}
\end{figure}

Following Table~\ref{tab:eval-runtime-32}, we break down the mean per-batch
runtime into major components to identify where time is spent.
Figure~\ref{fig:breakdown-bs32} reports this component-level breakdown under
the same measurement window. As shown in Figure~\ref{fig:32-sub-a}, the Neuron
Fetcher dominates runtime, while the Neuron Requester, Neuron Updater, and
remaining operations contribute negligible overhead, barely registering in the
figure. 
Motivated by this, Figure~\ref{fig:32-sub-b} further decomposes the Neuron
Fetcher into its internal stages. OHT construction and OHT-to-request
conversion each account for nearly half of the Fetcher time, as both are
dominated by computationally intensive oblivious sorting. The LSH scanning
stage, by contrast, contributes only a small fraction of the Fetcher runtime.
Overall, this breakdown pinpoints the post-optimization bottlenecks and
motivates the progressive ablation analysis in Section~\ref{sec:opt-eval}.

\begin{table}[tb]
\centering
\caption{Runtime on Wiki10-31K and Amazon-Sub (batch size = 64). 
Wiki-325K is omitted due to RAM constraints.}
\label{tab:eval-runtime-64}
\begin{tabular}{cccccc}
\toprule
Dataset & Method & \#Thd & Batch(s/\-h) & SU & Ep.(h) \\
\midrule

\multirow{4}{*}{Wiki10}
& \texttt{SLIDE+ORAM}               & 1   & 5.2\,h   & 1.0$\times$   & 1,149 \\
\cmidrule(lr){2-6}
& \multirow{3}{*}{\system} & 8  & 1,655.8\,s  &  11.3$\times$ & 101.6 \\
&                          & 22 & 795.9\,s   &  23.5$\times$ & 48.8 \\
&                          & 44 & 565.4\,s   &  33.1$\times$ & 34.7  \\
\midrule

\multirow{4}{*}{Amz-Sub}
& \texttt{SLIDE+ORAM}               & 1   & 32\,h   & 1.0$\times$  &  21,216 \\
\cmidrule(lr){2-6}
& \multirow{3}{*}{\system} & 8  & 11,496.9\,s  & 10$\times$  &  2,117.3\\
&                          & 22 & 5,185.5\,s  & 22.2$\times$  &  954.9\\
&                          & 44 & 3,514.4\,s  & 32.7$\times$ &  647.2\\

\bottomrule
\end{tabular}
\end{table}

\textbf{Batch-size Sensitivity.\xspace}
We additionally evaluate a larger batch size of 64 on Wiki10-31K and Amazon-Sub
(Table~\ref{tab:eval-runtime-64}); Wiki-325K is omitted due to memory
constraints on our evaluation testbed. While \system continues to deliver
substantial speedups (10$\times$--33.1$\times$), the gains are smaller than at
batch size 32. As shown in Figure~\ref{fig:breakdown-bs64}, this trend is
consistent with the breakdown in Figure~\ref{fig:breakdown-bs32}: the Neuron
Fetcher is dominated by sort-heavy stages (OHT construction and OHT-to-request
conversion) whose cost grows superlinearly with the number of items processed
and becomes increasingly memory-bound at larger batch sizes. Hence, increasing
the batch size amplifies the overhead of these stages and reduces parallel
efficiency.

\begin{figure}[tb]
  \centering
  \begin{subfigure}[b]{0.49\columnwidth}
    \centering
    \includegraphics[width=\columnwidth]{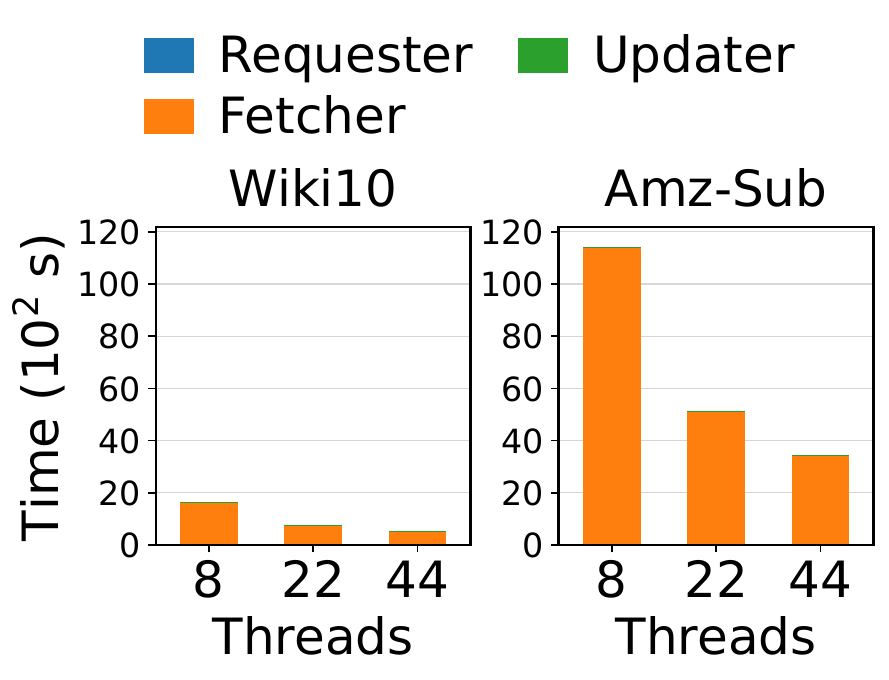}
    \caption{Total}
    \label{fig:64-sub-a}
  \end{subfigure}\hfill
  \begin{subfigure}[b]{0.49\columnwidth}
    \centering
    \includegraphics[width=\columnwidth]{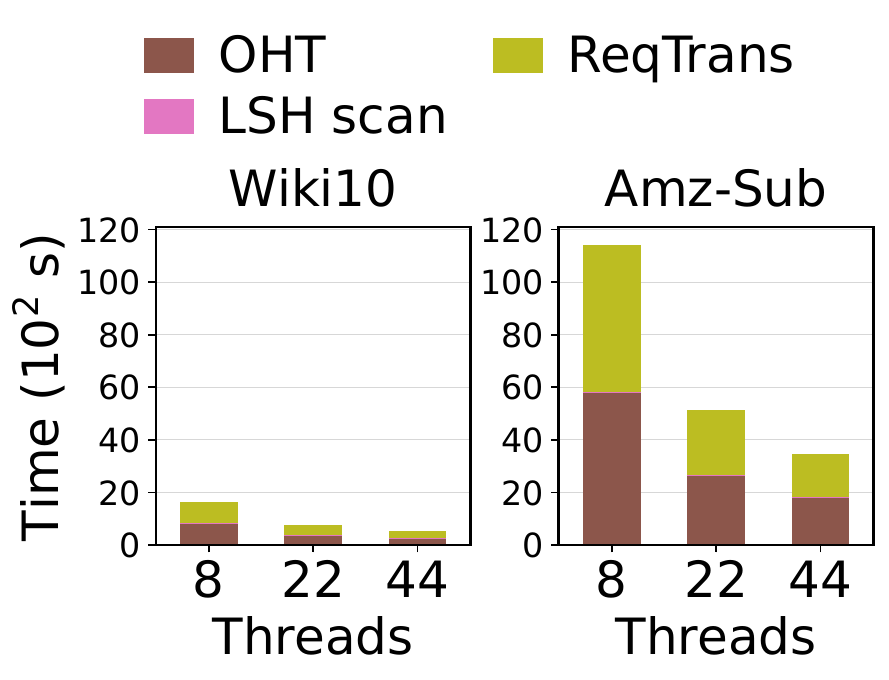}
    \caption{Fetcher}
    \label{fig:64-sub-b}
  \end{subfigure}
  \caption{Breakdown of Runtime (Total and Fetcher) across datasets and thread counts with batch size = 64.}
  \label{fig:breakdown-bs64}
\end{figure}

Training with a large effective batch size ($B_{\text{eff}}$) is nonetheless
achievable using small micro-batches via gradient
accumulation~\cite{goyal2017accurate, huang2019gpipe, narayanan2021memory}: $B_{\text{eff}} =
B_{\text{micro}} \times \textit{accum\_steps}$. Based on our evaluation,
\system can fix $B_{\text{micro}} = 32$ to maximize hardware efficiency and
avoid the superlinear memory and computational overheads associated with larger
micro-batches. Although we do not
explicitly evaluate gradient accumulation in this work, it provides a
straightforward mechanism to scale to the effective batch sizes required for
convergence without increasing per-step system cost.

\textbf{Memory Scaling.\xspace}
The OHT is the dominant batch-dependent component in peak memory, and it grows \emph{quadratically} with the batch size $B$. 
Specifically, the OHT number of bins has a linear dependence on the total number of requests, i.e., $\mathcal{O}(B\cdot\texttt{lenSeq})$, while the OHT per-bin capacity grows approximately linearly with $\mathcal{O}(B\cdot\texttt{PADSIZE})$. 
Given that the OHT memory footprint is bins $\times$ per-bin capacity, for a fixed model configuration, $M_{\mathsf{peak}}\approx M_{\mathsf{fixed}}+\mathcal{O}(B^2\cdot\texttt{lenSeq}\cdot\texttt{PADSIZE})$.
Appendix~\ref{app:memory} reports the measured peak memory and provides a detailed analysis.

\subsection{Optimization Ablations}
\label{sec:opt-eval}

We evaluate the performance impact of each optimization described
in Section~\ref{opt:neuronf-fetcher} using an ablation study. Unless stated
otherwise, \emph{All-Opt} enables all optimizations, and each ablated variant
disables exactly one optimization while keeping the rest enabled.

\begin{figure}[tb]
  \centering
  \begin{subfigure}[b]{0.49\columnwidth}
    \centering
    \includegraphics[width=\columnwidth]{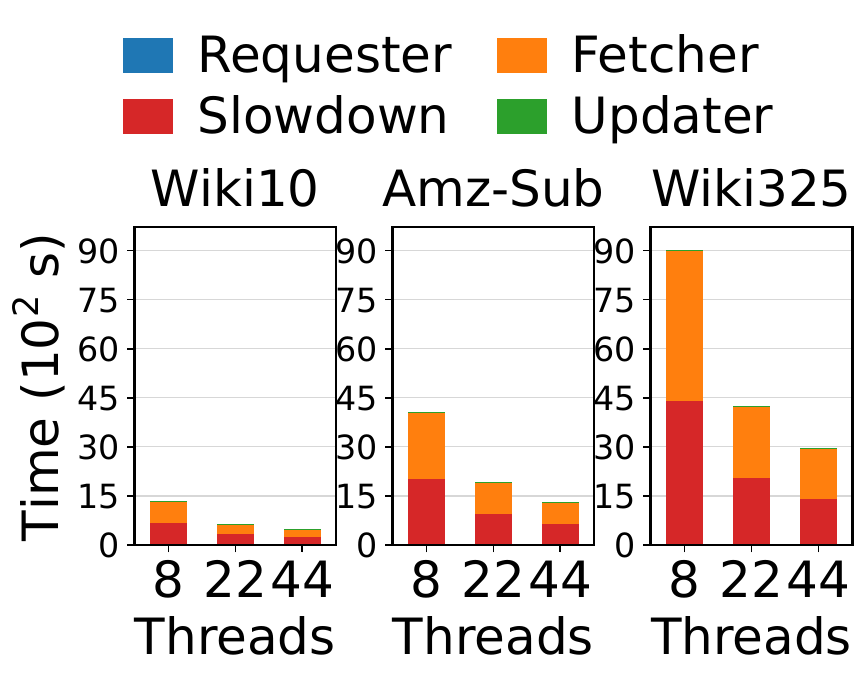}
    \caption{Total}
    \label{fig:opt-pad-32-sub-a}
  \end{subfigure}\hfill
  \begin{subfigure}[b]{0.49\columnwidth}
    \centering
    \includegraphics[width=\columnwidth]{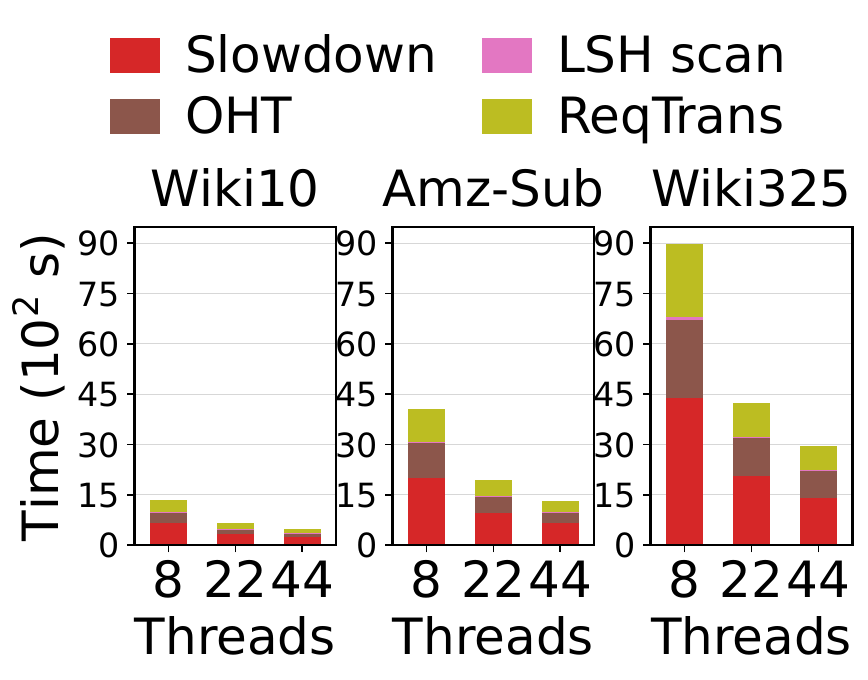}
    \caption{Fetcher}
    \label{fig:opt-pad-32-sub-b}
  \end{subfigure}
  \caption{Breakdown of Runtime (Total and Fetcher) without O1 (Deferred Dummy Payload) across datasets and thread counts with batch size = 32. (\emph{NO-Opt-Payload-Deferral})}
  \label{fig:opt-pad-breakdown-bs32}
\end{figure}

\textbf{Impact of O1 (Deferred Dummy Payload).\xspace}
\emph{NO-Opt-Payload\-Deferral} disables optimization O1 (i.e.,~deferred dummy
payload). Figure~\ref{fig:opt-pad-breakdown-bs32} shows the runtime breakdown,
with the red segment representing the slowdown relative to \emph{All-Opt}.
Materializing dummy neuron payloads eagerly during \texttt{OHT} construction
doubles the runtime, with the overhead concentrated in the \texttt{OHT}
building stage.

\begin{figure}[tb]
  \centering
  \begin{subfigure}[b]{0.49\columnwidth}
    \centering
    \includegraphics[width=\columnwidth]{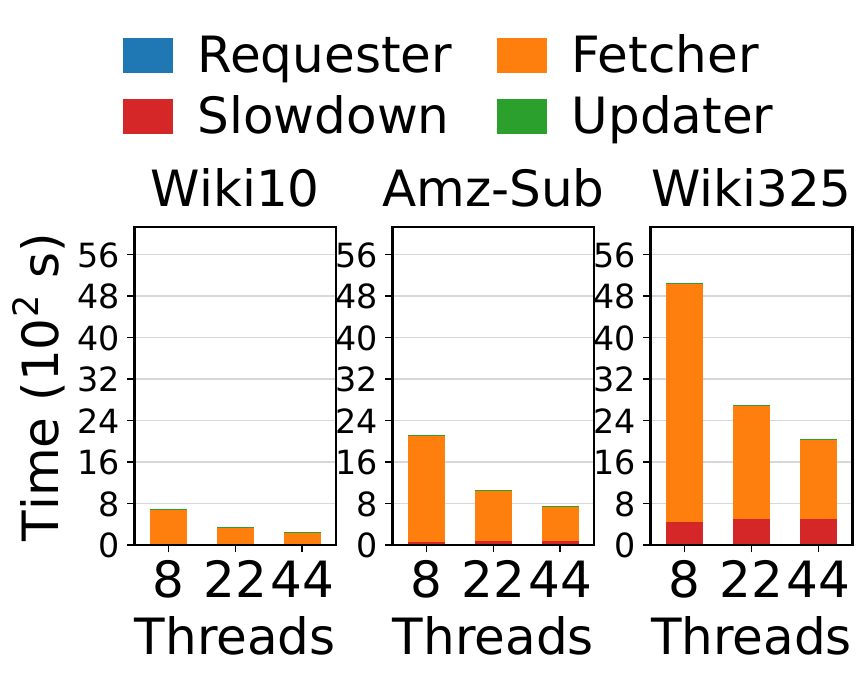}
    \caption{Total}
    \label{fig:opt-parallel-32-sub-a}
  \end{subfigure}\hfill
  \begin{subfigure}[b]{0.49\columnwidth}
    \centering
    \includegraphics[width=\columnwidth]{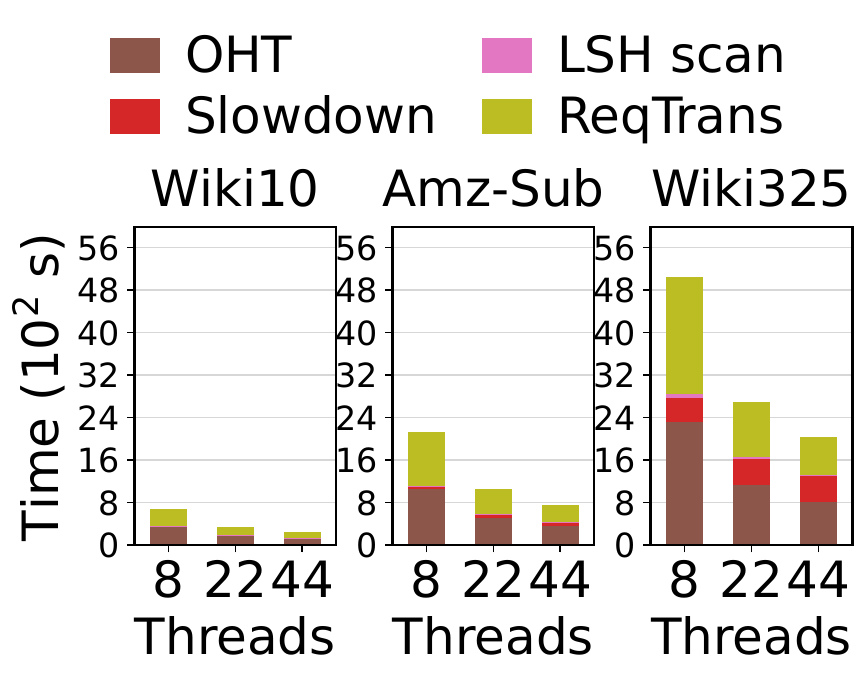}
    \caption{Fetcher}
    \label{fig:opt-parallel-32-sub-b}
  \end{subfigure}
  \caption{Breakdown of Runtime (Total and Fetcher) without O2 (Oblivious Parallelization) across datasets and thread counts with batch size = 32. (\emph{NO-Opt-ParallelScan})}
  \label{fig:opt-parallel-breakdown-bs32}
\end{figure}

\textbf{Impact of O2 (Oblivious Parallelization).\xspace}
\emph{NO-Opt-ParallelScan} disables optimization O2 (i.e.,~oblivious
parallelization), reverting the Neuron Fetcher to a sequential LSH linear scan.
Figure~\ref{fig:opt-parallel-breakdown-bs32} shows the runtime breakdown, with
the red segment indicating the overhead introduced by disabling this
optimization. As the dataset size increases, both the LSH tables and
intermediate \texttt{OHT} grow, amplifying the cost of sequential scanning.
Furthermore, as more threads are provisioned, the sequential scan becomes
the dominant fraction of end-to-end runtime, since the remaining
components parallelize well while the scan itself does not.

\begin{figure}[tb]
  \centering
  \begin{subfigure}[b]{0.49\columnwidth}
    \centering
    \includegraphics[width=\columnwidth]{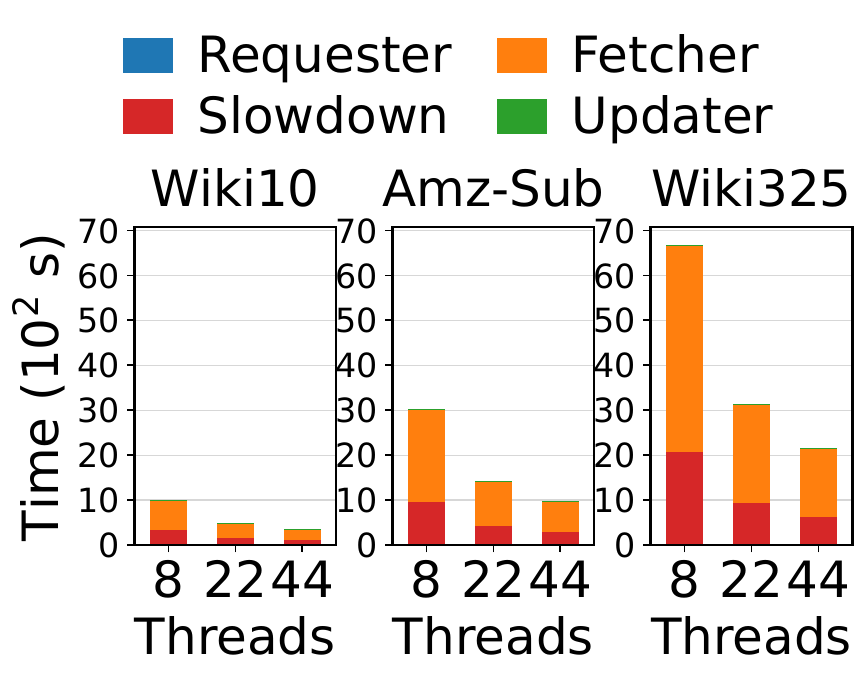}
    \caption{Total}
    \label{fig:opt-reuse-32-sub-a}
  \end{subfigure}\hfill
  \begin{subfigure}[b]{0.49\columnwidth}
    \centering
    \includegraphics[width=\columnwidth]{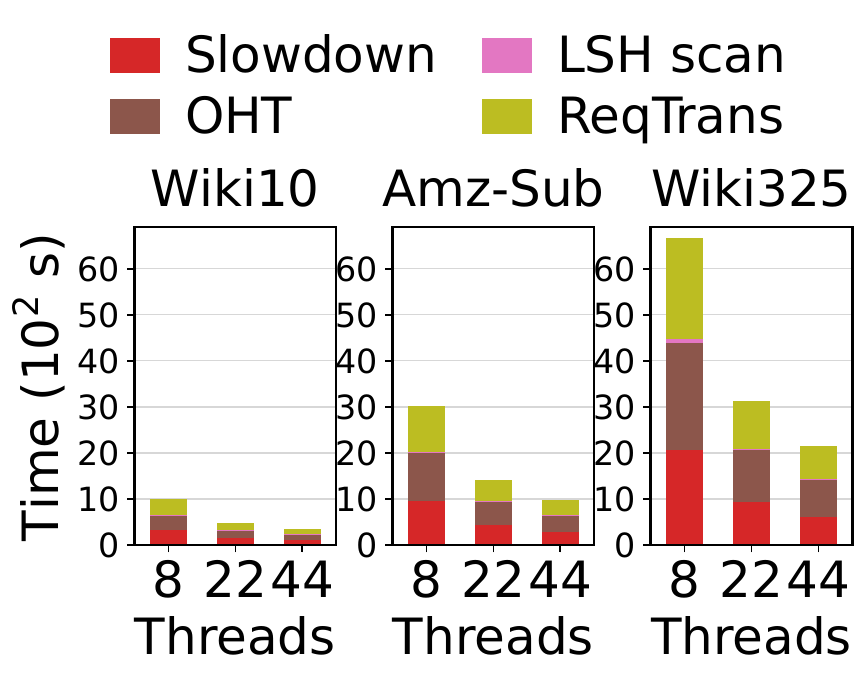}
    \caption{Fetcher}
    \label{fig:opt-reuse-32-sub-b}
  \end{subfigure}
  \caption{Breakdown of Runtime (Total and Fetcher) without O3 (OHT Layout Reuse) across datasets and thread counts with batch size = 32. (\emph{NO-Opt-OHTReuse})}
  \label{fig:opt-reuse-breakdown-bs32}
\end{figure}

\textbf{Impact of O3 (OHT Layout Reuse).\xspace}
\emph{NO-Opt-OHTReuse} disables optimization O3 (i.e.,~OHT layout reuse),
rebuilding the backpropagation \texttt{OHT} from scratch using the standard
construction procedure. Figure~\ref{fig:opt-reuse-breakdown-bs32} shows the
runtime breakdown, with the red segment indicating the overhead relative to
\emph{All-Opt}. Disabling layout reuse incurs additional oblivious sorting and
placement overhead, nearly doubling the \texttt{OHT} build time.

\subsection{Accuracy of MP-WTA}
\label{sec:acc-eval}

We evaluate predictive accuracy at the algorithm level, comparing
\texttt{MP-WTA} against standard LSH under identical training hyperparameters.
Accordingly, we report accuracy using the baseline \texttt{SLIDE} implementation and its
\texttt{MP-WTA} variant without \system's secure-execution mechanisms
(TEE/\-oblivious primitives), as our system transformations are
semantics-preserving and affect only memory-access patterns. Thus, any accuracy
differences observed here are attributable solely to the hashing scheme.
Models range from 10M to over 100M parameters,
with computation dominated by the final
classification layer, isolating the effect of replacing standard LSH with
\texttt{MP-WTA} on predictive performance across XC tasks.

\textbf{Setup.\xspace}
We isolate the effect of hashing by training two otherwise-identical models 
per dataset, sharing the same architecture, optimizer, and training schedule, 
and differing only in the final-layer hash function (WTA~vs.~MP-WTA). Unless 
stated otherwise, we use the same model configuration and hashing settings as 
in the performance evaluation (\S\ref{sec:eval-performance}), with batch size 
as the only exception. Hashing is applied \emph{only} to the final 
classification layer.

\textbf{Hashing.\xspace}
We adopt the dataset-specific $K$ from prior SLIDE work for Amazon-670K and 
Wiki-325K~\cite{daghaghi2021tale}, and use the same values in our performance evaluation. 
For Wiki10-31K and Amazon-Sub, we select $K$ via a lightweight sweep and fix it
for both WTA and MP-WTA. Concretely, $K = 3, 4, 5, 6$ for Wiki10-31K, Amazon-Sub, Wiki-325K, and Amazon-670K, respectively. We cap WTA at $L = 50$
tables and set MP-WTA to a single table ($L = 1$), reflecting our
compute/\-memory budget and testing whether MP-WTA retains accuracy with far
fewer tables. Unless stated otherwise, MP-WTA uses multi-probe parameters $r =
3$, $N = 3$.

\textbf{Hyperparameters.\xspace}
Each dataset uses a fully connected network with a single hidden layer, 
trained with Adam (learning rate $10^{-4}$)~\cite{kingma2014adam}.
Batch size is 128 for Amazon-670K, Wiki-325K, and Amazon-Sub, and 32 for
Wiki10-31K.

\begin{figure}[t]
  \centering

  \begin{subfigure}[b]{0.49\linewidth}
    \centering
    \includegraphics[width=\linewidth]{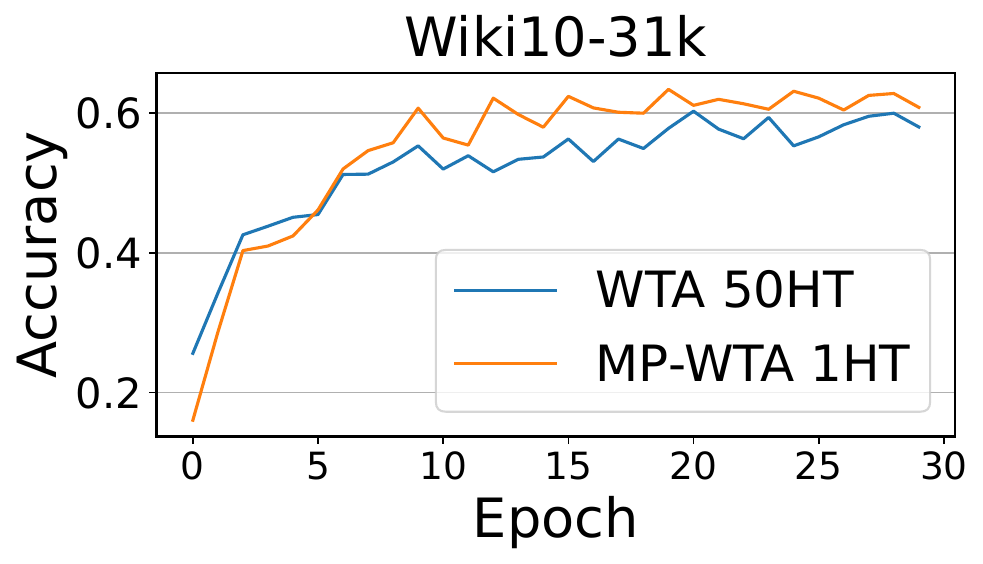}
    \caption{Wiki10‑31K}
    \label{fig:wiki10}
  \end{subfigure}
  \hfill
  \begin{subfigure}[b]{0.49\linewidth}
    \centering
    \includegraphics[width=\linewidth]{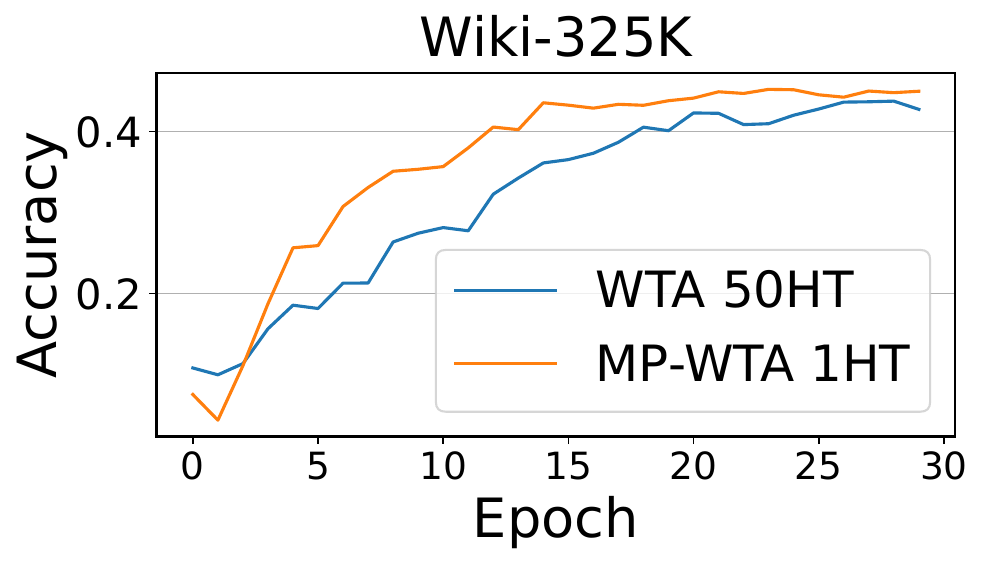}
    \caption{Wiki-325K}
    \label{fig:wiki325}
  \end{subfigure}

  \begin{subfigure}[b]{0.49\linewidth}
    \centering
    \includegraphics[width=\linewidth]{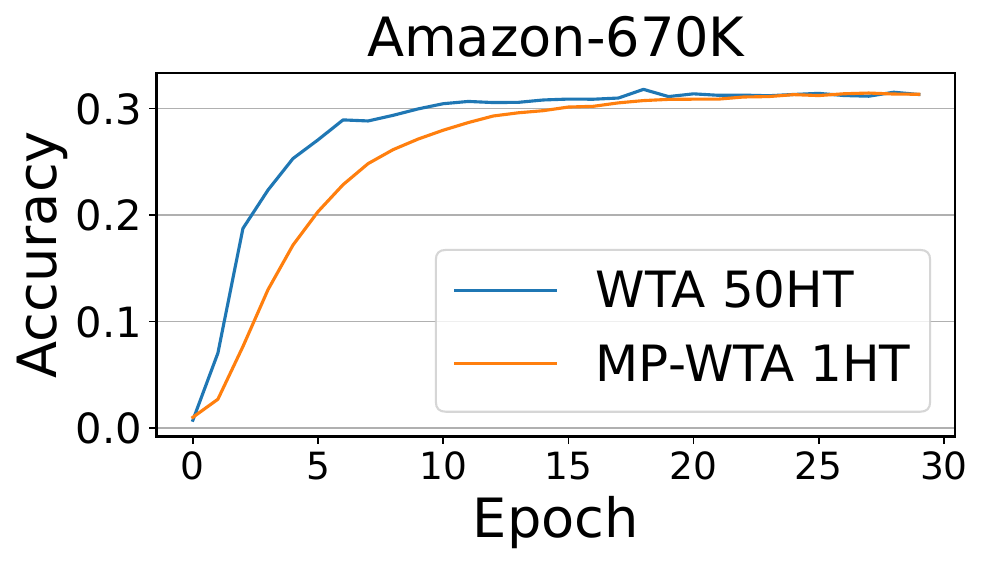}
    \caption{Amazon‑670K}
    \label{fig:amazon670}
  \end{subfigure}
  \hfill
  \begin{subfigure}[b]{0.49\linewidth}
    \centering
    \includegraphics[width=\linewidth]{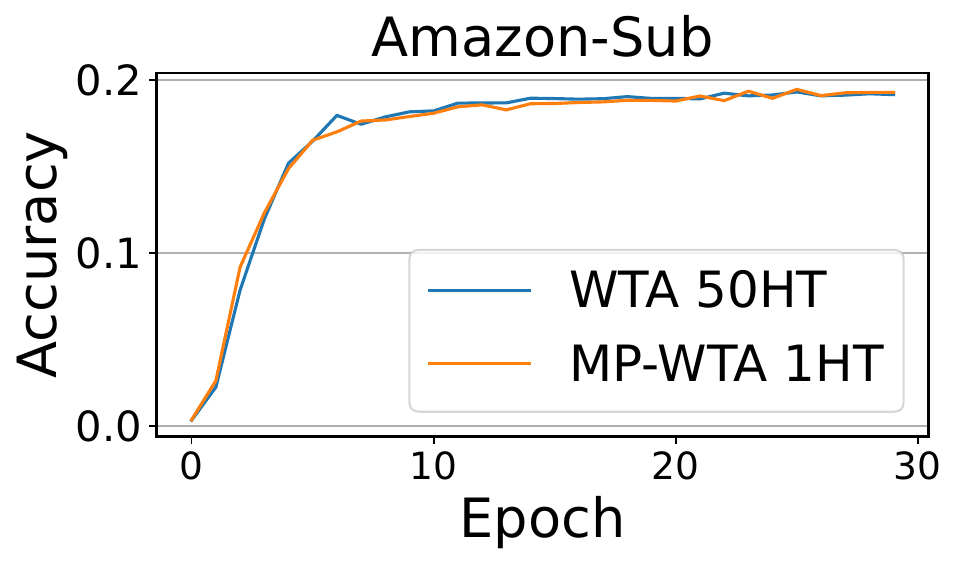}
    \caption{Amazon-Sub}
    \label{fig:amazonprime}
  \end{subfigure}

  \caption{Classification accuracy of WTA ($L{=}50$) vs. MP-WTA ($r{=}3, N{=}3$, $L{=}1$) across all four datasets (window size $M{=}8$).}
  \label{fig:mpwta}
\end{figure}

\textbf{Results.\xspace}
We report Precision at 1 (P@1)---i.e.,~the fraction of test samples whose top-predicted label matches a ground-truth label, following prior work~\cite{chen2020slide, daghaghi2021tale}. 
The results are averaged over 5 independent runs.
As shown in Figure \ref{fig:mpwta}, across all four datasets, the MP–WTA ($r=3,
N=3$) (1HT) curve closely tracks the WTA (50HT) baseline over training epochs.
The gap narrows quickly in early epochs and remains small thereafter; in
several cases the two curves become nearly indistinguishable by around ten
epochs.  With the window size held fixed, multi-probe substitutes for table
multiplicity, yielding comparable accuracy while requiring roughly $50\times$
fewer hash tables and proportionally less hash-table memory.

\begin{figure}[t]
  \centering

  \begin{subfigure}[b]{0.49\linewidth}
    \centering
    \includegraphics[width=\linewidth]{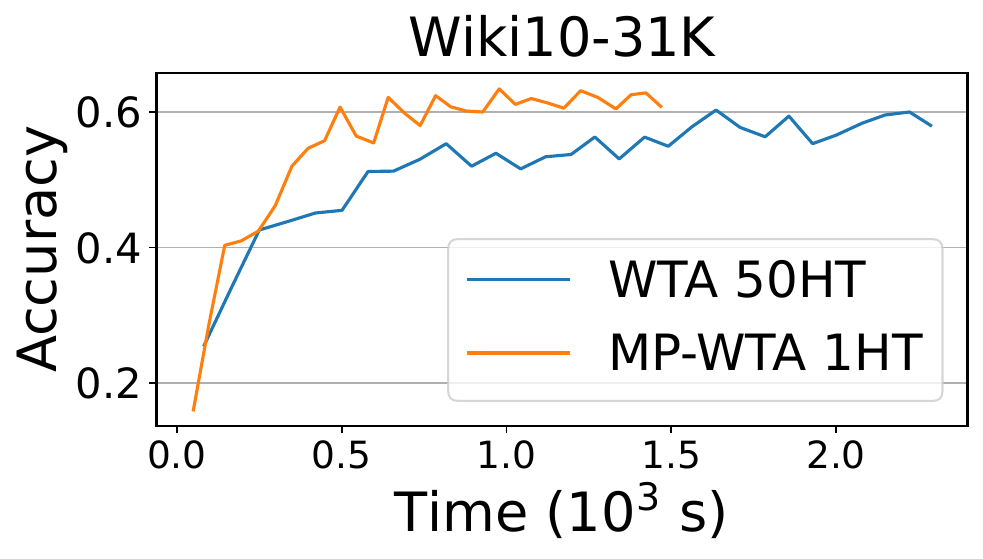}
    \caption{Wiki10‑31K}
    \label{fig:wiki10-t}
  \end{subfigure}
  \hfill
  \begin{subfigure}[b]{0.49\linewidth}
    \centering
    \includegraphics[width=\linewidth]{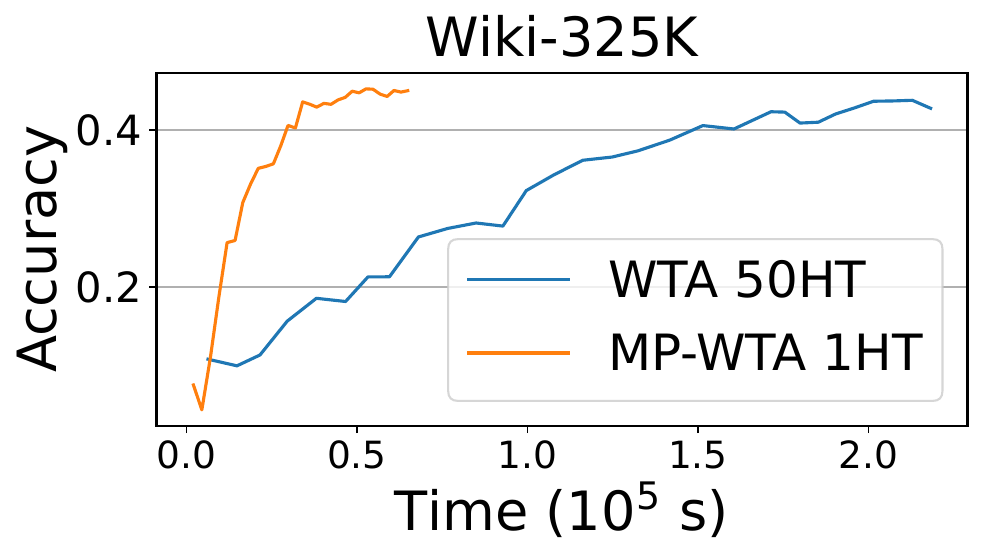}
    \caption{Wiki-325K}
    \label{fig:wiki325-t}
  \end{subfigure}

  \begin{subfigure}[b]{0.49\linewidth}
    \centering
    \includegraphics[width=\linewidth]{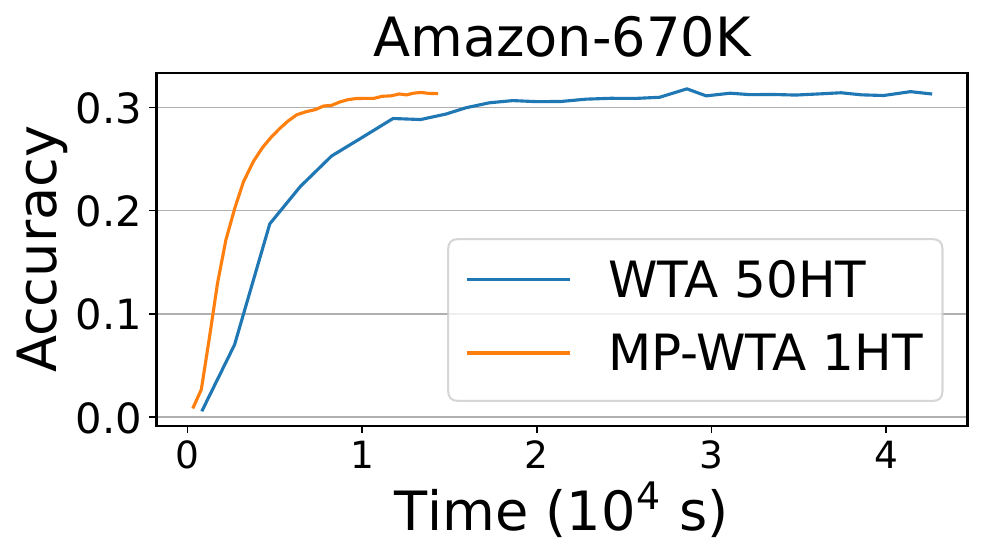}
    \caption{Amazon‑670K}
    \label{fig:amazon670-t}
  \end{subfigure}
  \hfill
  \begin{subfigure}[b]{0.49\linewidth}
    \centering
    \includegraphics[width=\linewidth]{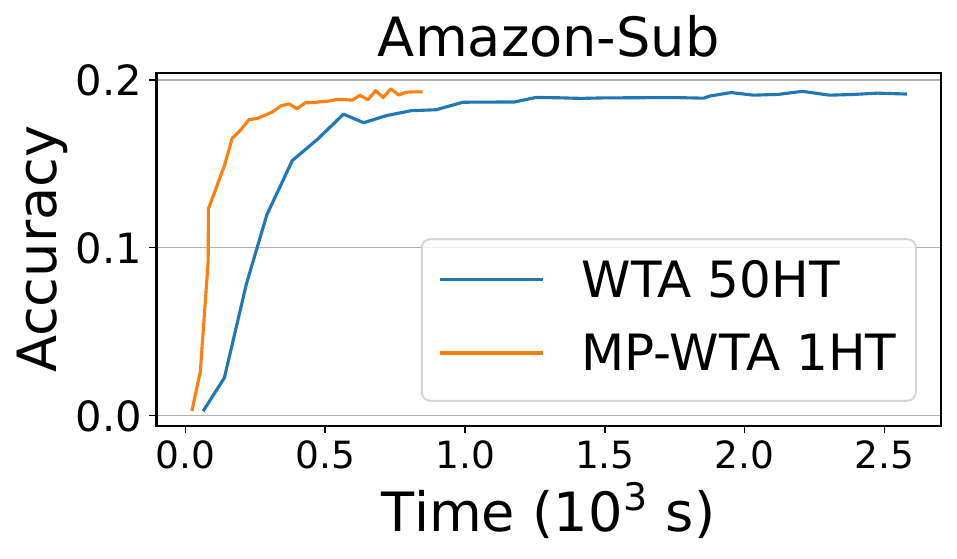}
    \caption{Amazon-Sub}
    \label{fig:amazonprime-t}
  \end{subfigure}

  \caption{Training time vs. classification accuracy comparing WTA ($L{=}50$) and MP-WTA ($r{=}3, N{=}3$, $L{=}1$) across all four datasets. All methods use window size $M{=}8$.}
  \label{fig:mpwta-time}
\end{figure}

We also show the time vs. accuracy results.  Figure \ref{fig:mpwta-time}
uses the same runs as Figure \ref{fig:mpwta}, with epochs converted to
wall-clock time. \emph{Compared with WTA, MP-WTA not only reduces storage
requirements but also accelerates training time by using a single LSH
table.}

\section{Related Work}
\label{sec:related-work}

\noindent\textbf{TEEs for ML.\xspace}
A substantial body of work leverages TEEs to protect models and inputs during
ML training and deployment, typically via enclave-based ML services, model
offloading, obfuscation, or partitioning across TEE/\-GPU
configurations~\cite{hunt2018chiron,tramer2018slalom,ng2021goten,niu20213legrace,bai2025phantom}.
TEESlice~\cite{li2025teeslice} further demonstrates that when adversaries have
access to public pre-trained models, such inference-time protections may fail.
However, these methods primarily target model/input protection or
\emph{inference-time threats}, and do not address end-to-end
\emph{training-time leakage} through data-dependent access patterns.

\noindent\textbf{Obliviousness for Learning.\xspace}
Prior studies have shown that enclave isolation alone is insufficient when
attackers can observe memory access patterns. In response, researchers
introduced algorithm-specific oblivious components for select learning
tasks~\cite{ohrimenko2016oblivious}, including federated aggregation with
sparsification~\cite{kato2023olive}. Sparse and embedding-based models are
especially vulnerable because their access patterns can reveal private
features~\cite{hashemi2022data}. \system treats \emph{training-time
access-pattern leakage} as a core threat and co-designs a dynamic
sparse-training pipeline with integrated oblivious primitives, rather than
safeguarding a single algorithm or aggregation step.

\noindent\textbf{Oblivious Primitives and Operators.\xspace}
ORAM, O$_2$\-RAM, and oblivious data structures provide foundational
access-hiding
primitives~\cite{stefanov2018path,mishra2018oblix,tinoco2023enigmap,asharov2023futorama,zheng2025h2o2ram}.
LAORAM~\cite{rajat2023laoram}, in contrast, optimizes access locality in
embedding-table training by grouping anticipated future accesses. 
Recent oblivious sorting, shuffling, and TEE-based parallel relational
operators further improve the scalability of sorting- and operator-heavy
routines~\cite{gu2025flexway,ngai2024distributed,asharov2020bucket,mavrogiannakis2025obliviator}.
These works complement \system at the primitive/operator level, while \system presents a \emph{system-level solution} for dynamic, sparse training.

\noindent\textbf{Cryptographic PPML (MPC/HE).\xspace}
Approaches based on secure multiparty computation (MPC) and homomorphic encryption (HE) provide strong privacy guarantees without relying on trusted hardware, but tend to incur high computational costs, especially at scale~\cite{mohassel2017secureml,mohassel2018aby3,wagh2019securenn,wagh2021falcon,juvekar2018gazelle}. 
\system instead leverages TEEs for isolation and incorporates data-oblivious primitives, occupying a more practical point in the design space that balances privacy and performance for real-world training tasks.

\section{Conclusion}
\label{sec:conc}

We presented \system, a CPU-friendly training system that enables the training
of wide (and potentially sparse) neural networks within TEEs, while eliminating
value-dependent NN
access-pattern leakage across the whole memory hierarchy, including both
enclave and external memory,
beyond the publicly declared input support. 
\system achieves this by co-designing the training
pipeline with doubly-oblivious building blocks. It combines ORAM-backed storage
with oblivious algorithms and data structures that make the memory behavior of
each training iteration independent of the underlying data and model state. A
key lesson from \system is that practicality hinges on system-level
co-optimization. Naively wrapping training with a baseline ORAM incurs
prohibitive overhead, whereas exploiting the structure of training can
substantially reduce redundant work and bring overhead to a more manageable
level. Our evaluation demonstrates that \system significantly outperforms
standard Path ORAM-based baselines. \system provides stable performance across datasets and thread configurations, while preserving the
intended leakage profile.
Looking forward, \system is well-positioned to benefit from advances in
TEE-friendly oblivious sorting and next-generation doubly-oblivious
primitives, which can further accelerate the components that dominate today's
costs.

\begin{acks}
We thank our anonymous shepherd and reviewers for their feedback. This work was
supported by NSF awards \#2238467, \#2154732, and \#2439951, and by a
Commonwealth Cyber Initiative (CCI) grant from the Industry-Faculty
Collaboration Acceleration Program (IFCAP). Any opinions, findings, and
conclusions or recommendations expressed herein are those of the authors and do
not necessarily reflect the views of the US government, NSF, or CCI.
\end{acks}

\bibliographystyle{ACM-Reference-Format}
\bibliography{main}

@inproceedings{he2016deep,
  title     = {Deep Residual Learning for Image Recognition},
  author    = {He, Kaiming and Zhang, Xiangyu and Ren, Shaoqing and Sun, Jian},
  booktitle = {Proc. of the IEEE Conference on Computer Vision and Pattern Recognition},
  pages     = {770--778},
  year      = {2016}
}

@inproceedings{vaswani2017attention,
  title     = {Attention Is All You Need},
  author    = {Vaswani, Ashish and Shazeer, Noam and Parmar, Niki and Uszkoreit, Jakob and Jones, Llion and Gomez, Aidan N. and Kaiser, {\L}ukasz and Polosukhin, Illia},
  booktitle = {Advances in Neural Information Processing Systems},
  pages     = {5998--6008},
  year      = {2017}
}

@misc{naumov2019deep,
  title         = {Deep Learning Recommendation Model for Personalization and Recommendation Systems},
  author        = {Naumov, Maxim and Mudigere, Dheevatsa and Shi, Hao-Jun Michael and Huang, Jianyu and Sundaraman, Narayanan and Park, Jongsoo and Wang, Xiaodong and Gupta, Udit and Wu, Carole-Jean and Azzolini, Alisson G. and others},
  year          = {2019},
  eprint        = {1906.00091},
  archivePrefix = {arXiv},
}

@inproceedings{covington2016deep,
  title     = {Deep Neural Networks for {YouTube} Recommendations},
  author    = {Covington, Paul and Adams, Jay and Sargin, Emre},
  booktitle = {Proc. of the ACM Conference on Recommender Systems},
  pages     = {191--198},
  year      = {2016}
}

@inproceedings{babbar2017dismec,
  title     = {{DiSMEC}: Distributed Sparse Machines for Extreme Multi-label Classification},
  author    = {Babbar, Rohit and Sch{\"o}lkopf, Bernhard},
  booktitle = {Proc. of the ACM International Conference on Web Search and Data Mining},
  pages     = {721--729},
  year      = {2017}
}

@article{yu2022pecos,
  title   = {{PECOS}: Prediction for Enormous and Correlated Output Spaces},
  author  = {Yu, Hsiang-Fu and Zhong, Kai and Zhang, Jiong and Chang, Wei-Cheng and Dhillon, Inderjit S.},
  journal = {Journal of Machine Learning Research},
  volume  = {23},
  number  = {98},
  pages   = {1--32},
  year    = {2022}
}

@inproceedings{medini2019extreme,
  title     = {Extreme Classification in Log Memory Using Count-Min Sketch: A Case Study of {Amazon} Search with 50{M} Products},
  author    = {Medini, Tharun Kumar Reddy and Huang, Qixuan and Wang, Yiqiu and Mohan, Vijai and Shrivastava, Anshumali},
  booktitle = {Advances in Neural Information Processing Systems},
  pages     = {13244--13254},
  year      = {2019}
}

@inproceedings{schultheis2023towards,
  title     = {Towards Memory-Efficient Training for Extremely Large Output Spaces: Learning with 670k Labels on a Single Commodity {GPU}},
  author    = {Schultheis, Erik and Babbar, Rohit},
  booktitle = {Proc. of the Joint European Conference on Machine Learning and Knowledge Discovery in Databases},
  pages     = {689--704},
  year      = {2023}
}

@article{hoefler2021sparsity,
  title   = {Sparsity in Deep Learning: Pruning and Growth for Efficient Inference and Training in Neural Networks},
  author  = {Hoefler, Torsten and Alistarh, Dan and Ben-Nun, Tal and Dryden, Nikoli and Peste, Alexandra},
  journal = {Journal of Machine Learning Research},
  volume  = {22},
  number  = {241},
  pages   = {1--124},
  year    = {2021}
}

@inproceedings{han2015learning,
  title     = {Learning Both Weights and Connections for Efficient Neural Networks},
  author    = {Han, Song and Pool, Jeff and Tran, John and Dally, William J.},
  booktitle = {Advances in Neural Information Processing Systems},
  pages     = {1135--1143},
  year      = {2015}
}

@inproceedings{frankle2018lottery,
  title     = {The Lottery Ticket Hypothesis: Finding Sparse, Trainable Neural Networks},
  author    = {Frankle, Jonathan and Carbin, Michael},
  booktitle = {Proc. of the International Conference on Learning Representations},
  year      = {2019}
}

@misc{bengio2015conditional,
  title         = {Conditional Computation in Neural Networks for Faster Models},
  author        = {Bengio, Emmanuel and Bacon, Pierre-Luc and Pineau, Joelle and Precup, Doina},
  year          = {2015},
  eprint        = {1511.06297},
  archivePrefix = {arXiv},
}

@inproceedings{shazeer2017outrageously,
  title     = {Outrageously Large Neural Networks: The Sparsely-Gated Mixture-of-Experts Layer},
  author    = {Shazeer, Noam and Mirhoseini, Azalia and Maziarz, Krzysztof and Davis, Andy and Le, Quoc and Hinton, Geoffrey and Dean, Jeff},
  booktitle = {Proc. of the International Conference on Learning Representations},
  year      = {2017}
}

@inproceedings{spring2017scalable,
  title     = {Scalable and Sustainable Deep Learning via Randomized Hashing},
  author    = {Spring, Ryan and Shrivastava, Anshumali},
  booktitle = {Proc. of the ACM SIGKDD Conference on Knowledge Discovery and Data Mining},
  pages     = {445--454},
  year      = {2017}
}

@article{fedus2022switch,
  title   = {Switch Transformers: Scaling to Trillion Parameter Models with Simple and Efficient Sparsity},
  author  = {Fedus, William and Zoph, Barret and Shazeer, Noam},
  journal = {Journal of Machine Learning Research},
  volume  = {23},
  number  = {120},
  pages   = {1--39},
  year    = {2022}
}

@inproceedings{chen2020mongoose,
  title     = {{MONGOOSE}: A Learnable {LSH} Framework for Efficient Neural Network Training},
  author    = {Chen, Beidi and Liu, Zichang and Peng, Binghui and Xu, Zhaozhuo and Li, Jonathan Lingjie and Dao, Tri and Song, Zhao and Shrivastava, Anshumali and R{\'e}, Christopher},
  booktitle = {Proc. of the International Conference on Learning Representations},
  year      = {2021}
}

@inproceedings{state-of-the-uniform,
  title     = {The State of the Uniform: Attacks on Encrypted Databases Beyond the Uniform Query Distribution},
  author    = {Kornaropoulos, Evgenios M. and Papamanthou, Charalampos and Tamassia, Roberto},
  booktitle = {Proc. of the IEEE Symposium on Security and Privacy},
  pages     = {1223--1240},
  year      = {2020}
}

@inproceedings{DBLP:conf/ccs/CashGPR15,
  title     = {Leakage-Abuse Attacks Against Searchable Encryption},
  author    = {Cash, David and Grubbs, Paul and Perry, Jason and Ristenpart, Thomas},
  booktitle = {Proc. of the ACM SIGSAC Conference on Computer and Communications Security},
  pages     = {668--679},
  year      = {2015}
}

@inproceedings{10.1109/SP.2019.00030,
  title     = {Learning to Reconstruct: Statistical Learning Theory and Encrypted Database Attacks},
  author    = {Grubbs, Paul and Lacharit{\'e}, Marie-Sarah and Minaud, Brice and Paterson, Kenneth G.},
  booktitle = {Proc. of the IEEE Symposium on Security and Privacy},
  pages     = {496--512},
  year      = {2019}
}

@inproceedings{response-hiding,
  title     = {Response-Hiding Encrypted Ranges: Revisiting Security via Parametrized Leakage-Abuse Attacks},
  author    = {Kornaropoulos, Evgenios M. and Papamanthou, Charalampos and Tamassia, Roberto},
  booktitle = {Proc. of the IEEE Symposium on Security and Privacy},
  pages     = {1502--1519},
  year      = {2021}
}

@inproceedings{k-nn-attack,
  title     = {Data Recovery on Encrypted Databases with $k$-Nearest Neighbor Query Leakage},
  author    = {Kornaropoulos, Evgenios M. and Papamanthou, Charalampos and Tamassia, Roberto},
  booktitle = {Proc. of the IEEE Symposium on Security and Privacy},
  pages     = {1033--1050},
  year      = {2019}
}

@inproceedings{wilke2024tdxdown,
  title     = {{TDXdown}: Single-Stepping and Instruction Counting Attacks against {Intel TDX}},
  author    = {Wilke, Luca and Sieck, Florian and Eisenbarth, Thomas},
  booktitle = {Proc. of the ACM SIGSAC Conference on Computer and Communications Security},
  pages     = {79--93},
  year      = {2024}
}

@inproceedings{wang2022hertzbleed,
  title     = {{Hertzbleed}: Turning Power Side-Channel Attacks into Remote Timing Attacks on x86},
  author    = {Wang, Yingchen and Paccagnella, Riccardo and He, Elizabeth Tang and Shacham, Hovav and Fletcher, Christopher W. and Kohlbrenner, David},
  booktitle = {Proc. of the USENIX Security Symposium},
  pages     = {679--697},
  year      = {2022}
}

@inproceedings{zhang2024cachewarp,
  title     = {{CacheWarp}: Software-based Fault Injection using Selective State Reset},
  author    = {Zhang, Ruiyi and Gerlach, Lukas and Weber, Daniel and Hetterich, Lorenz and L{\"u}, Youheng and Kogler, Andreas and Schwarz, Michael},
  booktitle = {Proc. of the USENIX Security Symposium},
  pages     = {1135--1151},
  year      = {2024}
}

@inproceedings{kogler2023collide+,
  title     = {{Collide+Power}: Leaking Inaccessible Data with Software-based Power Side Channels},
  author    = {Kogler, Andreas and Juffinger, Jonas and Giner, Lukas and Gerlach, Lukas and Schwarzl, Martin and Schwarz, Michael and Gruss, Daniel and Mangard, Stefan},
  booktitle = {Proc. of the USENIX Security Symposium},
  pages     = {7285--7302},
  year      = {2023}
}

@inproceedings{huang2019gpipe,
  title     = {{GPipe}: Easy Scaling with Micro-Batch Pipeline Parallelism},
  author    = {Huang, Yanping and Cheng, Youlong and Bapna, Ankur and Firat, Orhan and Chen, Mia Xu and Chen, Dehao and Lee, HyoukJoong and Ngiam, Jiquan and Le, Quoc V. and Wu, Yonghui and Chen, Zhifeng},
  booktitle = {Advances in Neural Information Processing Systems},
  pages     = {103--112},
  year      = {2019}
}

@inproceedings{narayanan2021memory,
  title     = {Memory-Efficient Pipeline-Parallel {DNN} Training},
  author    = {Narayanan, Deepak and Phanishayee, Amar and Shi, Kaiyu and Chen, Xie and Zaharia, Matei},
  booktitle = {Proc. of the International Conference on Machine Learning},
  pages     = {7937--7947},
  year      = {2021}
}

@misc{goyal2017accurate,
  title         = {Accurate, Large Minibatch {SGD}: Training {ImageNet} in 1 Hour},
  author        = {Goyal, Priya and Doll{\'a}r, Piotr and Girshick, Ross and Noordhuis, Pieter and Wesolowski, Lukasz and Kyrola, Aapo and Tulloch, Andrew and Jia, Yangqing and He, Kaiming},
  year          = {2017},
  eprint        = {1706.02677},
  archivePrefix = {arXiv},
}

@inproceedings{zheng2025h2o2ram,
  title     = {{H2O2RAM}: A High-Performance Hierarchical Doubly Oblivious {RAM}},
  author    = {Zheng, Leqian and Zhang, Zheng and Dong, Wentao and Zhang, Yao and Wu, Ye and Wang, Cong},
  booktitle = {Proc. of the USENIX Security Symposium},
  pages     = {8501--8520},
  year      = {2025}
}

@inproceedings{rajat2023laoram,
  title     = {{LAORAM}: A Look Ahead {ORAM} Architecture for Training Large Embedding Tables},
  author    = {Rajat, Rachit and Wang, Yongqin and Annavaram, Murali},
  booktitle = {Proc. of the Annual International Symposium on Computer Architecture},
  pages     = {1--15},
  year      = {2023}
}

@article{kato2023olive,
  title   = {Olive: Oblivious Federated Learning on Trusted Execution Environment Against the Risk of Sparsification},
  author  = {Kato, Fumiyuki and Cao, Yang and Yoshikawa, Masatoshi},
  journal = {Proc. of the VLDB Endowment},
  volume  = {16},
  number  = {10},
  pages   = {2404--2417},
  year    = {2023}
}

@inproceedings{gu2025flexway,
  title     = {Flexway {O}-Sort: Enclave-Friendly and Optimal Oblivious Sorting},
  author    = {Gu, Tianyao and Wang, Yilei and Tinoco, Afonso and Chen, Bingnan and Yi, Ke and Shi, Elaine},
  booktitle = {Proc. of the USENIX Security Symposium},
  pages     = {7563--7582},
  year      = {2025}
}

@inproceedings{mavrogiannakis2025obliviator,
  title     = {{OBLIVIATOR}: Oblivious Parallel Joins and Other Operators in Shared Memory Environments},
  author    = {Mavrogiannakis, Apostolos and Wang, Xian and Demertzis, Ioannis and Papadopoulos, Dimitrios and Garofalakis, Minos},
  booktitle = {Proc. of the USENIX Security Symposium},
  pages     = {8521--8540},
  year      = {2025}
}

@inproceedings{tramer2018slalom,
  title     = {Slalom: Fast, Verifiable and Private Execution of Neural Networks in Trusted Hardware},
  author    = {Tram{\`e}r, Florian and Boneh, Dan},
  booktitle = {Proc. of the International Conference on Learning Representations},
  year      = {2019}
}

@misc{hunt2018chiron,
  title         = {Chiron: Privacy-Preserving Machine Learning as a Service},
  author        = {Hunt, Tyler and Song, Congzheng and Shokri, Reza and Shmatikov, Vitaly and Witchel, Emmett},
  year          = {2018},
  eprint        = {1803.05961},
  archivePrefix = {arXiv},
}

@inproceedings{ng2021goten,
  title     = {Goten: {GPU}-Outsourcing Trusted Execution of Neural Network Training},
  author    = {Ng, Lucien K. L. and Chow, Sherman S. M. and Woo, Anna P. Y. and Wong, Donald P. H. and Zhao, Yongjun},
  booktitle = {Proc. of the AAAI Conference on Artificial Intelligence},
  pages     = {14876--14883},
  year      = {2021}
}

@misc{niu20213legrace,
  title         = {{3LegRace}: Privacy-Preserving {DNN} Training over {TEE}s and {GPU}s},
  author        = {Niu, Yue and Ali, Ramy E. and Avestimehr, Salman},
  year          = {2021},
  eprint        = {2110.01229},
  archivePrefix = {arXiv},
}

@inproceedings{bai2025phantom,
  title     = {Phantom: Privacy-Preserving Deep Neural Network Model Obfuscation in Heterogeneous {TEE} and {GPU} System},
  author    = {Bai, Juyang and Chowdhuryy, Md Hafizul Islam and Li, Jingtao and Yao, Fan and Chakrabarti, Chaitali and Fan, Deliang},
  booktitle = {Proc. of the USENIX Security Symposium},
  pages     = {5565--5582},
  year      = {2025}
}

@article{li2025teeslice,
  title   = {{TEESlice}: Protecting Sensitive Neural Network Models in Trusted Execution Environments When Attackers Have Pre-Trained Models},
  author  = {Li, Ding and Zhang, Ziqi and Yao, Mengyu and Cai, Yifeng and Guo, Yao and Chen, Xiangqun},
  journal = {ACM Transactions on Software Engineering and Methodology},
  volume  = {34},
  number  = {6},
  pages   = {1--49},
  year    = {2025}
}

@misc{hashemi2022data,
  title         = {Data Leakage via Access Patterns of Sparse Features in Deep Learning-Based Recommendation Systems},
  author        = {Hashemi, Hanieh and Xiong, Wenjie and Ke, Liu and Maeng, Kiwan and Annavaram, Murali and Suh, G. Edward and Lee, Hsien-Hsin S.},
  year          = {2022},
  eprint        = {2212.06264},
  archivePrefix = {arXiv},
}

@misc{costan2016intel,
  title        = {{Intel SGX} Explained},
  author       = {Costan, Victor and Devadas, Srinivas},
  year         = {2016},
  howpublished = {Cryptology ePrint Archive, Paper 2016/086},
  url          = {https://eprint.iacr.org/2016/086}
}

@article{pinto2019demystifying,
  title   = {Demystifying {Arm TrustZone}: A Comprehensive Survey},
  author  = {Pinto, Sandro and Santos, Nuno},
  journal = {ACM Computing Surveys},
  volume  = {51},
  number  = {6},
  pages   = {1--36},
  year    = {2019}
}

@misc{sev2020strengthening,
  title        = {{AMD SEV-SNP}: Strengthening {VM} Isolation with Integrity Protection and More},
  author       = {{AMD}},
  year         = {2020},
  howpublished = {White Paper},
  url          = {https://docs.amd.com/v/u/en-US/SEV-SNP-strengthening-vm-isolation-with-integrity-protection-and-more}
}

@article{cheng2024intel,
  title   = {{Intel TDX} Demystified: A Top-Down Approach},
  author  = {Cheng, Pau-Chen and Ozga, Wojciech and Valdez, Enriquillo and Ahmed, Salman and Gu, Zhongshu and Jamjoom, Hani and Franke, Hubertus and Bottomley, James},
  journal = {ACM Computing Surveys},
  volume  = {56},
  number  = {9},
  pages   = {1--33},
  year    = {2024}
}

@inproceedings{brasser2017software,
  title     = {Software Grand Exposure: {SGX} Cache Attacks Are Practical},
  author    = {Brasser, Ferdinand and M{\"u}ller, Urs and Dmitrienko, Alexandra and Kostiainen, Kari and Capkun, Srdjan and Sadeghi, Ahmad-Reza},
  booktitle = {Proc. of the USENIX Workshop on Offensive Technologies},
  year      = {2017}
}

@inproceedings{lee2020off,
  title     = {An Off-Chip Attack on Hardware Enclaves via the Memory Bus},
  author    = {Lee, Dayeol and Jung, Dongha and Fang, Ian T. and Tsai, Chia-Che and Popa, Raluca Ada},
  booktitle = {Proc. of the USENIX Security Symposium},
  pages     = {487--504},
  year      = {2020}
}

@inproceedings{lee2017inferring,
  title     = {Inferring Fine-Grained Control Flow Inside {SGX} Enclaves with Branch Shadowing},
  author    = {Lee, Sangho and Shih, Ming-Wei and Gera, Prasun and Kim, Taesoo and Kim, Hyesoon and Peinado, Marcus},
  booktitle = {Proc. of the USENIX Security Symposium},
  pages     = {557--574},
  year      = {2017}
}

@inproceedings{moghimi2017cachezoom,
  title     = {{CacheZoom}: How {SGX} Amplifies the Power of Cache Attacks},
  author    = {Moghimi, Ahmad and Irazoqui, Gorka and Eisenbarth, Thomas},
  booktitle = {Proc. of the International Conference on Cryptographic Hardware and Embedded Systems},
  pages     = {69--90},
  year      = {2017}
}

@inproceedings{moghimi2020copycat,
  title     = {{CopyCat}: Controlled Instruction-Level Attacks on Enclaves},
  author    = {Moghimi, Daniel and Van Bulck, Jo and Heninger, Nadia and Piessens, Frank and Sunar, Berk},
  booktitle = {Proc. of the USENIX Security Symposium},
  pages     = {469--486},
  year      = {2020}
}

@inproceedings{shinde2016preventing,
  title     = {Preventing Page Faults from Telling Your Secrets},
  author    = {Shinde, Shweta and Chua, Zheng Leong and Narayanan, Viswesh and Saxena, Prateek},
  booktitle = {Proc. of the ACM Asia Conference on Computer and Communications Security},
  pages     = {317--328},
  year      = {2016}
}

@inproceedings{van2017telling,
  title     = {Telling Your Secrets without Page Faults: Stealthy Page Table-Based Attacks on Enclaved Execution},
  author    = {Van Bulck, Jo and Weichbrodt, Nico and Kapitza, R{\"u}diger and Piessens, Frank and Strackx, Raoul},
  booktitle = {Proc. of the USENIX Security Symposium},
  pages     = {1041--1056},
  year      = {2017}
}

@inproceedings{xu2015controlled,
  title     = {Controlled-Channel Attacks: Deterministic Side Channels for Untrusted Operating Systems},
  author    = {Xu, Yuanzhong and Cui, Weidong and Peinado, Marcus},
  booktitle = {Proc. of the IEEE Symposium on Security and Privacy},
  pages     = {640--656},
  year      = {2015}
}

@inproceedings{costan2016sanctum,
  title     = {Sanctum: Minimal Hardware Extensions for Strong Software Isolation},
  author    = {Costan, Victor and Lebedev, Ilia and Devadas, Srinivas},
  booktitle = {Proc. of the USENIX Security Symposium},
  pages     = {857--874},
  year      = {2016}
}

@inproceedings{bourgeat2019mi6,
  title     = {{MI6}: Secure Enclaves in a Speculative Out-of-Order Processor},
  author    = {Bourgeat, Thomas and Lebedev, Ilia and Wright, Andrew and Zhang, Sizhuo and Arvind and Devadas, Srinivas},
  booktitle = {Proc. of the Annual IEEE/ACM International Symposium on Microarchitecture},
  pages     = {42--56},
  year      = {2019}
}

@inproceedings{shih2017t,
  title     = {{T-SGX}: Eradicating Controlled-Channel Attacks Against Enclave Programs},
  author    = {Shih, Ming-Wei and Lee, Sangho and Kim, Taesoo and Peinado, Marcus},
  booktitle = {Proc. of the Network and Distributed System Security Symposium},
  year      = {2017}
}

@inproceedings{chen2017detecting,
  title     = {Detecting Privileged Side-Channel Attacks in Shielded Execution with {D{\'e}j{\'a} Vu}},
  author    = {Chen, Sanchuan and Zhang, Xiaokuan and Reiter, Michael K. and Zhang, Yinqian},
  booktitle = {Proc. of the ACM Asia Conference on Computer and Communications Security},
  pages     = {7--18},
  year      = {2017}
}

@article{goldreich1996software,
  title   = {Software Protection and Simulation on Oblivious {RAM}s},
  author  = {Goldreich, Oded and Ostrovsky, Rafail},
  journal = {Journal of the ACM},
  volume  = {43},
  number  = {3},
  pages   = {431--473},
  year    = {1996}
}

@misc{gconnell2022oram,
  title        = {Technology Deep Dive: Building a Faster {ORAM} Layer for Enclaves},
  author       = {Connell, Graeme},
  year         = {2022},
  howpublished = {\url{https://signal.org/blog/building-faster-oram/}},
  note         = {Signal blog. Accessed: 2026-04-17}
}

@article{stefanov2018path,
  title   = {Path {ORAM}: An Extremely Simple Oblivious {RAM} Protocol},
  author  = {Stefanov, Emil and van Dijk, Marten and Shi, Elaine and Chan, T-H. Hubert and Fletcher, Christopher and Ren, Ling and Yu, Xiangyao and Devadas, Srinivas},
  journal = {Journal of the ACM},
  volume  = {65},
  number  = {4},
  pages   = {1--26},
  year    = {2018}
}

@inproceedings{wang2015circuit,
  title     = {Circuit {ORAM}: On Tightness of the {Goldreich-Ostrovsky} Lower Bound},
  author    = {Wang, Xiao and Chan, T-H. Hubert and Shi, Elaine},
  booktitle = {Proc. of the ACM SIGSAC Conference on Computer and Communications Security},
  pages     = {850--861},
  year      = {2015}
}

@inproceedings{patel2018panorama,
  title     = {{PanORAMa}: Oblivious {RAM} with Logarithmic Overhead},
  author    = {Patel, Sarvar and Persiano, Giuseppe and Raykova, Mariana and Yeo, Kevin},
  booktitle = {Proc. of the IEEE Annual Symposium on Foundations of Computer Science},
  pages     = {871--882},
  year      = {2018}
}

@inproceedings{asharov2020optorama,
  title     = {{OptORAMa}: Optimal Oblivious {RAM}},
  author    = {Asharov, Gilad and Komargodski, Ilan and Lin, Wei-Kai and Nayak, Kartik and Peserico, Enoch and Shi, Elaine},
  booktitle = {Advances in Cryptology -- EUROCRYPT},
  pages     = {403--432},
  year      = {2020}
}

@inproceedings{asharov2023futorama,
  title     = {{Futorama}: A Concretely Efficient Hierarchical Oblivious {RAM}},
  author    = {Asharov, Gilad and Komargodski, Ilan and Michelson, Yehuda},
  booktitle = {Proc. of the ACM SIGSAC Conference on Computer and Communications Security},
  pages     = {3313--3327},
  year      = {2023}
}

@inproceedings{mishra2018oblix,
  title     = {Oblix: An Efficient Oblivious Search Index},
  author    = {Mishra, Pratyush and Poddar, Rishabh and Chen, Jerry and Chiesa, Alessandro and Popa, Raluca Ada},
  booktitle = {Proc. of the IEEE Symposium on Security and Privacy},
  pages     = {279--296},
  year      = {2018}
}

@misc{eskandarian2017oblidb,
  title         = {{ObliDB}: Oblivious Query Processing for Secure Databases},
  author        = {Eskandarian, Saba and Zaharia, Matei},
  year          = {2017},
  eprint        = {1710.00458},
  archivePrefix = {arXiv},
}

@inproceedings{wang2014oblivious,
  title     = {Oblivious Data Structures},
  author    = {Wang, Xiao Shaun and Nayak, Kartik and Liu, Chang and Chan, T-H. Hubert and Shi, Elaine and Stefanov, Emil and Huang, Yan},
  booktitle = {Proc. of the ACM SIGSAC Conference on Computer and Communications Security},
  pages     = {215--226},
  year      = {2014}
}

@inproceedings{tinoco2023enigmap,
  title     = {{EnigMap}: External-Memory Oblivious Map for Secure Enclaves},
  author    = {Tinoco, Afonso and Gao, Sixiang and Shi, Elaine},
  booktitle = {Proc. of the USENIX Security Symposium},
  pages     = {4033--4050},
  year      = {2023}
}

@inproceedings{crooks2018obladi,
  title     = {Obladi: Oblivious Serializable Transactions in the Cloud},
  author    = {Crooks, Natacha and Burke, Matthew and Cecchetti, Ethan and Harel, Sitar and Agarwal, Rachit and Alvisi, Lorenzo},
  booktitle = {Proc. of the USENIX Symposium on Operating Systems Design and Implementation},
  pages     = {727--743},
  year      = {2018}
}

@inproceedings{batcher1968sorting,
  title     = {Sorting Networks and Their Applications},
  author    = {Batcher, Kenneth E.},
  booktitle = {Proc. of the AFIPS Spring Joint Computer Conference},
  pages     = {307--314},
  year      = {1968}
}

@inproceedings{asharov2020bucket,
  title     = {Bucket Oblivious Sort: An Extremely Simple Oblivious Sort},
  author    = {Asharov, Gilad and Chan, T-H. Hubert and Nayak, Kartik and Pass, Rafael and Ren, Ling and Shi, Elaine},
  booktitle = {Proc. of the Symposium on Simplicity in Algorithms},
  pages     = {8--14},
  year      = {2020}
}

@inproceedings{sasy2022fast,
  title     = {Fast Fully Oblivious Compaction and Shuffling},
  author    = {Sasy, Sajin and Johnson, Aaron and Goldberg, Ian},
  booktitle = {Proc. of the ACM SIGSAC Conference on Computer and Communications Security},
  pages     = {2565--2579},
  year      = {2022}
}

@inproceedings{ngai2024distributed,
  title     = {Distributed \& Scalable Oblivious Sorting and Shuffling},
  author    = {Ngai, Nicholas and Demertzis, Ioannis and Chamani, Javad Ghareh and Papadopoulos, Dimitrios},
  booktitle = {Proc. of the IEEE Symposium on Security and Privacy},
  pages     = {4277--4295},
  year      = {2024}
}

@inproceedings{chan2017oblivious,
  title     = {Oblivious Hashing Revisited, and Applications to Asymptotically Efficient {ORAM} and {OPRAM}},
  author    = {Chan, T-H. Hubert and Guo, Yue and Lin, Wei-Kai and Shi, Elaine},
  booktitle = {Advances in Cryptology -- ASIACRYPT},
  pages     = {660--690},
  year      = {2017}
}

@inproceedings{sasy2018zerotrace,
  title     = {{ZeroTrace}: Oblivious Memory Primitives from {Intel SGX}},
  author    = {Sasy, Sajin and Gorbunov, Sergey and Fletcher, Christopher W.},
  booktitle = {Proc. of the Network and Distributed System Security Symposium},
  year      = {2018}
}

@inproceedings{dauterman2021snoopy,
  title     = {Snoopy: Surpassing the Scalability Bottleneck of Oblivious Storage},
  author    = {Dauterman, Emma and Fang, Vivian and Demertzis, Ioannis and Crooks, Natacha and Popa, Raluca Ada},
  booktitle = {Proc. of the ACM Symposium on Operating Systems Principles},
  pages     = {655--671},
  year      = {2021}
}

@inproceedings{ren2013design,
  title     = {Design Space Exploration and Optimization of Path Oblivious {RAM} in Secure Processors},
  author    = {Ren, Ling and Yu, Xiangyao and Fletcher, Christopher W. and van Dijk, Marten and Devadas, Srinivas},
  booktitle = {Proc. of the Annual International Symposium on Computer Architecture},
  pages     = {571--582},
  year      = {2013}
}

@article{mittal2021survey,
  title   = {A Survey of Deep Learning on {CPU}s: Opportunities and Co-optimizations},
  author  = {Mittal, Sparsh and Rajput, Poonam and Subramoney, Sreenivas},
  journal = {IEEE Transactions on Neural Networks and Learning Systems},
  volume  = {33},
  number  = {10},
  pages   = {5095--5115},
  year    = {2021}
}

@inproceedings{chen2020slide,
  title     = {{SLIDE}: In Defense of Smart Algorithms over Hardware Acceleration for Large-Scale Deep Learning Systems},
  author    = {Chen, Beidi and Medini, Tharun and Farwell, James and Tai, Charlie and Shrivastava, Anshumali and others},
  booktitle = {Proc. of the Annual Conference on Machine Learning and Systems},
  pages     = {291--306},
  year      = {2020}
}

@inproceedings{indyk1998approximate,
  title     = {Approximate Nearest Neighbors: Towards Removing the Curse of Dimensionality},
  author    = {Indyk, Piotr and Motwani, Rajeev},
  booktitle = {Proc. of the Annual ACM Symposium on Theory of Computing},
  pages     = {604--613},
  year      = {1998}
}

@inproceedings{broder1998min,
  title     = {Min-Wise Independent Permutations},
  author    = {Broder, Andrei Z. and Charikar, Moses and Frieze, Alan M. and Mitzenmacher, Michael},
  booktitle = {Proc. of the Annual ACM Symposium on Theory of Computing},
  pages     = {327--336},
  year      = {1998}
}

@inproceedings{charikar2002simhash,
  title     = {Similarity Estimation Techniques from Rounding Algorithms},
  author    = {Charikar, Moses S.},
  booktitle = {Proc. of the Annual ACM Symposium on Theory of Computing},
  pages     = {380--388},
  year      = {2002}
}

@inproceedings{yagnik2011power,
  title     = {The Power of Comparative Reasoning},
  author    = {Yagnik, Jay and Strelow, Dennis and Ross, David A. and Lin, Ruei-sung},
  booktitle = {Proc. of the IEEE International Conference on Computer Vision},
  pages     = {2431--2438},
  year      = {2011}
}

@inproceedings{andoni2015practical,
  title     = {Practical and Optimal {LSH} for Angular Distance},
  author    = {Andoni, Alexandr and Indyk, Piotr and Laarhoven, Thijs and Razenshteyn, Ilya and Schmidt, Ludwig},
  booktitle = {Advances in Neural Information Processing Systems},
  pages     = {1225--1233},
  year      = {2015}
}

@inproceedings{shrivastava2014asymmetric,
  title     = {Asymmetric {LSH} ({ALSH}) for Sublinear Time Maximum Inner Product Search ({MIPS})},
  author    = {Shrivastava, Anshumali and Li, Ping},
  booktitle = {Advances in Neural Information Processing Systems},
  pages     = {2321--2329},
  year      = {2014}
}

@inproceedings{kitaev2020reformer,
  title     = {Reformer: The Efficient Transformer},
  author    = {Kitaev, Nikita and Kaiser, {\L}ukasz and Levskaya, Anselm},
  booktitle = {Proc. of the International Conference on Learning Representations},
  year      = {2020}
}

@article{chamani2023graphos,
  title   = {{GraphOS}: Towards Oblivious Graph Processing},
  author  = {Chamani, Javad Ghareh and Demertzis, Ioannis and Papadopoulos, Dimitrios and Papamanthou, Charalampos and Jalili, Rasool},
  journal = {Proc. of the VLDB Endowment},
  volume  = {16},
  number  = {13},
  pages   = {4324--4338},
  year    = {2023}
}

@inproceedings{lipp2021platypus,
  title     = {{PLATYPUS}: Software-Based Power Side-Channel Attacks on x86},
  author    = {Lipp, Moritz and Kogler, Andreas and Oswald, David and Schwarz, Michael and Easdon, Catherine and Canella, Claudio and Gruss, Daniel},
  booktitle = {Proc. of the IEEE Symposium on Security and Privacy},
  pages     = {355--371},
  year      = {2021}
}

@inproceedings{matetic2017rote,
  title     = {{ROTE}: Rollback Protection for Trusted Execution},
  author    = {Matetic, Sinisa and Ahmed, Mansoor and Kostiainen, Kari and Dhar, Aritra and Sommer, David and Gervais, Arthur and Juels, Ari and Capkun, Srdjan},
  booktitle = {Proc. of the USENIX Security Symposium},
  pages     = {1289--1306},
  year      = {2017}
}

@inproceedings{lv2007multi,
  title     = {Multi-Probe {LSH}: Efficient Indexing for High-Dimensional Similarity Search},
  author    = {Lv, Qin and Josephson, William and Wang, Zhe and Charikar, Moses and Li, Kai},
  booktitle = {Proc. of the International Conference on Very Large Data Bases},
  pages     = {950--961},
  year      = {2007}
}

@misc{Bhatia16,
  title        = {The Extreme Classification Repository: Multi-label Datasets and Code},
  author       = {Bhatia, K. and Dahiya, K. and Jain, H. and Kar, P. and Mittal, A. and Prabhu, Y. and Varma, M.},
  year         = {2016},
  howpublished = {\url{http://manikvarma.org/downloads/XC/XMLRepository.html}}
}

@inproceedings{daghaghi2021tale,
  title     = {A Tale of Two Efficient and Informative Negative Sampling Distributions},
  author    = {Daghaghi, Shabnam and Medini, Tharun and Meisburger, Nicholas and Chen, Beidi and Zhao, Mengnan and Shrivastava, Anshumali},
  booktitle = {Proc. of the International Conference on Machine Learning},
  pages     = {2319--2329},
  year      = {2021}
}

@inproceedings{juvekar2018gazelle,
  title     = {{GAZELLE}: A Low Latency Framework for Secure Neural Network Inference},
  author    = {Juvekar, Chiraag and Vaikuntanathan, Vinod and Chandrakasan, Anantha},
  booktitle = {Proc. of the USENIX Security Symposium},
  pages     = {1651--1669},
  year      = {2018}
}

@inproceedings{mohassel2018aby3,
  title     = {{ABY3}: A Mixed Protocol Framework for Machine Learning},
  author    = {Mohassel, Payman and Rindal, Peter},
  booktitle = {Proc. of the ACM SIGSAC Conference on Computer and Communications Security},
  pages     = {35--52},
  year      = {2018}
}

@inproceedings{mohassel2017secureml,
  title     = {{SecureML}: A System for Scalable Privacy-Preserving Machine Learning},
  author    = {Mohassel, Payman and Zhang, Yupeng},
  booktitle = {Proc. of the IEEE Symposium on Security and Privacy},
  pages     = {19--38},
  year      = {2017}
}

@inproceedings{wagh2019securenn,
  title   = {{SecureNN}: 3-Party Secure Computation for Neural Network Training},
  author  = {Wagh, Sameer and Gupta, Divya and Chandran, Nishanth},
  booktitle = {Proc. on Privacy Enhancing Technologies},
  volume  = {2019},
  number  = {3},
  pages   = {26--49},
  year    = {2019}
}

@inproceedings{wagh2021falcon,
  title   = {Falcon: Honest-Majority Maliciously Secure Framework for Private Deep Learning},
  author  = {Wagh, Sameer and Tople, Shruti and Benhamouda, Fabrice and Kushilevitz, Eyal and Mittal, Prateek and Rabin, Tal},
  booktitle = {Proc. on Privacy Enhancing Technologies},
  volume  = {2021},
  number  = {1},
  pages   = {188--208},
  year    = {2021}
}

@article{canetti2000security,
  title   = {Security and Composition of Multiparty Cryptographic Protocols},
  author  = {Canetti, Ran},
  journal = {Journal of Cryptology},
  volume  = {13},
  number  = {1},
  pages   = {143--202},
  year    = {2000}
}

@incollection{lindell2017simulate,
  title     = {How to Simulate It: A Tutorial on the Simulation Proof Technique},
  author    = {Lindell, Yehuda},
  booktitle = {Tutorials on the Foundations of Cryptography},
  pages     = {277--346},
  year      = {2017}
}

@inproceedings{ohrimenko2016oblivious,
  title     = {Oblivious Multi-Party Machine Learning on Trusted Processors},
  author    = {Ohrimenko, Olga and Schuster, Felix and Fournet, C{\'e}dric and Mehta, Aastha and Nowozin, Sebastian and Vaswani, Kapil and Costa, Manuel},
  booktitle = {Proc. of the USENIX Security Symposium},
  pages     = {619--636},
  year      = {2016}
}

@inproceedings{kingma2014adam,
  title     = {Adam: A Method for Stochastic Optimization},
  author    = {Kingma, Diederik P. and Ba, Jimmy},
  booktitle = {Proc. of the International Conference on Learning Representations},
  year      = {2015}
}

@inproceedings{schoppmann2019make,
  title     = {Make Some Room for the Zeros: Data Sparsity in Secure Distributed Machine Learning},
  author    = {Schoppmann, Phillipp and Gasc{\'o}n, Adri{\`a} and Raykova, Mariana and Pinkas, Benny},
  booktitle = {Proc. of the ACM SIGSAC Conference on Computer and Communications Security},
  pages     = {1335--1350},
  year      = {2019}
}

@misc{damie2025secure,
  title         = {Secure Sparse Matrix Multiplications and Their Applications to Privacy-Preserving Machine Learning},
  author        = {Damie, Marc and Hahn, Florian and Peter, Andreas and Ramon, Jan},
  year          = {2025},
  eprint        = {2510.14894},
  archivePrefix = {arXiv},
}

@inproceedings{umar2025efficient,
  title     = {Efficient Memory Side-Channel Protection for Embedding Generation in Machine Learning},
  author    = {Umar, Muhammad and Marathe, Akhilesh Parag and Gupta, Monami Dutta and Ghosh, Shubham Jogprakash and Suh, G. Edward and Xiong, Wenjie},
  booktitle = {Proc. of the IEEE International Symposium on High-Performance Computer Architecture},
  pages     = {423--441},
  year      = {2025}
}

\appendix

\section{Open Science}

This appendix lists the artifacts necessary to evaluate the paper's contributions and how readers can access them.
Note that, due to hardware requirements (Intel TDX), full end-to-end reproduction may not be feasible on all platforms; we detail the specific requirements below.

\textbf{Datasets.\xspace}
All datasets used in this paper are publicly available in standard XML format (e.g.,~Wiki10-31K, Wiki-325K, and Amazon-670K).
They can be obtained directly from the XML Repository: \url{http://manikvarma.org/downloads/XC/XMLRepository.html}.
For any derived subset (e.g.,~Amazon-sub), our artifact provides the exact raw data used and the scripts to generate it.

\textbf{Source code.\xspace}
We provide an artifact containing:
(i)~our full implementation,
(ii)~the baseline integration,
(iii)~all the configuration files used in our paper,
and (iv)~a Python script to generate all plots from pre-collected experiment outputs.
Our artifact is available at the following repository: \url{https://github.com/dorafan/TENNOR}

\textbf{Build \& Run.\xspace}
The artifact includes a \texttt{README} with:
(1)~the required toolchain versions,
(2)~a single-command build (CMake),
and (3)~commands to run each experiment.
\emph{Hardware requirements:}
The full pipeline (including TEE-based experiments) requires an Intel
TDX-enabled platform with sufficient memory. The accuracy component
(\texttt{Slide-MP}) can run on commodity hardware, though large configurations
(e.g.,~50 hash tables) benefit from ample memory.

\textbf{Figure/table generation.\xspace}
We provide a Python script that generates all plots and tables from
pre-collected experiment outputs included in the artifact.  To reproduce these
outputs from scratch, readers need first to run the corresponding experiments
on a compatible platform and collect the raw results; the \texttt{README}
documents the required steps.

\section{Ethical Considerations}

\noindent\textbf{Summary and scope.\xspace}
\system is a defensive system that enables training of wide-layer neural networks inside Trusted Execution Environments (TEEs) while mitigating input-dependent memory-access patterns from sparse training.
It reduces confidentiality risks in untrusted cloud settings.
Our evaluation uses public extreme multi-label classification benchmarks (Wiki10-31K, Wiki-325K, Amazon-670K, and a derived Amazon-670K subset).
We do not collect new user data, conduct human-subject experiments, use deception, or interact with third-party services.
All experiments run on authorized Google Cloud Intel TDX instances and an authorized institutional CPU cluster.
No IRB/ERB review was applicable.
 
\noindent\textbf{Stakeholders.\xspace}
The main stakeholders are:
(1) data owners whose training data may be sensitive;
(2) model owners whose parameters and training dynamics may be proprietary;
(3) practitioners deploying TEE-based ML systems;
(4) cloud operators and hardware vendors providing TEE infrastructure;
(5) researchers building on oblivious training, ORAM, LSH, and sparse-learning techniques;
(6) end users affected by models trained with such systems;
and (7) dual-use actors who may reason about the costs and limits of privacy-preserving training from the published design and performance analysis.

\noindent\textbf{Benefits and harms.\xspace}
\emph{Benefits}---\system makes privacy-preserving training of wide neural networks more practical, benefiting deployments where training data or model parameters are confidential.
\emph{Harms.}---The main foreseeable risk is misinterpretation of security guarantees: 
\system protects input-dependent memory-access patterns under its stated threat model, but does not provide end-to-end protection against all side channels (microarchitectural timing, power, physical), rollback attacks, software bugs, output-level leakage, or incorrect deployment.
A secondary dual-use risk is that the performance analysis may help adversaries estimate the cost of privacy-preserving training.
We consider this limited since the paper doesn't provide any attack tools or evasion techniques.
Finally, oblivious protection increases computation and memory traffic relative to non-oblivious training, which may raise energy and infrastructure costs.
 
\noindent\textbf{Risk mitigation.\xspace}
We state a narrow threat model with explicit non-goals (Section~3), so users do not over-rely on the guarantees.
All experiments are conducted in controlled, authorized environments with no probing of third-party systems.
The released artifact is designed for reproducibility and does not include tools for attacking deployed systems or extracting private data.
 
\noindent\textbf{Respect for persons.\xspace}
We use benchmark datasets according to XC repository's published research-use terms and do not attempt to re-identify individuals.
Dataset artifacts and generation scripts are released only to the extent permitted by those terms.
 
\noindent\textbf{Justice.\xspace}
The primary beneficiaries are organizations that need (training) data confidentiality.
The burdens of additional computational cost and TEE-capable infrastructure may disadvantage resource-constrained institutions.
To mitigate this, we evaluate on publicly available cloud infrastructure and transparently report overheads so practitioners can assess tradeoffs without re-running expensive experiments.
 
\noindent\textbf{Respect for law and public interest.\xspace}
We experiment only on authorized systems and do not discover, exploit, or disclose vulnerabilities in TEE implementations or cloud platforms.
Should vulnerabilities be discovered in future work, we would follow coordinated disclosure practices.
 
\noindent\textbf{Conclusion.\xspace}
We believe the benefits of reducing access-pattern leakage for privacy-preserving training outweigh the foreseeable risks, provided \system is deployed within its stated threat model.

\section{Generative AI Usage}
\label{app:ai-usage}

We used Claude (Anthropic) and ChatGPT (OpenAI) during the preparation of this paper. 
These tools assisted with rephrasing and grammar refinement in the manuscript, and with drafting portions of the implementation and plotting code. 
All AI-generated text and code were reviewed and verified by the authors before inclusion. 
The experimental results reported in this paper were produced by the authors on their own infrastructure. 
The authors take full responsibility for the correctness of all claims, results, and references in this paper.
\section{Dataset and Preprocessing Details}
\label{app:data}

We use public multi-label classification benchmarks from the Extreme Classification Repository~\cite{Bhatia16}.
All datasets are provided as sparse bag-of-words features with integer label IDs, using the standard
\texttt{train.txt}/\-\texttt{test.txt} split.
Except for the Amazon-Sub dataset (synthetic), all others are available on the XC repository.

\textbf{Wiki10-31K (Wiki10).}
Wiki10-31K is a tag-prediction dataset constructed from Wikipedia pages, where each instance is represented as a sparse bag-of-words vector and labeled with user-provided tags.
We use the standard split with 14,146 training instances and 6,616 test instances, 30,938 labels, and 101,938 features.
We use the repository-provided features and labels directly.

\textbf{Amazon-670K (Amz670).}
Amazon-670K is a product recommendation dataset with 670K labels.
Each input corresponds to a product represented as a sparse feature vector, and labels denote other products that a user may be interested in purchasing.
We use the standard split (490,449 train/\-153,025 test) with 670,091 labels and 135,909 features.

\textbf{WikiLSHTC-325K (Wiki325).}
WikiLSHTC-325K is extracted from Wikipedia, where each instance is a document represented using sparse bag-of-words features and the label indicates its category.
The benchmark contains 1,778,351 training instances and 587,084 test instances, with 325,056 labels and 1,617,899 features.
We use the standard split and the repository-provided sparse representation.

\textbf{Amazon-Sub (Amz-Sub).}
To obtain a mid-scale benchmark that falls between Wiki10-31K and WikiLSHTC-325K while preserving the original feature space, we construct Amazon-sub from Amazon-670K using a strict label-budgeted subsetting procedure.
We set the target training size and label budget as:
$N_{\text{target}}=\textit{multiplier}\cdot 14{,}146$ and $L_{\text{target}}=\textit{multiplier}\cdot 30{,}938$,
where \textit{multiplier}$\in\{2,3\}$ matches the scale of Wiki10-31K.

For training, we enforce an \emph{ALL} policy: selected training instances keep all of their labels (no partial deletion),
and the union of labels is bounded by $L_{\text{target}}$.
We use a two-phase cover-and-fill heuristic: 
Phase A greedily selects instances that introduce many new labels,
biased toward frequent labels under the label budget; 
Phase B fills the remaining quota using only instances whose label sets are subsets of the kept label set (introducing no new labels).

For testing, we enforce an \emph{ANY} policy by keeping only instances with at least one kept label and dropping labels outside the kept label set.
We then remap labels to a contiguous range.
Optionally, we downsample the resulting test set to approximately one-third of the selected training set size, matching the train/\-test split used in Wiki10-31K and Amazon-670K.

Our scripts output \texttt{train.small.txt}, \texttt{test.small.txt}, \texttt{label\_\-map.json}, \texttt{kept\_train\_idx.txt}, and \texttt{stats.json}, where \texttt{stats.\-json} records the realized instance/label counts after subsetting.

\section{Memory Usage}
\label{app:memory}

\begin{table}[t]
\centering
\caption{Peak resident memory (GB) of \system on the Google Cloud Confidential C3 instance (352 GB RAM; see~\S\ref{sec:eval-performance}). The last column is the ratio between the two batch sizes; ``--'' marks a configuration that exceeds available memory.}
\label{tab:memory}
\begin{tabular}{@{}lrrrr@{}}
\toprule
Dataset & \#Labels & $B{=}32$ & $B{=}64$ & Ratio \\
\midrule
Wiki10-31K & 30{,}938  & 34.6 GB  & 69.5 GB  & 2.0$\times$ \\
Amz-Sub    & 92{,}814  & 77.7 GB  & 326.3 GB & 4.2$\times$ \\
Wiki-325K  & 325{,}056 & 222.8 GB & --    & --          \\
\bottomrule
\end{tabular}
\end{table}

Table~\ref{tab:memory} reports \system's peak resident memory for each evaluated dataset at batch sizes 32 and 64, measured on the 352~GB Google Cloud Confidential C3 instance described in Section~\ref{sec:eval-performance}.
The measurements exhibit two distinct scaling regimes.
For Amazon-Sub, doubling the batch size from 32 to 64 increases peak memory by approximately $4.2\times$, whereas for Wiki10-31K the same change increases peak memory by approximately $2.0\times$.
As explained below, these phenomena correspond to an ``uncapped regime'' in which the dominant OHT memory grows approximately quadratically with the batch size, and a ``capped regime'' in which it grows approximately linearly.

We next explain how this behavior follows from the choice of public parameters used for the OHT.
Let $B$ denote the batch size, $\mathsf{lenSeq}$ the fixed number of MP-WTA probes (i.e., requests) issued per input, and $T = B \cdot \mathsf{lenSeq}$ the total number of requests generated for one batch.
The OHT contains $\mathsf{numBin}$ bins in each of its two tiers, and each bin has a fixed capacity of $\mathsf{binCap}$ request entries.
Thus, the OHT memory is determined by the product $\mathsf{numBin} \times \mathsf{binCap}$.
The associated \texttt{Request-Array} also contains $T$ entries, but the OHT dominates the measured batch-dependent memory.
The LSH table is a smaller, secondary cost: it stores $\mathsf{numLshBucs} \times \mathsf{PADSIZE}$ neuron slots (i.e., the number of LSH buckets times the LSH bucket capacity) and is independent of $B$, although it grows as the output layer labels.

The two public parameters of OHT, $\mathsf{numBin}$ and $\mathsf{binCap}$, scale differently with the batch size.
$\mathsf{numBin}$, the number of OHT bins,  grows linearly with the number of requests $T = B \cdot \mathsf{lenSeq}$, and therefore linearly with $B$ for a fixed $\mathsf{lenSeq}$.
However, the number of OHT bins $\mathsf{numBin}$ is capped at the public value $\mathsf{numLshBucs}$, the number of LSH buckets.
Because every request is associated with one of
the $\mathsf{numLshBucs}$ possible \texttt{LSHbucketID}s.
Once $\mathsf{numBin}$ reaches this value, increasing it further does not increase the number of distinct LSH buckets that may be associated with the requests; it only occurs additional OHT capacity (i.e., different requests are associated with the same LSH bucket, which has the same \texttt{LSHbucketID}).
We therefore use the public parameter $\mathsf{numLshBucs}$ as an upper bound on $\mathsf{numBin}$.
Because both $\mathsf{numLshBucs}$ and this bounding rule are public, applying the cap introduces no additional leakage.

The per-bin capacity $\mathsf{binCap}$ grows linearly with $B$.
As the batch size increases, requests generated by different inputs may be associated with the same \texttt{LSHbucketID} and therefore placed in the same OHT bin.
The number of requests that a bin must accommodate can thus increase proportionally with the number of inputs in the batch, so $\mathsf{binCap}$ grows approximately linearly with $B$.

These relationships produce two memory-scaling regimes.
In the uncapped regime, $\mathsf{numBin}$ remains below $\mathsf{numLshBucs}$.
Since $T$ grows linearly with $B$, $\mathsf{numBin}$ also grows linearly with $B$.
At the same time, $\mathsf{binCap}$ grows linearly with $B$.
Because the dominant OHT allocation is the product $\mathsf{numBin} \times \mathsf{binCap}$, its memory grows approximately quadratically with the batch size $B$ in this regime.

In the capped regime, $\mathsf{numBin}$ has already reached $\mathsf{numLshBucs}$ and no longer increases with $B$.
Only $\mathsf{binCap}$ continues to grow linearly with the batch size.
The dominant OHT memory therefore grows approximately linearly with $B$ in this regime.
These scaling relationships describe the dominant OHT allocation rather than every component of process memory.
The measured peak resident memory also includes the model, optimizer state, LSH table, dataset buffers, \texttt{Request-Array}, so the end-to-end measurements are expected to follow these trends only approximately.

Table~\ref{tab:memory} illustrate these two regimes.
For Amazon-Sub, $\mathsf{numBin}$ remains below the $\mathsf{numLshBucs}$ cap at both evaluated batch sizes.
Doubling $B$ therefore increases both $\mathsf{numBin}$ and $\mathsf{binCap}$.
Peak memory rises from 77.7~GB at $B=32$ to 326.3~GB at $B=64$, a $4.2\times$ increase that is consistent with the quadratic regime.

For Wiki10-31K, the smaller LSH key space causes $\mathsf{numBin}$ to reach the $\mathsf{numLshBucs}$ cap already at $B=32$.
Doubling the batch size therefore increases $\mathsf{binCap}$ but not $\mathsf{numBin}$.
Peak memory rises from 34.6~GB to 69.5~GB, a $2.0\times$ increase that is consistent with the linear regime.

For the largest evaluated workload, Wiki-325K, peak memory is 222.8~GB at batch size 32.
Batch size 64 exceeds the 352~GB capacity of our evaluation instance and is therefore not reported.
This limitation is consistent with the uncapped regime: for large workloads, increasing the batch size simultaneously increases the number of provisioned OHT bins and the capacity of each bin.

At a fixed batch size of 32, peak memory also increases across the evaluated workloads as the output layer widens, from 34.6~GB for Wiki10-31K to 77.7~GB for Amazon-Sub and 222.8~GB for Wiki-325K.
This trend is driven by both the larger persistent structures, including the model and LSH table, and the larger OHT configuration used by the wider workloads.
Because the datasets also differ in their model and LSH parameters, this cross-dataset comparison characterizes the memory usage of the evaluated workloads rather than isolating label count as the only changing variable.

The 352~GB value is the capacity of the cloud instance used for our evaluation, rather than a fixed memory requirement imposed by \system.
The actual requirement depends on the workload, output-layer width, batch size, and public parameters.
Our largest successful configuration, Wiki-325K at batch size 32, uses 222.8~GB, approximately 63\% of the available memory.

For workloads in the uncapped regime, increasing the micro-batch size is particularly expensive because the dominant OHT allocation grows approximately quadratically with $B$.
As discussed in Section~\ref{sec:eval-performance}, a larger effective optimization batch can instead be obtained by fixing $B_{\mathsf{micro}}=32$ and accumulating gradients across multiple micro-batches.
This increases the effective batch size without increasing the per-micro-batch OHT working set.

\section{Additional Component Details}
\label{app:component-details}

\subsection{MP-WTA}
\label{app:mp-wta}

\begin{algorithm}[ht]
  \caption{\textsc{$h^2$ WTA} Hash}
      \footnotesize
  \label{alg:abstract_max2wta}
  \begin{algorithmic}[1]
    \Require vector $\mathbf{x}\!\in\!\mathbb{R}^{d}$, window size $M$, \# hash functions $K$
    \Ensure WTA signature $\mathbf{h}\!\in\!\mathbb{N}^K$, second-choice signature $\mathbf{h}^{(2)}\!\in\!\mathbb{N}^K$
    \For{$i\gets1$ \textbf{to} $K$}
        \State $S_i \gets$ the fixed sampled subset of size $M$ associated with the $i$-th hash function
        \State $h_i \gets$ index of the maximum among the elements of $S_i$ 
        \State $h^2_i \gets$ index of the second-maximum among the elements of $S_i$
    \EndFor
    \State \Return $(\mathbf{h},\mathbf{h}^{(2)})$
  \end{algorithmic}
\end{algorithm}

Algorithm \ref{alg:abstract_max2wta} illustrates a high-level method for computing $h^2(\cdot)$ WTA hash values.
If multiple sampled features have the same value, ties are broken deterministically according to the fixed sampled order of that hash function.
The computation of $h^3(\cdot)$ is analogous, using the third-largest sampled feature.

The WTA hash can be computed obliviously using comparison and assignment primitives (see Appendix~\ref{app:obl-primitive}).

\subsection{Neuron Fetcher}
\label{app:neuron-fetcher}

\begin{algorithm}[htb]
    \caption{\textsc{Neuron Fetcher for layer $\ell$}}
    \label{alg:neuron-fetcher}
    \footnotesize
    \begin{algorithmic}[1]
        \Require
            Layer index $\ell$;
            
            $\mathsf{ReqArray}_\ell$ of $T = B \cdot \mathsf{lenSeq}_\ell$ entries, where:
            (i) in \textsc{Read} mode, it is the initialized array produced by Neuron Requester (Algorithm~\ref{alg:neuron-requester}); and
            (ii) in \textsc{Write} mode, it is the merged-and-updated array produced after \textsc{OblMergeRequests} (Algorithm~\ref{alg:obl-merge}) and the oblivious optimizer (Appendix~\ref{app:obl-optimizer}), with at most one non-dummy entry per \texttt{LSHbucketID};
            
            LSH table $\mathsf{LSH}_\ell$ consisting of $\mathsf{numLshBucs}$ buckets, each padded to fixed capacity $\mathsf{PADSIZE}$ with dummy neurons;
            
            OHT tier-1 and tier-2 hash functions $\mathsf{H}_1, \mathsf{H}_2: \{0,\ldots,\mathsf{numLshBucs}-1\} \to \{0,\ldots,\mathsf{numBin}-1\}$~\cite{chan2017oblivious},
            with fixed bin capacity $\mathsf{binCap}$;
            
            $\mathsf{mode} \in \{\textsc{Read}, \textsc{Write}\}$, where \textsc{Read} transfers neurons from $\mathsf{LSH}_\ell$ into \texttt{Neuron-Buffer} (feedforward) and \textsc{Write} transfers in the reverse direction (backpropagation);
            
            oblivious primitives~\cite{dauterman2021snoopy}
            (\S\ref{app:obl-primitive}): \texttt{OblCompare}$(a,b)$ (returns $1$ iff $a=b$, data-obliviously) and
            \texttt{OblChoose}$(\mathsf{flag}, v_1, v_0)$ (obliviously returns $v_1$ if $\mathsf{flag}=1$, else $v_0$, data-obliviously)

        \Ensure $\mathsf{ReqArray}_\ell$ with each entry's \texttt{Neuron-Buffer} populated (\textsc{Read}) or $\mathsf{LSH}_\ell$ updated with modified neurons (\textsc{Write})

        \Statex
        \State \textit{// Phase 1: Build OHT from Request-Array (key = \texttt{LSHbucketID})}
        \State $\mathsf{OHT} \gets \texttt{OHT.Build}(\mathsf{ReqArray}_\ell)$
            \Comment{oblivious construction via oblivious sorting~\cite{chan2017oblivious}}

        \Statex
        \State \textit{// Phase 2: Linear scan of LSH table; In \textsc{Write} mode, this step assumes that $\mathsf{ReqArray}_\ell$ has already been merged so that at most one non-dummy entry matches any given \texttt{LSHbucketID}; otherwise multiple matching writes would overwrite one another.}
        \For{$b = 0$ \textbf{to} $\mathsf{numLshBucs} - 1$}
        \Statex \textit{// performs \texttt{OHT.lookup}~\cite{chan2017oblivious}}
            \State $\mathsf{bin}_1 \gets \mathsf{H}_1(b)$;~~$\mathsf{bin}_2 \gets \mathsf{H}_2(b)$
            \For{$\mathsf{tier} \in \{1, 2\}$}
                \For{each entry $e$ in $\mathsf{OHT}.\mathsf{bins}[\mathsf{bin}_\mathsf{tier}]$}
                    \Comment{iterate all $\mathsf{binCap}$ entries}
                    
                    \If{$\mathsf{mode} = \textsc{Read}$}
                        \State $\mathsf{match} \gets \texttt{OblCompare}(e.\texttt{LSHbucketID},\; b)$
                        \State $e.\texttt{Neuron-Buffer} \gets$
                        \Statex \hspace{\algorithmicindent}$\texttt{OblChoose}(\mathsf{match},\; \mathsf{LSH}_\ell[b],\; e.\texttt{Neuron-Buffer})$
                    \EndIf
                    \If{$\mathsf{mode} = \textsc{Write}$}
                        \State $\mathsf{matchW} \gets \texttt{OblCompare}(e.\texttt{LSHbucketID}, b)$
                        \Statex \hspace{\algorithmicindent}$\land\,(\neg e.\mathsf{isDummy})$
                        \State $\mathsf{LSH}_\ell[b] \gets$
                        \Statex \hspace{\algorithmicindent}$\texttt{OblChoose}(\mathsf{matchW},\; e.\texttt{Neuron-Buffer},\; \mathsf{LSH}_\ell[b])$
                    \EndIf
                \EndFor
            \EndFor
        \EndFor

        \Statex
        \State \textit{// Phase 3: Extract populated entries from OHT back into Request-Array}
        \State Collect all entries from $\mathsf{OHT}$ bins into a flat array $A$, each entry carries an \texttt{isDummy} flag.
        \State $\texttt{OblSort}(A)$ by key $(\texttt{isDummy},\;\texttt{Batch-Rank})$
            \Comment{non-dummy entries first; dummy padding after them}
        \State $\mathsf{ReqArray}_\ell \gets A[0..T-1]$
            \Comment{restore fixed-length request array}

        \State \Return $\mathsf{ReqArray}_\ell, \mathsf{LSH}_\ell$
    \end{algorithmic}
\end{algorithm}

Algorithm~\ref{alg:neuron-fetcher} details the Neuron Fetcher procedure outlined in Section~\ref{sec:neuron-fetcher}.
It takes the initialized $\mathsf{ReqArray}_\ell$ produced by Neuron Requester (Algorithm~\ref{alg:neuron-requester}) and populates each entry's \texttt{Neuron-Buffer} with the neurons stored in the corresponding LSH bucket.
The algorithm is parameterized by a $\mathsf{mode}$ flag: in \textsc{Read} mode (feedforward), neurons are transferred from the LSH table into the \texttt{Neuron-Buffer}s; in \textsc{Write} mode (backpropagation), the direction is reversed, writing updated neurons back to the LSH table.
The procedure consists of three phases.
 
\emph{Phase~1: OHT Construction (line~2).}
The $T = B \cdot \mathsf{lenSeq}_\ell$ entries of $\mathsf{ReqArray}_\ell$ are transformed into an Oblivious Two-Tier Hash Table (\texttt{OHT})~\cite{chan2017oblivious}, keyed by \texttt{LSHbucketID}, via oblivious sorting (e.g.,~bitonic sort~\cite{batcher1968sorting}).
Since the \texttt{OHT} stores request records and resolves hash collisions through fixed-capacity bins, a single bin may contain entries for multiple distinct \texttt{LSHbucketID}s and may also contain multiple records with the same \texttt{LSHbucketID}.
Additionally, each bin is padded with dummy entries (marked by an \texttt{isDummy$=1$} flag) to its public capacity $\mathsf{binCap}$, ensuring that the bin occupancy does not leak information about the request distribution.
 
\emph{Phase~2: Linear Scan (lines~3--13).}
Neuron Fetcher linearly scans all $\mathsf{numLshBucs}$ buckets of $\mathsf{LSH}_\ell$ (line~4).
For each bucket~$b$, it computes the two candidate \texttt{OHT} bins via the tier-1 and tier-2 hash functions $\mathsf{H}_1(b)$ and $\mathsf{H}_2(b)$ (line~5), then iterates over every entry in both bins (lines~6--7).

In \textsc{Read} mode (lines~8--10):
for each candidate entry~$e$, the oblivious equality test \texttt{OblCompare}~\cite{dauterman2021snoopy} checks whether $e.\texttt{LSHbucketID}$ matches~$b$ (line~9).
The result drives \texttt{OblChoose}~\cite{dauterman2021snoopy}, an oblivious conditional select that returns one of two candidate values while always executing the same fixed instruction structure.
Then, a match updates $e.\texttt{Neuron-Buffer}$ with $\mathsf{LSH}_\ell[b]$ (line~10).

In \textsc{Write} mode (lines~11--13), the scan tests the same bucket match and the non-dummy guard (line~12).
A matching real entry updates $\mathsf{LSH}_\ell[b]$ with its \texttt{Neuron-Buffer} (line~13).

Because the outer loop visits every LSH bucket exactly once, the single-query requirement of the \texttt{OHT}~\cite{chan2017oblivious} is satisfied.
 
\emph{Phase~3: Request-Array Restoration (lines~14--18).}
After the scan, all \texttt{OHT} bin entries (real and dummy) are collected into a flat array and sorted by a single \texttt{OblSort} pass keyed on $(\texttt{isDummy},\;\texttt{Batch-Rank})$.
This places non-dummy entries first, grouped by their originating input (\texttt{Batch-Rank}), while dummies are pushed to the tail.
The first $T$ positions constitute the restored fixed-length $\mathsf{ReqArray}_\ell$, with dummy padding when Write mode has merged duplicate requests.
In \textsc{Write} mode, the returned $\mathsf{ReqArray}_\ell$ is unused, but we return it uniformly to simplify the algorithm presentation.
The ordering of requests within the same \texttt{Batch-Rank} may differ from the original $\mathsf{ReqArray}_\ell$ produced by Algorithm~\ref{alg:neuron-requester}; this does not affect correctness, as the Neuron Updater (Section~\ref{sec:neuron-updater}) processes all requests for a given input independently of their internal order.

\subsection{Neuron Updater}
\label{app:neuron-updater}

The Neuron Updater operates on the populated \texttt{Request-Array} for the LSH-selected layer and on a public-order dense array for layer~0.
For each input $i\in[1..B]$ and layer $\ell$, it maintains two fixed-length arrays of public size $R_\ell$:
$\mathsf{ActiveNodes}[i][\ell][1..R_\ell]$, which stores neuron IDs;
and $\mathsf{ActiveVals}[i]$ $[\ell][1..R_\ell]$, which stores the corresponding activation values.
For LSH-selected layers, $R_\ell = \mathsf{lenSeq}_\ell \cdot \mathsf{PADSIZE}$; for dense layer~0, $R_0=N_0$ and slots enumerate all dense neurons in public order.
For LSH-selected layers, each slot $r$ maps to a unique, publicly determined location in $\mathsf{ReqArray}_\ell$ via the helper \textsc{SlotNode} (Algorithm~\ref{alg:slotnode}), so both feedforward and backpropagation traverse the same fixed slot layout.
Dense layer~0 feedforward and optimizer updates scan the $N_0$ dense neurons in public order.

\textbf{Obliviousness Invariant.}
All updater routines consist of fixed-length scans over $[1..R_\ell]$ (or nested scans over $[1..R_{\ell+1}]\times[1..R_\ell]$), with loop bounds determined solely by public parameters such as $B$, $\mathsf{lenSeq}_\ell$, $\mathsf{PADSIZE}$, and $N_0$.
Dummy slots carry neutral values (zero weights, zero bias, zero activations) and are processed identically to real slots.
Every data-dependent conditional is realized through branch-free oblivious primitives such as \texttt{OblChoose} and \texttt{OblGt}~\cite{dauterman2021snoopy}, so the updater reveals no private information beyond the public parameters in~$P$.

\textbf{Algorithm~\ref{alg:slotnode} (\textsc{SlotNode}).}
This helper defines the canonical addressing rule for \texttt{Request-Array}.
Recall that each input $i$ is allocated $\mathsf{lenSeq}_\ell$ consecutive entries in $\mathsf{ReqArray}_\ell$, and each entry contains a \texttt{Neuron-Buffer} of $\mathsf{PADSIZE}$ neuron records.
Given a linear slot index $r\in[1..R_\ell]$, \textsc{SlotNode} deterministically computes a probe index $t=\lfloor (r{-}1)/\mathsf{PADSIZE}\rfloor$ and an in-buffer offset $s=(r{-}1)\bmod \mathsf{PADSIZE}$ (lines~1--2), then returns the neuron record at $\mathsf{ReqArray}_\ell[(i{-}1)\cdot\mathsf{lenSeq}_\ell+t].\texttt{Neuron-Buffer}[s]$ (lines~3--4).
Because the entire mapping depends only on $i$, $r$, and public parameters, every subsequent access via \textsc{SlotNode} follows a fixed, input-independent pattern.

\begin{algorithm}[ht]
  \caption{\textsc{SlotNode}$(\mathsf{ReqArray}_\ell, i, r)$}
  \label{alg:slotnode}
  \footnotesize
  \begin{algorithmic}[1]
    \Require $\mathsf{ReqArray}_\ell[0..T-1]$, input index $i\in[1..B]$, slot $r\in[1..R_\ell]$
    \Ensure The neuron record stored at the public slot $(i,r)$
    \State $t \gets \lfloor (r-1) / \mathsf{PADSIZE} \rfloor$ \Comment{probe index in $[0..\mathsf{lenSeq}_\ell-1]$}
    \State $s \gets (r-1) \bmod \mathsf{PADSIZE}$ \Comment{slot index in buffer $[0..\mathsf{PADSIZE}-1]$}
    \State $\texttt{ind} \gets (i-1)\cdot \mathsf{lenSeq}_\ell + t$ \Comment{public, deterministic index}
    \State \Return $\mathsf{ReqArray}_\ell[\texttt{ind}].\texttt{Neuron-Buffer}[s]$
  \end{algorithmic}
\end{algorithm}

\textbf{Algorithm~\ref{alg:updater-ff} (Feed-Forward).}
This algorithm computes activations for all $B$ inputs at one layer: LSH-selected layers use \texttt{Request-Array}, while dense layer~0 uses a public-order scan over all $N_0$ neurons.
For the output layer, it also performs softmax normalization.
It proceeds in two passes per input.

\begin{algorithm}[htb]
  \caption{\textsc{Neuron Updater} -- Feed-Forward at layer $\ell$}
  \label{alg:updater-ff}
  \footnotesize
  \begin{algorithmic}[1]
    \Require
        Layer index $\ell$;
        
        activation type $\mathsf{type}_\ell \in \{\text{ReLU}, \text{Softmax}\}$ (public);
        
        populated $\mathsf{ReqArray}_\ell$ for LSH-selected layers; dense layer~0 is accessed as $\text{Layer}_0.\text{getNode}(r)$ in public order;
        
        batch size $B$;
        
        $R_\ell = \mathsf{lenSeq}_\ell \cdot \mathsf{PADSIZE}$ for LSH-selected layers, and $R_0=N_0$ for dense layer~0;
        
        previous-layer active-node IDs and values
        $\mathsf{ActiveNodes}[i][\ell{-}1][1..R_{\ell-1}]$,
        $\mathsf{ActiveVals}[i][\ell{-}1][1..R_{\ell-1}]$
        for each input $i$ (for $\ell = 0$, these are the raw input feature indices and values with $R_{\ell-1}$ set to input length);
        
        \textsc{SlotNode} (Algorithm~\ref{alg:slotnode});

        oblivious primitives~\cite{dauterman2021snoopy}
        (\S\ref{app:obl-primitive}): $\texttt{OblGt}(a,b)$ (returns $1$ iff $a > b$, data-obliviously) and \texttt{OblChoose}$(\mathsf{flag},v_1,v_0)$ (returns $v_1$ if $\mathsf{flag}=1$, else $v_0$, data-obliviously)

    \Ensure
        $\mathsf{ActiveNodes}[i][\ell][1..R_\ell]$ and $\mathsf{ActiveVals}[i][\ell][1..R_\ell]$ for all $i$;
        
        for $\mathsf{type}_\ell = \text{Softmax}$: normalization constants $\mathsf{normConst}[i]$ and each neuron's $u.\mathsf{lastActivation}[i]$ written back

    \Statex
    \For{$i = 1$ to $B$}
      \State $m_i \gets -\infty$ \Comment{for oblivious compare; used only if Softmax}

      \Statex
      \State \textit{// Pass 1: compute activations}
      \For{$r = 1$ to $R_\ell$}
        \State $u \gets \text{Layer}_0.\text{getNode}(r)$ if $\ell=0$; otherwise $u \gets \Call{SlotNode}{\mathsf{ReqArray}_\ell, i, r}$
        \State $a \gets u.\mathsf{bias}$
        \For{$r' = 1$ to $R_{\ell-1}$}
          \State $a \gets a + u.\mathsf{weights}\bigl[\mathsf{ActiveNodes}[i][\ell{-}1][r']\bigr]
                 \;\cdot\; \mathsf{ActiveVals}[i][\ell{-}1][r']$
        \EndFor
        \If{$\mathsf{type}_\ell = \text{ReLU}$}
          \State $a \gets \texttt{OblChoose}\bigl(\texttt{OblGt}(a, 0),\; a,\; 0\bigr)$
              \Comment{ReLU}
        \EndIf
        \State $\mathsf{ActiveNodes}[i][\ell][r] \gets u.\mathsf{ID}$
        \State $\mathsf{ActiveVals}[i][\ell][r] \gets a$
        \State $u.\mathsf{lastActivation}[i] \gets a$
        \If{$\mathsf{type}_\ell = \text{Softmax}$}
          \State $m_i \gets \texttt{OblChoose}\bigl(\texttt{OblGt}(a, m_i),\; a,\; m_i\bigr)$
        \EndIf
      \EndFor

      \Statex
      \State \textit{// Pass 2: softmax normalization (output layer only)}
      \If{$\mathsf{type}_\ell = \text{Softmax}$}
        \State $\mathsf{normConst}[i] \gets 0$
        \For{$r = 1$ to $R_\ell$}
          \State $u \gets \Call{SlotNode}{\mathsf{ReqArray}_\ell, i, r}$
          \State $z \gets \exp\bigl(\mathsf{ActiveVals}[i][\ell][r] - m_i\bigr)$
          \State $z \gets \texttt{OblChoose}(u.\mathsf{isDummy},\; 0,\; z)$
              \Comment{dummy slots contribute zero}
          \State $\mathsf{ActiveVals}[i][\ell][r] \gets z$
          \State $u.\mathsf{lastActivation}[i] \gets z$
          \State $\mathsf{normConst}[i] \gets \mathsf{normConst}[i] + z$
        \EndFor
        \State $\mathsf{normConst}[i] \gets$
        \Statex \hspace{\algorithmicindent}$\texttt{OblChoose}\bigl(
        \texttt{OblGt}(\mathsf{normConst}[i],\; 0),\;
        \mathsf{normConst}[i],\; 1\bigr)$
        \Comment{guard against all-dummy batch}
      \EndIf
    \EndFor

    \State \Return $\mathsf{ActiveNodes}$, $\mathsf{ActiveVals}$, $\mathsf{normConst}$
  \end{algorithmic}
\end{algorithm}

\emph{Scope: two-layer architecture.}
In \system's two-layer architecture (Figure~\ref{fig:overview}), only the output layer ($\ell = L{-}1$) is wide and uses LSH-based neuron selection. 
Layer~$0$ is a dense hidden layer whose $N_0$ neurons are all accessed for every input.
Therefore, lines~7--8 scan all $N_0$ neurons in layer~0.
The corresponding weight access is implemented as a linear scan rather than a data-dependent lookup.
Extending to multiple sparse layers would require additional oblivious indirection, such as a linear scan over $u.\mathsf{weights}$, and is left to future work.

In \emph{Pass~1} (lines~4--15), the algorithm scans all $R_\ell$ slots for input $i$.
At each slot $r$, it retrieves the neuron record $u$ via \textsc{SlotNode} for LSH-selected layers, or via public-order layer~0 access for the dense layer (line~5), and computes the pre-activation as a weighted sum over the previous layer's active state:
$a = u.\mathsf{bias} + \sum_{r'=1}^{R_{\ell-1}} u.\mathsf{weights}\bigl[\mathsf{ActiveNodes}[i]$ $[\ell-1][r']\bigr] \cdot \mathsf{ActiveVals}[i][\ell-1][r']$ (lines~6--8).
For $\ell = 0$ (the first hidden layer), the previous-layer arrays are replaced by the raw input feature indices and values.
If the layer uses ReLU activation, the pre-activation is clamped to zero via an oblivious conditional (lines~9--10).
The resulting neuron ID and activation are recorded into $\mathsf{ActiveNodes}[i][\ell][r]$ and $\mathsf{ActiveVals}[i][\ell][r]$ (lines~11--13).
For the Softmax output layer, a running oblivious maximum $m_i$ is maintained via $\texttt{OblGt}$ and $\texttt{OblChoose}$ over all slot activations (lines~14--15).

\emph{Pass~2} (lines~16--26) executes only for the output layer ($\mathsf{type}_\ell = \text{Softmax}$).
Using the oblivious maximum $m_i$ from Pass~1, the algorithm re-scans all $R_\ell$ slots and computes numerically stable exponentials $z = \exp(\mathsf{ActiveVals}[i][\ell][r] - m_i)$ (line~21).
Dummy slots are then explicitly zeroed via $z \gets \texttt{OblChoose}(u.\mathsf{isDummy},\, 0,\, z)$ (line~22), ensuring that padding neurons contribute exactly zero to the normalization constant rather than a data-dependent residual.
Each slot's activation is replaced in place by $z$ (lines~23--24), and the normalization constant $\mathsf{normConst}[i]$ accumulates the sum of these exponentials (line~25).
The actual softmax probability $p = z / \mathsf{normConst}[i]$ is deferred to the backpropagation stage (Algorithm~\ref{alg:updater-bp}), where it is computed immediately before the output-layer delta calculation.
A guard ensures $\mathsf{normConst}[i] \geq 1$ via \texttt{OblChoose}, so that an all-dummy batch (where all exponentials are zeroed) yields $p = 0$ rather than a division by zero (line~26).

Both passes consist of fixed-length scans whose bounds are determined by public parameters.
In Pass~1, dummy slots contribute zero to the weighted sum due to their zero weights and activations.
In Pass~2, dummy slots are explicitly zeroed (line~22), so they contribute exactly zero to $\mathsf{normConst}[i]$ and propagate zero delta throughout backpropagation.
All data-dependent choices (ReLU clamping, oblivious maximum, dummy zeroing) use branch-free primitives.
Hence, the feedforward pass reveals no private information beyond the public parameters (including the public sparsity pattern).

\textbf{Algorithm~\ref{alg:updater-bp} (Backpropagation).}
This algorithm computes gradient accumulators for the LSH-selected layer's \texttt{Request-Array} and for dense layer~0.
It proceeds in three phases, mirroring standard backpropagation over the fixed LSH-selected slot layout and the public-order dense layer~0 array.

\begin{algorithm*}[htb]
  \caption{\textsc{Neuron Updater} -- Backpropagation}
  \label{alg:updater-bp}
  \footnotesize
  \begin{algorithmic}[1]
    \Require
        $\mathsf{ReqArray}_\ell$ for each LSH-selected layer $\ell$; dense layer~0 is stored as a public-order array of $N_0$ neurons;
        
        $\mathsf{ActiveVals}[i][\ell][1..R_\ell]$ and $\mathsf{normConst}[i]$ from feedforward (Algorithm~\ref{alg:updater-ff});
        
        number of layers $L$, batch size $B$;
        
        $R_\ell = \mathsf{lenSeq}_\ell \cdot \mathsf{PADSIZE}$ for each LSH-selected layer $\ell$, and $R_0=N_0$ for dense layer~0;

        labels $Y_i \subseteq [0..C-1]$ with sizes $|Y_i|$ for each input $i$ in a batch;

        dummy neuron IDs are public sentinel values outside $[0..C-1]$;
        
        $N_0$: number of neurons in layer~$0$ (size of each layer, publicly known);

        sparse training inputs $\mathbf{x}_i$, each containing $\mathsf{nnz}$ non-zero features out of $D$ total dimensions, represented as (index, value) pairs $(d_j, x_{d_j})$ for $j \in [1..\mathsf{nnz}]$; both $\mathsf{nnz}$ and the non-zero indices $\{d_1,\ldots,d_{\mathsf{nnz}}\}$ are public, while the feature values $x_{d_j}$ are private;
        
        \textsc{SlotNode} (Algorithm~\ref{alg:slotnode});
        
        oblivious primitives~\cite{dauterman2021snoopy}
        (\S\ref{app:obl-primitive}):
        $\texttt{OblCompare}(a,b)$ (returns $1$ iff $a = b$),
        $\texttt{OblGt}(a,b)$ (returns $1$ iff $a > b$),
        $\texttt{OblChoose}(\mathsf{flag},v_1,v_0)$ (returns $v_1$ if $\mathsf{flag}=1$, else $v_0$) 

    \Ensure Per-neuron gradient accumulators updated for all neurons in each LSH-selected $\mathsf{ReqArray}_\ell$ and in dense layer~0

    \Statex
    \State \textit{// Phase 1: Output-layer delta computation}
    \State $\ell \gets L-1$
    \For{$i = 1$ to $B$}
      \For{$r = 1$ to $R_\ell$}
        \State $u \gets \Call{SlotNode}{\mathsf{ReqArray}_\ell, i, r}$
        \State $p \gets u.\mathsf{lastActivation}[i] \;/\; \mathsf{normConst}[i]$
            \Comment{softmax probability}
        \State $u.\mathsf{lastActivation}[i] \gets p$
        \State $\mathsf{isLabel} \gets 0$
        \For{each $y \in Y_i$}
          \State $\mathsf{isLabel} \gets \texttt{OblChoose}\bigl(\texttt{OblCompare}(y,\; u.\mathsf{ID}),\; 1,\; \mathsf{isLabel}\bigr)$
        \EndFor
        \State $u.\mathsf{delta}[i] \gets \texttt{OblChoose}\bigl(\mathsf{isLabel},\; (1/|Y_i| - p)/B,\; {-p}/{B}\bigr)$
      \EndFor
    \EndFor

    \Statex
    \State Reset dense-layer delta scratch state by a fixed public scan over $[1..B]\times[1..N_0]$

    \State \textit{// Phase 2: Layer-by-layer propagation and gradient accumulation}
    \For{$\ell = L{-}1$ down to $1$}
      \For{$i = 1$ to $B$}
        \For{$r_u = 1$ to $R_\ell$}
          \State $u \gets \Call{SlotNode}{\mathsf{ReqArray}_\ell, i, r_u}$
          \State $\delta_u \gets u.\mathsf{delta}[i]$

          \Statex
          \State \textit{// Inner loop: propagate delta to layer $\ell{-}1$ and accumulate weight gradients}
          \If{$\ell > 1$}
              \For{$r_v = 1$ to $R_{\ell-1}$} \Comment{iterate active neurons in sparse layer}
                \State $v \gets \Call{SlotNode}{\mathsf{ReqArray}_{\ell-1}, i, r_v}$
                \State $\mathsf{inc} \gets \delta_u \cdot u.\mathsf{weights}[v.\mathsf{ID}]$
                \State $\mathsf{isAct} \gets \texttt{OblGt}\bigl(v.\mathsf{lastActivation}[i],\; 0\bigr)$
                \State $v.\mathsf{delta}[i] \gets \texttt{OblChoose}\bigl(\mathsf{isAct},\; v.\mathsf{delta}[i] + \mathsf{inc},\; v.\mathsf{delta}[i]\bigr)$
                \State $u.\mathsf{t}[v.\mathsf{ID}] \gets u.\mathsf{t}[v.\mathsf{ID}] + \delta_u \cdot v.\mathsf{lastActivation}[i]$
              \EndFor
        \Else
              \For{$n = 1$ to $N_0$} \Comment{iterate dense layer}
                \State $v \gets \text{Layer}_0.\text{getNode}(n)$
                \State $\mathsf{inc} \gets \delta_u \cdot u.\mathsf{weights}[v.\mathsf{ID}]$
                \State $\mathsf{isAct} \gets \texttt{OblGt}\bigl(v.\mathsf{lastActivation}[i],\; 0\bigr)$
                \State $v.\mathsf{delta}[i] \gets \texttt{OblChoose}\bigl(\mathsf{isAct},\; v.\mathsf{delta}[i] + \mathsf{inc},\; v.\mathsf{delta}[i]\bigr)$
                \State $u.\mathsf{t}[v.\mathsf{ID}] \gets u.\mathsf{t}[v.\mathsf{ID}] + \delta_u \cdot v.\mathsf{lastActivation}[i]$
              \EndFor
        \EndIf

          \Statex
          \State \textit{// Bias gradient: once per $(u, i)$, outside the inner loop}
          \State $u.\mathsf{tbias} \gets u.\mathsf{tbias} + \delta_u$
        \EndFor
      \EndFor
    \EndFor

    \Statex
    \State \textit{// Phase 3: First-layer gradient from input features}
    \For{$i = 1$ to $B$}
      \For{$n = 1$ to $N_0$}
        \State $v \gets \text{Layer}_0.\text{getNode}(n)$
        \State $\delta_v \gets v.\mathsf{delta}[i]$
        \For{$j = 1$ to $\mathsf{nnz}$}
           \Statex \Comment{$(d_j, x_{d_j})$: index and value of the $j$-th non-zero feature of $\mathbf{x}_i$}
          \State $v.\mathsf{t}[d_j] \gets v.\mathsf{t}[d_j] + \delta_v \cdot x_{d_j}$
        \EndFor
        \State $v.\mathsf{tbias} \gets v.\mathsf{tbias} + \delta_v$
      \EndFor
    \EndFor
  \end{algorithmic}
\end{algorithm*}

\emph{Scope: two-layer architecture.}
In \system's two-layer architecture (Figure~\ref{fig:overview}), only the output layer ($\ell = L{-}1$) is wide and uses LSH-based neuron selection. 
Layer~$0$ is a dense hidden layer whose $N_0$ neurons are all accessed for every input.
This design is critical for obliviousness of the Neuron Updater: the backpropagation loop always follows lines~26--32 rather than the sparse-layer branch in lines~19--25.
In the inner loop of Phase~2 (lines~26--32), the weight access $u.\mathsf{weights}[v.\mathsf{ID}]$ iterates over all $N_0$ neurons of the dense layer~$0$, so the access pattern is independent of the input data.
If the network has more than two layers and multiple layers use LSH, then the sparse-layer branch in lines~19--25 can make $v.\mathsf{ID}$ depend on which neurons were selected for each input, making the weight access pattern data-dependent.
Extending \system to support additional sparse layers would require an oblivious indirection mechanism (e.g., an additional linear scan per weight access) and is left to future work.
We emphasize, however, that the output layer, which is responsible for an extremely large number of labels, dominates the computational cost.
Therefore, the cost of additional oblivious operations for multiple-layer extension is manageable.

\emph{Phase~1: Output-layer deltas} (lines~1--11).
For each input $i$ and slot $r$ at the output layer $\ell = L-1$, the algorithm retrieves the neuron record $u$ and computes its softmax probability $p$ (line~6).
Here $p = u.\mathsf{lastActivation}[i] / \mathsf{normConst}[i]$, where $u.\mathsf{lastActivation}[i]$ holds the exponential computed during feedforward Pass~2.
The algorithm then determines whether $u$ corresponds to a ground-truth label in $Y_i$ via an oblivious scan over $Y_i$ using \texttt{OblCompare} (lines~9--10).
The output delta is set to $(1/|Y_i| - p)/B$ for positive labels and $-p/B$ for negative labels, selected obliviously via \texttt{OblChoose} (line~11).
Note that $|Y_i|$ is a public parameter in the multi-label classification setting.
For dummy slots, since their exponential was explicitly zeroed during feedforward (Algorithm~\ref{alg:updater-ff}, line~22), we have $p = 0$.
This yields a delta of $0$ because the dummy ID matches no label, so dummy slots contribute zero in all subsequent phases.

\emph{Phase~2: Layer-by-layer propagation} (lines~12--35).
For each layer $\ell$ from $L{-}1$ down to $1$, and for each input $i$, the algorithm scans all $R_\ell$ slots of layer $\ell$.
For each neuron $u$ in layer $\ell$ with accumulated delta $\delta_u = u.\mathsf{delta}[i]$ (line~18), a nested loop over all $R_{\ell-1}$ slots of the previous layer performs two operations per pair $(u, v)$:

\emph{(1) Delta propagation} (lines~22--25): the error signal is added to $v.\mathsf{delta}[i]$, gated by the ReLU derivative of $v$.
The signal is $\delta_u$ times the weight from $v$ to $u$.
The increment is applied only if $v.\mathsf{lastActivation}[i] > 0$, i.e., the neuron was not killed by ReLU during feedforward.
This condition is checked via \texttt{OblGt} and written via \texttt{OblChoose}.
For dummy neurons, whose activation is zero, \texttt{OblGt}$(0, 0)$ returns $0$, so no delta is propagated into dummy slots.

\emph{(2) Weight gradient accumulation} (line~26): the gradient for the edge from $v$ to $u$ is $\delta_u \cdot v.\mathsf{lastActivation}[i]$ and is accumulated into $u.\mathsf{t}[v.\mathsf{ID}]$.
Since killed or dummy neurons have $\mathsf{lastActivation} = 0$, their contribution is automatically zero.

After the inner loop completes, the bias gradient $\delta_u$ is accumulated into $u.\mathsf{tbias}$ \emph{exactly once} per $(u, i)$ pair (line~35).
This placement is critical: accumulating the bias gradient inside the inner loop would multiply it by $R_{\ell-1}$, the number of slots in the previous layer.

Recall that in our two-layer architecture, Phase 2 executes only for $\ell = 1$, so the inner loop always iterates over the $N_0$ neurons of the dense layer 0. The sparse-layer branch (lines 20--26) is included for generality.

When $\ell = 1$, the previous layer (layer~$0$) is a dense layer that does not use \texttt{Request-Array}.
In this case, the inner loop iterates over all $N_0$ neurons of layer~$0$ directly rather than through \textsc{SlotNode}.

\emph{Phase~3: First-layer gradients} (lines~36--43).
For each input $i$, the algorithm iterates over all $N_0$ neurons of layer~$0$.
Each neuron $v$ accumulates weight gradients from the input features of $\mathbf{x}_i$: for each feature $(d, x_d)$, the gradient $\delta_v \cdot x_d$ is added to $v.\mathsf{t}[d]$ (line~42).
The bias gradient $\delta_v$ is accumulated once per $(v, i)$ pair (line~43).

\textbf{Remark on First-layer Sparsity.}
In Phase~3, the gradient update $v.\mathsf{t}[d_j] \mathrel{+}= \delta_v \cdot x_{d_j}$ accesses weight dimension $d_j$, which depends on the sparsity pattern of the input $\mathbf{x}_i$.
Consequently, the adversary can learn the set of non-zero feature indices of each training input, but not the corresponding feature values, weights, gradients, or any other intermediate quantities.
Fully hiding the sparsity pattern is possible (e.g.,~via an oblivious scatter based on oblivious sorting) but incurs $O(N_0 \cdot D \log D)$ cost per batch, where $D$ is the input dimensionality.
For high-dimensional sparse datasets this overhead dominates the training cost; we therefore treat the sparsity pattern as public information (see~\S\ref{app:proof}) and leave a more efficient oblivious treatment to future work.

\textbf{Obliviousness.}
In Phase~1, the scan over $R_\ell$ slots and the label-matching loop over $|Y_i|$ both have publicly determined bounds; $|Y_i|$ is a public parameter in the multi-label setting.
In Phase~2, the nested loop bounds $R_\ell \times R_{\ell-1}$ depend only on public parameters, and every slot pair executes the same sequence of \texttt{OblGt}, \texttt{OblChoose}, and arithmetic operations regardless of match outcome.
Phase~3 operates on the dense layer~$0$; the outer loop over $N_0$ neurons and the inner loop over $\mathsf{nnz}$ non-zero features both have publicly known bounds.
The gradient accumulators reside within the corresponding \texttt{Request-Array} records for the LSH-selected layer and dense layer~0 records for first-layer gradients, and are written at locations determined by public loop indices and the public sparsity pattern.
Hence, the backpropagation pass reveals no private information beyond the public sparsity pattern included in~$P$.

\textbf{Oblivious Optimizer}
After backpropagation completes, the gradient accumulators 
$u.\mathsf{t}[\cdot]$ and $u.\mathsf{tbias}$ are distributed across multiple request entries in $\mathsf{ReqArray}_\ell$ that may reference the same LSH bucket.
Each such entry stores a copy of that bucket in its \texttt{Neuron-Buffer}, so entries with the same \texttt{LSHbucketID} contain slot-aligned copies of the same padded neuron set.
These copies are first consolidated by \textsc{OblMergeRequests} (Appendix~\ref{app:obl-merge}), which obliviously groups entries by \texttt{LSHbucketID} and folds their gradient state slot-by-slot into a single surviving representative.
The merged accumulators are then consumed by the oblivious optimizer (Appendix~\ref{app:obl-optimizer}), which applies the Adam optimizer to each neuron before write-back to the LSH table.

\subsection{Oblivious Neuron Merge}
\label{app:obl-merge}

After backpropagation (Algorithm~\ref{alg:updater-bp}), multiple entries in $\mathsf{ReqArray}_\ell$ may share the same \texttt{LSHbucketID}.
Such entries correspond to different requests (inputs) for the same LSH bucket and therefore contain slot-aligned copies of the same padded neuron set in their \texttt{Neuron-Buffer}s.
Each copy carries gradient accumulators ($u.\mathsf{t}[\cdot]$ and $u.\mathsf{tbias}$) contributed by its respective input.
Before the parameter update (Appendix~\ref{app:obl-optimizer}) and write-back to the LSH table, these bucket copies must be consolidated so that, for each slot position in the \texttt{Neuron-Buffer}, the surviving copy holds the sum of all per-input gradients associated with the neuron stored in that slot.

\begin{algorithm}[t]
  \caption{\textsc{OblMergeRequests}$(\mathsf{ReqArray}_\ell)$}
  \label{alg:obl-merge}
  \footnotesize
  \begin{algorithmic}[1]
    \Require
        $\mathsf{ReqArray}_\ell[0..T{-}1]$ with $T = B \cdot \mathsf{lenSeq}_\ell$, where each entry stores \texttt{LSHbucketID}, \texttt{isDummy} flag, and a \texttt{Neuron-Buffer}$[0..\mathsf{PADSIZE}{-}1]$; 

        $D_\ell$: weight dimension at layer $\ell$;
        
        each neuron record carries gradient accumulators $u.\mathsf{t}[0..D_\ell-1]$ and $u.\mathsf{tbias}$ populated by backpropagation (Algorithm~\ref{alg:updater-bp});

        oblivious primitives~\cite{dauterman2021snoopy}:
        $\texttt{OblCompare}(a,b)$, $\texttt{OblChoose}(\mathsf{flag},v_1,v_0)$;
        $\texttt{OblSort}$ (oblivious sorting, e.g.,~bitonic sort~\cite{batcher1968sorting})

    \Ensure Merged $\mathsf{ReqArray}_\ell$ with at most one non-dummy entry per \texttt{LSHbucketID}; for each surviving entry, the gradient accumulators in every \texttt{Neuron-Buffer} slot equal the slot-wise sum over all original copies of that LSH bucket

    \Statex \textit{// group request entries that reference the same LSH bucket}
    \State $\Call{OblSort}{\mathsf{ReqArray}_\ell,\ \texttt{key}=\texttt{LSHbucketID}}$

    \For{$j = T-1$ down to $1$}
      \State $\mathsf{same} \gets $ 
      \Statex \quad $ \Call{OblCompare}{\mathsf{ReqArray}_\ell[j].\texttt{LSHbucketID},\ \mathsf{ReqArray}_\ell[j{-}1].\texttt{LSHbucketID}}$
      \Statex \textit{// slot-wise fold of gradient state from copy $j$ into aligned copy $j-1$}
      \State \Call{FoldIntoPrev}{$\mathsf{same},\ \mathsf{ReqArray}_\ell[j],\ \mathsf{ReqArray}_\ell[j{-}1]$}
      \State \Call{MarkDummy}{$\mathsf{same},\ \mathsf{ReqArray}_\ell[j]$}
          \Comment{set $\texttt{isDummy} \gets 1$ if $\mathsf{same}$}
    \EndFor

    \State $\Call{OblSort}{\mathsf{ReqArray}_\ell,\ \texttt{key}=\texttt{isDummy}}$
    \Comment{push dummies to end; array length unchanged}

    \State \Return $\mathsf{ReqArray}_\ell$
  \end{algorithmic}
\end{algorithm}

Algorithm~\ref{alg:obl-merge} performs this consolidation in three steps.
First, an oblivious sort (e.g.,~bitonic sort~\cite{batcher1968sorting}) groups request entries by \texttt{LSHbucketID} (line~1), making all request entries for the same bucket contiguous.
This step does not sort by neuron ID.

Second, a backward linear scan from $j = T-1$ down to $1$ (lines~2--5) processes each pair of adjacent entries.
For each $j$, \texttt{OblCompare}~\cite{dauterman2021snoopy} tests whether entries $j$ and $j-1$ share the same \texttt{LSHbucketID} (line~3).
If so, \textsc{FoldIntoPrev} obliviously accumulates the gradient state of entry $j$ into entry $j-1$ slot-by-slot: for each neuron slot $s \in [0..\mathsf{PADSIZE}{-}1]$ and each weight dimension $d \in [0..D_\ell{-}1]$, it adds $\mathsf{ReqArray}_\ell[j].\texttt{Neuron-Buffer}[s].\mathsf{t}[d]$ into the corresponding field of entry $j{-}1$, and likewise for $\mathsf{tbias}$ (line~4).
This slot-wise fold is correct because entries with the same \texttt{LSHbucketID} contain aligned copies of the same bucket, so the same neuron appears in the same slot across those copies.
The donor entry $j$ is then marked as dummy via \textsc{MarkDummy} (line~5).
Both \textsc{FoldIntoPrev} and \textsc{MarkDummy} use \texttt{OblChoose} to ensure branch-free execution regardless of the match outcome.
Because the scan proceeds backward, gradients from a group of $k$ entries with the same \texttt{LSHbucketID} are successively folded into the first entry of the group, which retains the complete aggregated gradient state for that bucket copy.

Third, a final oblivious sort by \texttt{isDummy} (line~6) pushes all marked-dummy entries to the end while keeping the array length $T$ unchanged, producing a compact prefix of unique bucket entries followed by dummy padding.

\textbf{Obliviousness.}
All three steps operate on fixed-length arrays with publicly determined sizes.
The two oblivious sorts and the linear scan each produce access patterns that depend only on $T$.
All conditional operations within the scan (\texttt{OblCompare}, \texttt{OblChoose} inside \textsc{FoldIntoPrev} and \textsc{MarkDummy}) are branch-free.
Hence, the merge reveals nothing beyond public parameters.

\subsection{Oblivious Optimizer}
\label{app:obl-optimizer}
\textbf{Algorithm~\ref{alg:adam-update} (Adam Update).}
After gradient accumulation (Algorithm~\ref{alg:updater-bp}) and oblivious merge (Algorithm~\ref{alg:obl-merge}), each neuron record in $\mathsf{ReqArray}_\ell$ holds a single consolidated gradient in $u.\mathsf{t}[\cdot]$ and $u.\mathsf{tbias}$, aggregated across all inputs in the batch that requested that neuron.
Algorithm~\ref{alg:adam-update} applies the Adam optimizer~\cite{kingma2014adam} to update every neuron's weights and bias before write-back to the LSH table.
For dense layer~0, the same Adam rule is applied by a separate fixed public-order scan over all $N_0$ dense neurons.
This consumes the layer-0 gradients accumulated in Phase~3 of Algorithm~\ref{alg:updater-bp}.

The algorithm first computes a bias-corrected learning rate $\tilde{\alpha} = \alpha \cdot \sqrt{1 - \beta_2^{\,\mathsf{step}}} \;/\; (1 - \beta_1^{\,\mathsf{step}})$ (line~1), absorbing the standard Adam bias-correction terms into the learning rate rather than correcting the moment estimates directly.
This is mathematically equivalent to the standard formulation~\cite{kingma2014adam} and avoids two additional divisions per weight per neuron.

The main body (lines~2--15) performs a flat linear scan over all $T \cdot \mathsf{PADSIZE}$ neuron slots in $\mathsf{ReqArray}_\ell$.
For each neuron $u$, it iterates over all $D_\ell$ weight dimensions: the first-moment estimate $u.\mathsf{m}[d]$ is updated with the gradient $u.\mathsf{t}[d]$ (line~6), the second-moment estimate $u.\mathsf{v}[d]$ is updated with the squared gradient (line~7), and the weight is adjusted accordingly (line~8).
The gradient accumulator $u.\mathsf{t}[d]$ is then reset to zero (line~9) in preparation for the next batch.
The same update is applied to the bias terms $u.\mathsf{mBias}$, $u.\mathsf{vBias}$, and $u.\mathsf{bias}$ (lines~11--14).

For dummy neurons, the gradient accumulators $u.\mathsf{t}[\cdot]$ and $u.\mathsf{tbias}$ are zero throughout training (as established in the backpropagation analysis).
Consequently, the moment estimates remain at their initialized values of zero, and the weight update $\tilde{\alpha} \cdot 0 / (0 + \epsilon) = 0$ leaves dummy weights unchanged.
Dummy neurons are thus processed identically to real neurons without branching.

\textbf{Obliviousness.}
The loop bounds $T$, $\mathsf{PADSIZE}$, and $D_\ell$ are all public parameters.
Every neuron slot executes the identical sequence of arithmetic operations regardless of whether it is real or dummy.
The bias-corrected learning rate $\tilde{\alpha}$ depends only on the public training step $\mathsf{step}$ and public hyperparameters.
Hence, the Adam update follows an input-independent access pattern.

\begin{algorithm}[ht]
  \caption{\textsc{Adam Update for layer $\ell$}}
  \label{alg:adam-update}
  \footnotesize
  \begin{algorithmic}[1]
    \Require
        Merged $\mathsf{ReqArray}_\ell[0..T-1]$ from \textsc{OblMergeRequests} (Algorithm~\ref{alg:obl-merge}), where $T = B \cdot \mathsf{lenSeq}_\ell$;
        
        each entry contains a \texttt{Neuron-Buffer} of $\mathsf{PADSIZE}$ neuron records;

        $D_\ell$: weight dimension at layer $\ell$;
        
        each neuron record $u$ stores: weights $u.\mathsf{weights}[0..D_\ell{-}1]$, bias $u.\mathsf{bias}$,
        first-moment estimates $u.\mathsf{m}[0..D_\ell-1]$ and $u.\mathsf{mBias}$,
        second-moment estimates $u.\mathsf{v}[0..D_\ell-1]$ and $u.\mathsf{vBias}$,
        and gradient accumulators $u.\mathsf{t}[0..D_\ell-1]$ and $u.\mathsf{tbias}$ populated by backpropagation (Algorithm~\ref{alg:updater-bp}) and merge;

        hyperparameters: base learning rate $\alpha$, decay rates $\beta_1, \beta_2$, stability constant $\epsilon$;
        global training step $\mathsf{step}$ (public)

    \Ensure All neuron weights and biases updated; moment estimates updated; gradient accumulators reset to zero

    \Statex
    \State $\tilde{\alpha} \gets \alpha \cdot \sqrt{1 - \beta_2^{\,\mathsf{step}}} \;/\; (1 - \beta_1^{\,\mathsf{step}})$
        \Comment{bias-corrected learning rate}

    \Statex
    \For{$k = 0$ to $T - 1$}
      \For{$s = 0$ to $\mathsf{PADSIZE} - 1$}
        \State $u \gets \mathsf{ReqArray}_\ell[k].\texttt{Neuron-Buffer}[s]$

        \Statex
        \State \textit{// Update weights}
        \For{$d = 0$ to $D_\ell - 1$}
          \State $u.\mathsf{m}[d] \gets \beta_1 \cdot u.\mathsf{m}[d] + (1 - \beta_1) \cdot u.\mathsf{t}[d]$
          \State $u.\mathsf{v}[d] \gets \beta_2 \cdot u.\mathsf{v}[d] + (1 - \beta_2) \cdot u.\mathsf{t}[d]^{\,2}$
          \State $u.\mathsf{weights}[d] \gets u.\mathsf{weights}[d] + \tilde{\alpha} \cdot u.\mathsf{m}[d] \;/\; (\sqrt{u.\mathsf{v}[d]} + \epsilon)$
          \State $u.\mathsf{t}[d] \gets 0$
              \Comment{reset gradient accumulator}
        \EndFor

        \Statex
        \State \textit{// Update bias}
        \State $u.\mathsf{mBias} \gets \beta_1 \cdot u.\mathsf{mBias} + (1 - \beta_1) \cdot u.\mathsf{tbias}$
        \State $u.\mathsf{vBias} \gets \beta_2 \cdot u.\mathsf{vBias} + (1 - \beta_2) \cdot u.\mathsf{tbias}^{\,2}$
        \State $u.\mathsf{bias} \gets u.\mathsf{bias} + \tilde{\alpha} \cdot u.\mathsf{mBias} \;/\; (\sqrt{u.\mathsf{vBias}} + \epsilon)$
        \State $u.\mathsf{tbias} \gets 0$
            \Comment{reset gradient accumulator}
      \EndFor
    \EndFor
  \end{algorithmic}
\end{algorithm}

\subsection{Oblivious LSH Refresh}
\label{app:lsh-refresh}

The LSH table maps each neuron to an LSH bucket determined by its WTA hash signature.
As neuron weights change during training, these signatures may shift, so the table must be periodically refreshed to maintain hash consistency.
Algorithm~\ref{alg:lsh-refresh-obl} describes the oblivious refresh procedure.

The refresh operates on a flat array~$A$ that contains all entries currently stored in~$\mathsf{LSHTable}_\ell$.
During initialization ($\mathsf{isInit}=1$), the algorithm first flattens the initial real-neuron records into $A$ and then appends $\PADSIZE$ dummies per bucket, each with a public target bucket label (lines~1--3).
The subsequent sort relocates both dummy and real entries to their target buckets.
In subsequent refreshes, $A$ already contains the full padded table (including the overflow area described below), so no new dummies are added.
 
The procedure has four steps.
First, every entry in~$A$ receives a fresh bucket assignment $\texttt{newID}$ by recomputing its WTA hash (lines~4--6).
Second, an oblivious sort groups entries by bucket and dummy flag, placing real entries before dummies within each bucket (line~7).
Third, \textsc{EnforceCapacity} (Algorithm~\ref{alg:enforce-capacity}) marks every entry beyond position~$\PADSIZE$ within its bucket as overflow (line~8).
Fourth, a second oblivious sort packs all non-overflow entries before overflow entries while preserving bucket order (line~9).
The main prefix of~$A$ is written back to $\mathsf{LSHTable}_\ell$; the remaining positions form the overflow area (line~10--11).
 
\textbf{Overflow Area.}
After rehashing, buckets may receive fewer than, exactly, or more than $\PADSIZE$ real neurons.
The dummy reservoir fills unused main-region slots, and any excess entries are stored in a dedicated overflow area appended to~$\mathsf{LSHTable}_\ell$.
 
This design ensures that the size of~$A$ and the number of entries written back are both fixed and publicly determined, so the refresh procedure's memory-access pattern reveals no information about the distribution of neurons across buckets.
$\PADSIZE$ is chosen conservatively based on the dataset, and any dataset-dependent selection rule is treated as public.

\textbf{Remark on Overflow and Security.} 
We emphasize that an attacker \emph{cannot} distinguish whether the content of an overflow slot is a real or a dummy neuron. 
The refresh algorithm consists of fixed-length oblivious sorting, a fixed linear scan (Algorithm~\ref{alg:enforce-capacity}), and write-backs of size based on public parameters, none of which depend on the nature (real or dummy) of the content in the overflow region. 

Moreover, the main region (LSH table) and overflow area have fixed sizes, and the refresh procedure always writes the same number of entries to each region regardless of whether real-neuron overflow occurs. 
Consequently, real neurons in the overflow region do not change the observable access trace; 
the security guarantee of Theorem~\ref{thm:system-final} therefore holds whether the overflow region contains only dummy entries or also real neurons.

What overflow affects is the speed of the training, not leakage: an overflowing neuron is excluded from the current training iteration, i.e., not updated, and re-enters the candidate pool at the next refresh.

\textbf{Choice of \texttt{PADSIZE} and Overflow Handling.}
We model the bucket capacity \texttt{PADSIZE} as a public parameter that we choose empirically. 
Notice that in order to derive a theoretical analysis on the maximum bucket load, one has to assume knowledge of the distribution of weights of neurons during the training process. 
Recall here that in traditional hashing, the item hashes to a location randomly; therefore, the analysis can be done independently of the hashed item's value, but this is not the case with LSH since the hashed location depends on the value.
To the best of our knowledge, there is no universal and commonly agreed behavior of the values of neurons during training; therefore, theoretical analysis is not meaningful.
This is why we followed an empirical approach, matching the bucket capacity used by SLIDE\cite{chen2020slide}, and set \texttt{PADSIZE}$ = 128$ for all datasets.

Exceeding \texttt{PADSIZE} does not permanently drop neurons. 
When a bucket receives more than \texttt{PADSIZE} real neurons after a rehash, \texttt{OblRefreshLSH} (see Algorithm \ref{alg:lsh-refresh-obl}) marks the excess entries as \texttt{isOverflow} and stores them in a dedicated overflow region appended to the LSH table. 
An overflowing neuron is excluded from training only until the next refresh. 
At the next refresh point,  the neuron is written back into the candidate array, rehashed together with every other neuron, and competes for a main-region slot (i.e., inside the LSH table) on equal terms. 
No neuron is permanently evicted from training.

The residual effect on accuracy is therefore governed by how often, and for how long, a neuron is evicted. 
It is worth noting that training spans many epochs and triggers many refreshes at fixed intervals, so a neuron evicted during one refresh is reconsidered shortly afterward. 
Additionally, adaptive sparsification is approximate by construction: the active set produced by LSH is already only an estimate of the truly top-activated neurons (see~\S\ref{sec:SLIDE}), and training is robust to precisely this class of perturbation.

\textbf{Obliviousness.}
The WTA hash computation (line~5) processes every entry identically.
Both oblivious sorts (lines~7 and~9) depend only on the public array length~$|A|$.
\textsc{EnforceCapacity} (Algorithm~\ref{alg:enforce-capacity}) performs a single linear scan over~$A$ using branch-free primitives (\texttt{OblCompare}, \texttt{OblChoose}, \texttt{OblGt}), with access pattern determined solely by~$|A|$.
The final write-back copies the main prefix and overflow suffix to fixed regions.
Hence, the entire refresh procedure follows an input-independent access pattern.

\begin{algorithm}[ht]
  \caption{\textsc{OblRefreshLSH}$(\mathsf{LSHTable}_\ell)$}
  \label{alg:lsh-refresh-obl}
  \footnotesize
  \begin{algorithmic}[1]
    \Require
        LSH table state $\mathsf{LSHTable}_\ell$ with a main region of $\mathsf{numLshBucs}$ buckets, each of fixed capacity $\mathsf{PADSIZE}$,
        and a public-size overflow region;

        an initialization flag $\mathsf{isInit}\in\{0,1\}$, where $\mathsf{isInit}=1$ iff this is the initial construction of $\mathsf{LSHTable}_\ell$;
        
        WTA hash functions $h_1, \ldots, h_K$ for layer $\ell$;
        
        oblivious primitives~\cite{dauterman2021snoopy}:
        $\texttt{OblCompare}(a,b)$, $\texttt{OblChoose}(\mathsf{flag},v_1,v_0)$;
        $\texttt{OblSort}$ (oblivious sorting, e.g.,~bitonic sort~\cite{batcher1968sorting})

    \Ensure Refreshed $\mathsf{LSHTable}_\ell$ with updated main and overflow regions

    \Statex {//$\mathsf{A}$ contains all entries in the main region plus the overflow region}
    \State $\mathsf{A} \gets$ flatten all entries from the main and overflow regions of $\mathsf{LSHTable}_\ell$

    \If{$\mathsf{isInit}=1$} \Comment{add dummies only during initialization}
        \State Append $\mathsf{numLshBucs} \times \mathsf{PADSIZE}$ fresh dummies to $\mathsf{A}$, with $\mathsf{PADSIZE}$ dummies per public target bucket
    \EndIf

    \For{$j = 0$ to $|\mathsf{A}| - 1$}
      \State $\mathsf{candID} \gets \Call{WTA}{\mathsf{A}[j]}$
      \State $\mathsf{A}[j].\texttt{newID} \gets \Call{OblChoose}{\mathsf{A}[j].\texttt{isDummy},\; \mathsf{A}[j].\texttt{targetID},\; \mathsf{candID}}$
          \Comment{real entries use WTA; dummies keep public targets}
    \EndFor

    \State $\Call{OblSort}{\mathsf{A},\ \texttt{key} = (\texttt{newID},\, \texttt{isDummy})}$
        \Comment{group by bucket; reals before dummies}

    \State \Call{EnforceCapacity}{$\mathsf{A},\, \mathsf{PADSIZE}$}
        \Comment{linear scan; mark entries beyond $\mathsf{PADSIZE}$ per bucket as \texttt{isOverflow}}

    \State $\Call{OblSort}{\mathsf{A},\ \texttt{key} = (\texttt{isOverflow},\, \texttt{newID})}$
        \Comment{pack valid entries before the overflow area}

    \State Write $\mathsf{A}[0..\mathsf{numLshBucs} \times \mathsf{PADSIZE} - 1]$ back to $\mathsf{LSHTable}_\ell$ sequentially
    \State Write the remaining entries of $\mathsf{A}$ back to the overflow region sequentially
  \end{algorithmic}
\end{algorithm}

\textbf{EnforceCapacity (Algorithm~\ref{alg:enforce-capacity}).}
After the first oblivious sort groups entries by \texttt{newID} with real entries preceding dummies within each bucket, \textsc{EnforceCapacity} performs a single linear scan to detect overflow.
It maintains a bucket-local position counter over all entries seen in the current bucket; when a new bucket is encountered, the counter resets to zero via \texttt{OblChoose}.
Any entry whose bucket-local position exceeds $\mathsf{PADSIZE}$ is flagged \texttt{isOverflow = 1}; all other entries retain \texttt{isOverflow = 0}.
Because real entries precede dummies, the theorem's no-real-overflow condition implies that all overflow entries in the covered executions are dummy entries.
The scan touches every position in $A$ exactly once with branch-free primitives (\texttt{OblCompare}, \texttt{OblChoose}, \texttt{OblGt}), so the access pattern depends only on $|A|$, which is a public parameter.
Overflow entries are stored in a dedicated region at the tail of $\mathsf{LSHTable}_\ell$ and are excluded from the current training iteration.
They are reassigned to valid bucket positions during the next invocation of \textsc{OblRefreshLSH}.
$\mathsf{PADSIZE}$ is chosen conservatively based on the dataset.

\begin{algorithm}[htb]
  \caption{\textsc{EnforceCapacity}$(A,\, \mathsf{PADSIZE})$}
  \label{alg:enforce-capacity}
  \footnotesize
  \begin{algorithmic}[1]
    \Require
        Array $A$ sorted by $(\texttt{newID},\, \texttt{isDummy})$
        (real entries before dummies within each bucket);
        bucket capacity $\mathsf{PADSIZE}$

    \Ensure
        Each bucket retains at most $\mathsf{PADSIZE}$ entries;
        excess real and dummy entries are marked \texttt{isOverflow}

    \State $\mathsf{curBucket} \gets A[0].\texttt{newID}$
    \State $\mathsf{count} \gets 0$

    \For{$j = 0$ to $|A| - 1$}
      \State $\mathsf{bucketChanged} \gets \neg \Call{OblCompare}{A[j].\texttt{newID},\; \mathsf{curBucket}} $
      \State $\mathsf{count} \gets \Call{OblChoose}{\mathsf{bucketChanged},\; 0,\; \mathsf{count}}$
      \State $\mathsf{curBucket} \gets A[j].\texttt{newID}$
      \State $\mathsf{count} \gets \mathsf{count} + 1$
      \State $\mathsf{over} \gets  \Call{OblGt}{\mathsf{count}, \mathsf{PADSIZE}}$
      \State $A[j].\texttt{isOverflow} \gets \Call{OblChoose}{\mathsf{over},\; 1,\; 0}$
    \EndFor
  \end{algorithmic}
\end{algorithm}

\section{Oblivious Primitives}
\label{app:obl-primitive}

Although \system's request generation and OHT-based request identification are oblivious by construction (\S\ref{sec:components}),
the subsequent \emph{in-bin} processing must also avoid data-dependent leakage.
Concretely, after an OHT bin is selected, \system must (i) check which candidate slots match the target request and
(ii) copy the corresponding neuron/state into the request's neuron-buffer.
If implemented with ordinary branches (e.g., \texttt{if} statements) or data-dependent memory accesses,
these steps can leak through the instruction trace (branch predictor) and microarchitectural effects.
Therefore, \system relies on branch-free, data-independent \emph{oblivious primitives} for comparison and conditional assignment.

\subsection{Oblivious Comparison (\textsc{OblCompare})}
\label{subsec:obli_compare}
 
\textsc{OblCompare} is the fundamental equality predicate used throughout \system.
It performs a branch-free, data-independent test of whether two scalar operands are equal~\cite{dauterman2021snoopy}.
Algorithm~\ref{alg:obli_compare} shows the conceptual instruction sequence:
a \texttt{cmp} instruction sets the processor flags based on the difference $x - y$,
and \texttt{sete} copies the zero flag (\texttt{ZF}) into the output byte.
This sequence contains no data-dependent jumps, no secret-dependent addressing, and no variable-length loops.
 
\textsc{OblCompare} is consumed by operations that require identity checks,
such as OHT lookup (Algorithm~\ref{alg:neuron-fetcher}, line~7),
label matching during backpropagation (Algorithm~\ref{alg:updater-bp}, line~9),
and duplicate detection in oblivious merge (Algorithm~\ref{alg:obl-merge}, line~3).
 
\begin{algorithm}[H]
\caption{\textsc{OblCompare}$(x,y)$}
\label{alg:obli_compare}
\begin{algorithmic}[1]
  \Require Scalars $x$, $y$
  \Ensure $\textit{eq}=1$ iff $x=y$; $0$ otherwise
  \State \texttt{cmp }$y,x$                 
  \Comment{updates \texttt{ZF} $\leftarrow (x-y==0)$}
  \State \texttt{sete eq}                   
  \Comment{$eq \leftarrow 1$ if \texttt{ZF}; $0$ otherwise}
  \State \Return \textit{eq}
\end{algorithmic}
\end{algorithm}

\subsection{Oblivious Greater-Than}
\label{subsec:obli_gt}

\textsc{OblGt} is an oblivious predicate that materializes a strict-greater-than bit for use by later branch-free updates~\cite{dauterman2021snoopy}.
For integral values, our implementation uses the x86 \texttt{cmp}/\texttt{setg} sequence, analogous to \textsc{OblCompare}.
For floating-point values, we omit the algorithmic details here, as it follows similar logic.
A type-specific fixed FP/SIMD handles floating-point operands.
\textsc{OblCompare} first sets its output from the zero flag (\texttt{ZF}) via \texttt{sete},
\textsc{OblGt} then reads the sign and overflow flags via \texttt{setg} (set~if~greater, signed comparison) after the same \texttt{cmp} instruction.
The resulting instruction sequence (Algorithm~\ref{alg:obli_gt}) is identical in length and structure:
it contains no data-dependent jumps, no secret-dependent addressing, and no variable-length loops,
so the instruction trace is independent of the operand values.

\begin{algorithm}[H]
\caption{\textsc{OblGt}$(x,y)$ for integral operands}
\label{alg:obli_gt}
\begin{algorithmic}[1]
  \Require Scalars $x$, $y$
  \Ensure $\textit{gt}=1$ iff $x > y$
  \State \texttt{cmp }$y,x$                 
  \Comment{updates flags based on $x-y$}
  \State \texttt{setg gt}                   
  \Comment{$gt \leftarrow 1$ if $\texttt{SF}=\texttt{OF}$ and $\texttt{ZF}=0$; $0$ otherwise}
  \State \Return \textit{gt}
\end{algorithmic}
\end{algorithm}

Together, \textsc{OblCompare} and \textsc{OblGt} form a minimal pair of branch-free predicates:
\textsc{OblCompare} produces the equality predicate consumed by operations such as OHT lookup (Algorithm~\ref{alg:neuron-fetcher}) and label matching (Algorithm~\ref{alg:updater-bp}),
while \textsc{OblGt} produces the ordering predicate consumed by ReLU clamping and the oblivious maximum in the softmax pass (Algorithm~\ref{alg:updater-ff}).
Both predicates serve as inputs to \textsc{OblChoose} (Algorithm~\ref{alg:obli_choose}), which performs the actual conditional data movement.

\subsection{Oblivious Conditional Select (\textsc{OblChoose})}
\label{subsec:obli_choose}
 
After obtaining a predicate from \textsc{OblCompare} or \textsc{OblGt}, \system must select between two candidate values without revealing which one was chosen.
Algorithm~\ref{alg:obli_choose} shows the \textsc{OblChoose} primitive~\cite{dauterman2021snoopy}.
Given a Boolean predicate $\textit{pred}\in\{0,1\}$ and two candidate values $v_1$, $v_0$, it returns $v_1$ if $\textit{pred}=1$ and $v_0$ otherwise,
while producing an identical instruction and memory-access trace regardless of the predicate's value.
On x86, this is naturally implemented with a conditional move (\texttt{cmov}) that does not introduce a branch.
 
\begin{algorithm}[htb]
\caption{\textsc{OblChoose}$(\textit{pred},\;v_1,\;v_0)$}
\label{alg:obli_choose}
\begin{algorithmic}[1]
  \Require $\textit{pred}\in\{0,1\}$, values $v_1$, $v_0$
  \Ensure Returns $v_1$ if $\textit{pred}=1$; $v_0$ otherwise
  \State \texttt{mov  tmp,\;$v_1$}          
  \State \texttt{test pred,\;pred}           
  \Comment{sets \texttt{ZF}$\leftarrow(\textit{pred}==0)$}
  \State \texttt{cmovz tmp,\;$v_0$}         
  \Comment{if \texttt{ZF}, tmp $\leftarrow v_0$}
  \State \Return \textit{tmp}
\end{algorithmic}
\end{algorithm}

\subsection{Byte-Level Oblivious Assignment (\textsc{OblChooseBytes})}
\label{subsec:obli_choosebytes}

The scalar \textsc{OblChoose} (Algorithm~\ref{alg:obli_choose}) handles a single register-width value via one \texttt{cmov} instruction.
In practice, however, \system often needs to conditionally assign \emph{contiguous memory blocks}.
For instance, the weight array or the gradient accumulator array inside a neuron record, rather than a single scalar or an entire composite object at once.

\textsc{OblChooseBytes} (Algorithm~\ref{alg:obli_choosebytes}) generalises \textsc{OblChoose} to an arbitrary-length byte buffer.
Given a predicate $\mathit{pred}$, two source buffers $t\_buf$ (selected when $\mathit{pred}=1$) and $f\_buf$ (selected when $\mathit{pred}=0$), and a destination buffer $\mathit{dst}$, each of $N$ bytes, it writes $t\_buf$ into $\mathit{dst}$ if $\mathit{pred}=1$, and $f\_buf$ into $\mathit{dst}$ otherwise.

Internally, it partitions the $N$ bytes into fixed-size chunks and applies \textsc{OblChoose} to each chunk, from largest to smallest.
For 32-byte and 16-byte chunks, the single-scalar \texttt{cmov} in \textsc{OblChoose} is replaced by a wider vector conditional-select instruction (AVX2 or SSE) that processes the entire chunk; the semantics remain identical.
The chunk decomposition depends only on the public buffer size $N$, so the instruction count and sequence are fixed regardless of $\mathit{pred}$.

\begin{algorithm}[htb]
\caption{\textsc{OblChooseBytes}$(\mathit{pred},\;N,\;t\_buf,\;f\_buf,\;\mathit{dst})$}
\label{alg:obli_choosebytes}
\begin{algorithmic}[1]
  \Require $\mathit{pred}\in\{0,1\}$; byte buffers $t\_buf[0..N{-}1]$, $f\_buf[0..N{-}1]$, $\mathit{dst}[0..N{-}1]$
  \Ensure $\mathit{dst} \gets \mathit{pred}\;?\;t\_buf : f\_buf$; instruction trace independent of $\mathit{pred}$
  \State $k \gets 0$
  \While{$k + 32 \le N$}
      \Comment{32-byte chunks via AVX2 vector \textsc{OblChoose}}
    \State $\mathit{dst}[k..k{+}31] $ 
    \Statex $\ \ \ \ \ \  \gets \textsc{OblChoose}(\mathit{pred},\;t\_buf[k..k{+}31],\;f\_buf[k..k{+}31])$
    \State $k \gets k + 32$
  \EndWhile
  \If{$k + 16 \le N$}
      \Comment{16-byte remainder via SSE vector \textsc{OblChoose}}
    \State $\mathit{dst}[k..k{+}15] $ 
    \Statex $\ \ \ \ \ \  \gets \textsc{OblChoose}(\mathit{pred},\;t\_buf[k..k{+}15],\;f\_buf[k..k{+}15])$
    \State $k \gets k + 16$
  \EndIf
  \If{$k + 8 \le N$}
      \Comment{8-byte remainder via scalar \textsc{OblChoose}}
    \State $\mathit{dst}[k..k{+}7] $ 
    \Statex $\ \ \ \ \ \  \gets \textsc{OblChoose}(\mathit{pred},\;t\_buf[k..k{+}7],\;f\_buf[k..k{+}7])$
    \State $k \gets k + 8$
  \EndIf
  \If{$k + 4 \le N$}
      \Comment{4-byte remainder}
    \State  $\mathit{dst}[k..k{+}3] $ 
    \Statex $\ \ \ \ \ \  \gets \textsc{OblChoose}(\mathit{pred},\;t\_buf[k..k{+}3],\;f\_buf[k..k{+}3])$
    \State $k \gets k + 4$
  \EndIf
  \If{$k + 2 \le N$}
      \Comment{2-byte remainder}
    \State $\mathit{dst}[k..k{+}1] $ 
    \Statex $\ \ \ \ \ \  \gets \textsc{OblChoose}(\mathit{pred},\;t\_buf[k..k{+}1],\;f\_buf[k..k{+}1])$
    \State $k \gets k + 2$
  \EndIf
  \If{$k < N$}
      \Comment{1-byte remainder}
    \State $\mathit{dst}[k] \gets \textsc{OblChoose}(\mathit{pred},\;t\_buf[k],\;f\_buf[k])$
  \EndIf
\end{algorithmic}
\end{algorithm}

\noindent\textbf{Usage in \system.}
When processing a composite object such as a neuron record, \system does not apply \textsc{OblChooseBytes} to the entire struct in one call.
Instead, it applies \textsc{OblChooseBytes} separately to each contiguous array field within the struct (e.g., the weight array $u.\mathsf{weights}[0..D_\ell{-}1]$, etc.), sharing the same predicate obtained from \textsc{OblCompare} or \textsc{OblGt}.
This field-by-field approach avoids the need to pack the entire struct into a single contiguous buffer while still ensuring that every field is assigned in a fixed, branch-free manner.

\subsection{Extending Primitives to Composite Objects}
\label{app:obli_composite}

The primitives defined in Sections~\ref{subsec:obli_compare}--\ref{subsec:obli_choosebytes} operate on scalars and contiguous byte buffers.
However, \system's core data structures, such as neurons, request entries, and OHT slots, are composite objects containing a mix of scalar fields (e.g., neuron ID, flags, bias) and array fields (e.g., weights, gradient accumulators, moment estimates).
To conditionally assign an entire composite object, \system performs an \emph{oblivious structural copy}: given a single predicate $\mathit{pred}$ (obtained from \textsc{OblCompare} or \textsc{OblGt}), it walks the struct's fields in a fixed, publicly determined order and applies the appropriate primitive to each:

\begin{itemize}[nosep]
  \item \textbf{Scalar fields} (e.g., ID, bias, flags, counters, etc.): a single \textsc{OblChoose} call per field.
  \item \textbf{Contiguous array fields} (e.g., weight vector $u.\mathsf{weights}[0..D_\ell-1]$, etc.): a single \textsc{OblChooseBytes} call per array, processing the buffer in 32/16/8/.../1-byte chunks as described in Section~\ref{subsec:obli_choosebytes}.
\end{itemize}

\noindent
The number of fields, the size of each array, and the traversal order are all compile-time constants determined by the struct layout.
No memory is allocated conditionally; all destination storage is pre-allocated, and the primitives only \emph{select} which source values become visible at the destination.
As a result, the instruction count and memory-access pattern of an oblivious structural copy depend only on the struct definition, not on the predicate or the data values.

Concretely, when scanning candidate slots inside an \texttt{OHT} bin (Algorithm~\ref{alg:neuron-fetcher}, lines~6--11), \system applies an oblivious structural copy at every entry in this bin using the predicate from \textsc{OblCompare}.
Every entry is written to the matching entries, which receive the actual neuron data, while non-matching entries receive their own current values (a self-assignment).
The adversary therefore observes the same sequence of memory accesses regardless of which slot is the true match.\\

\noindent\textbf{Oblivious Sorting.}
\system uses a sorting-network-based oblivious sort, instantiated with the iterative formulation of bitonic sort~\cite{batcher1968sorting}.
The implementation consists of three nested loops whose bounds depend only on the (public) array length~$n$:
the outer loop doubles the merge stride $k$ from $2$ to $n$,
the middle loop halves the comparison distance $j$ from $k/2$ to $1$,
and the inner loop processes the public representatives
$i \in [0,n)$ satisfying $(i \oplus j)>i$.
For each $i$, the partner index $i \oplus j$ (bitwise XOR) and the comparison direction (ascending or descending, determined by the bit $i \mathbin{\&} k$) are both fixed functions of $i$, $j$, and $k$. They do not depend on element values.
At each representative pair, the two elements are compared to obtain a predicate and then conditionally swapped via \textsc{OblChoose}.
Because all loop bounds, index pairs, and comparison directions are determined solely by the public array length,
the iterative structure contains no data-dependent branches and no recursion, producing a fully fixed instruction and memory-access trace.

\emph{Parallelism.}
For a given $(k, j)$, 
all representative pairs $(i,\; i \oplus j)$ satisfying $(i \oplus j)>i$ are disjoint and independent.
The inner loop can therefore be executed in parallel, which is a well-known property of bitonic sorting networks~\cite{batcher1968sorting}.

\emph{Non-power-of-two lengths.}
The iterative bitonic sort requires the array length to be a power of two.
When the array length $n$ is not a power of two, \system decomposes it as $n = m + (n - m)$, where $m$ is the largest power of two less than $n$.
The first $m$ elements are sorted (in reverse) using the iterative bitonic sort;
the remaining $n - m$ elements are sorted recursively using the same decomposition;
and the two halves are combined via an oblivious bitonic merge.
The split point $m$ depends only on the public length $n$, so the decomposition introduces no data-dependent branching.

Because we extend \textsc{OblChoose} to composite objects (Section~\ref{app:obli_composite}) via \textsc{OblChooseBytes} (Section~\ref{subsec:obli_choosebytes}), the conditional swap operates on \emph{entire} composite records (e.g., request entries, neuron records) while keeping control flow, loop bounds, and accessed addresses unchanged.
The adversary can observe that pairs of elements are accessed in a fixed order (the sorting network's structural guarantee), but cannot distinguish whether any pair was swapped.
Therefore, the sorted order is not revealed, and the full oblivious sort is realized by combining the data-oblivious bitonic network with our composite \textsc{OblChoose} primitive.

\section{Security Proof}
\label{app:proof}

\subsection{Security Model}
\label{sec:security-model}

We model \system as a \emph{two-party computation} between a \emph{User} (i.e., $P_1$) and a \emph{Hypervisor} (i.e., $P_2$).
The User holds input $x_1 = (P, D)$, consisting of public parameters~$P$ and private training data~$D$.
The Hypervisor holds input $x_2 = P$.
We adopt the standard real/ideal paradigm for secure two-party computation~\cite{canetti2000security,lindell2017simulate}.
 
\noindent\textbf{Real world.}
In the real world, the two parties execute a protocol~$\pi_{\system}$.
Our protocol~$\pi_{\system}$ is \emph{\system \ running inside a TDX trust domain}.
The User sends their input $x_1=(P, D)$ to the trust domain, which executes the \system training algorithm and returns the trained model~$M$ to the User.
The Hypervisor is passive: it does not send messages to the trust domain.
 
Although all code and data reside within the trust domain and are protected by hardware memory encryption, in a \emph{real implementation} (as opposed to a theoretical framework), the Hypervisor can observe which encrypted memory locations are accessed at page and cache-line granularity.
For the proof, we model these observations by a canonical algorithmic trace~$\tau$, consisting of symbolic object/offset accesses and public per-worker schedules.
Concrete page/cache-line traces are covered only when the implementation maps symbolic accesses to addresses by a public layout and uses a public deterministic interleaving, or the public per-worker trace abstraction, for parallel work.
Formally, the view is defined as $\mathsf{view}_{\mathsf{H}} = \{P, \tau\}$.
 
\noindent\textbf{Ideal World.}
In the ideal world, trusted functionality $\F_{\mathsf{train}}$ between $P_1$ and $P_2$ is executed by a trusted party of the ideal world. $\F_{\mathsf{train}}$ receives inputs from both parties, computes the result, and returns an output to each party:
\[
    \F_{\mathsf{train}}\bigl[(P, D),\; P\bigr] = \bigl[M,\; \bot\bigr],
\]
where $P_1$'s input is $(P, D)$ and output is $M$, while $P_2$'s input is $P$ and output is $\bot$. 
The User receives the trained model~$M$; the Hypervisor receives nothing~($\bot$).

\begin{figure}[ht]
\centering
\fbox{\begin{minipage}{0.94\linewidth}
\textbf{Functionality $\F_{\mathsf{train}}$}
\\ \hrule
\begin{tabular}{@{}l@{}}
$(M,\; \bot) \leftarrow \F_{\mathsf{train}}\bigl((P, D),\; P\bigr)$: \\[4pt]
\quad 1: Receive input $(P, D)$ from $P_1$ (User) \\
\quad\quad \; and $P$ from $P_2$ (Hypervisor). \\[2pt]
\quad 2: Train model $M$ on dataset $D$ with public parameters $P$. \\[2pt]
\quad 3: \textbf{return} $M$ to $P_1$ and $\bot$ to $P_2$.
\end{tabular}
\end{minipage}}
\caption{Ideal training functionality.}
\label{func:train}
\end{figure}

\paragraph{Security Definition.}
Following the standard semi-honest definition~\cite{canetti2000security,lindell2017simulate}, we define:
\begin{itemize}[nosep]
    \item $\mathbf{Real}_{\pi_{\system},\mathcal{H}}(\kappa;\, (P, D),\, P)$:
    Run protocol~$\pi_{\system}$ with security parameter~$\kappa$, where the User uses input $(P, D)$ and the Hypervisor uses input $P$.
    Let $\mathsf{view}_{\mathsf{H}} = \{P, \tau\}$ be the Hypervisor's view (its input~$P$ and the observed trace~$\tau$), and let the parties' outputs be $y_1 = M$, $y_2 = \bot$.\\
    Output $(\mathsf{view}_{\mathsf{H}},\, (y_1, y_2))$.
    \item $\mathbf{Ideal}_{\F_{\mathsf{train}},\Sim}$:
    Compute $(y_1, y_2) \leftarrow \F_{\mathsf{train}}\bigl((P, D),\, P\bigr)$, where $y_1=M$ and $y_2=\bot$.\\
    Output $\bigl(\Sim(P, \bot),\, (M, \bot)\bigr)$.
\end{itemize}
 
\begin{definition}[Security of \system]
\label{def:security}
Protocol~$\pi_{\system}$ \emph{securely realizes} $\F_{\mathsf{train}}$ in the presence of a semi-honest adversary (the Hypervisor) if there exists a PPT simulator~$\Sim$ such that, for all public parameters~$P$ and all datasets~$D$ consistent with~$P$:
\[
    \bigl\{\mathbf{Ideal}_{\F_{\mathsf{train}},\Sim}(\kappa;\, (P,D),\, P)\bigr\}
    \;\approx\;
    \bigl\{\mathbf{Real}_{\pi_{\system},\mathcal{H}}(\kappa;\, (P,D),\, P)\bigr\}
\]
where $\approx$ denotes computational indistinguishability with respect to the security parameter~$\kappa$.
\end{definition}
 
Concretely, the simulator~$\Sim$ receives only~$P$ and must produce a simulated Hypervisor view indistinguishable from $\{P, \tau\}$.
Since the Hypervisor is passive, the simulation task reduces to generating a trace that is indistinguishable from the real execution trace without knowledge of~$D$.

\noindent\textbf{Simplification.}
The protocol~$\pi_{\system}$ computes the same function as $\F_{\mathsf{train}}$.
The only difference is how the output is produced.
In \system we follow an oblivious architecture with techniques such as dummy padding, oblivious sorting, and branch-free primitives that alter the memory access pattern but not the final output.
Therefore, the outputs $y_1 = M$ and $y_2 = \bot$ are identically distributed in both worlds, and the security condition reduces to comparing the Hypervisor's view alone.

Note that the joint distribution $(\mathsf{view}_{\mathsf{H}}, (y_1, y_2))$ is simulatable, not just the marginal $\mathsf{view}_{\mathsf{H}}$.
As we prove in Sections~\ref{sec:proof-req}--\ref{sec:proof-upd}, the canonical algorithmic trace~$\tau$ produced by each component depends only on~$P$ and on fresh randomness (e.g., dummy inputs) that is independent of the private data~$D$, and of the training outcome~$M$.
Consequently, $\tau$ is independent of~$M$ given~$P$, and the security condition reduces to simulating the Hypervisor's trace alone.

\begin{definition}[Trace-based Security of \system]
\label{def:security-working}
Protocol~$\pi_{\system}$ securely realizes $\F_{\mathsf{train}}$ if there exists a simulator~$\Sim$ such that, for all~$P$ and all~$D$ consistent with~$P$:
\[
    \bigl\{\Sim(P)\bigr\}
    \;\approx_c\;
    \bigl\{\mathsf{view}_{\mathsf{H}}^{\pi_{\system}}(P, D)\bigr\}
\]
where $\mathsf{view}_{\mathsf{H}}^{\pi_{\system}}(P, D) = \{P, \tau\}$ is the Hypervisor's (H) view during real execution, and $\approx_c$ denotes computational indistinguishability in the security parameter~$\kappa$.
\end{definition}

\noindent\textbf{Simulator Convention.}
In the component proofs that follow, each simulator is described as an algorithm that mirrors the real protocol's control flow on dummy inputs.
The simulator's \emph{output in the security statement} is the canonical algorithmic trace generated during its execution, denoted $\hat{\tau}$.
The data structures returned by the simulator (e.g., a dummy $\mathsf{ReqArray}$ or a dummy $\mathsf{LSH}$ table) are internal states used solely for composing consecutive components.
They are not part of the Hypervisor's view.
Formally, $\Sim(P)$ denotes $(P, \hat{\tau})$ where $\hat{\tau}$ is the canonical symbolic trace generated by the simulator; the concrete-address caveat is the one stated in the security model above.

\subsection{Public and Private Information.}
\label{sec:public-private}

The adversary (Hypervisor, H) is assumed to know the following \emph{public parameters}~$P$.
\begin{itemize}[nosep]
    \item \emph{NN architecture}: the number of layers~$L$, the neuron count per layer (specifically, layer-0 size~$N_0$), each layer's weight dimension~$D_\ell$, and each layer's activation type (ReLU or Softmax);
    \item \emph{Training configuration}: batch size~$B$, the number of training iterations, the Adam optimizer~\cite{kingma2014adam} hyperparameters ($\alpha, \beta_1, \beta_2, \epsilon$), and the training step;
    \item \emph{LSH configuration}: the number of hash functions~$K$, the number of buckets~$\mathsf{numLshBucs}$, the bucket capacity~$\PADSIZE$, any dataset-dependent rule used to choose $\PADSIZE$, the probing sequence length~$\mathsf{lenSeq}_\ell$, the MP-WTA parameters ($r{=}3$, $N{=}3$), and the LSH rebuild period;
    \item \emph{OHT configuration}: the number of bins per tier~$\mathsf{numBin}$, the bin capacity~$\mathsf{binCap}$, duplicate-record semantics for repeated \texttt{LSHbucketID} values, and the no-observable-overflow/failure capacity contract;
    \item \emph{Schedule metadata}: all public hash descriptions or seeds, sampled WTA subsets, deterministic tie-breaking rules, and public optimization schedules used by O1--O3;
    \item \emph{Data characteristics}: the dataset size, the input dimensionality~$D_{input}$, the number of non-zero features per input~$\mathsf{nnz}$, the non-zero feature indices $\{d_1,\ldots,d_{\mathsf{nnz}}\}$ of each input, and the label count~$|Y_i|$ per input;
    \item \emph{Derived quantities}: the \texttt{Request-Array} size $T = B \cdot \mathsf{lenSeq}_\ell$, the number of neuron slots per input ($R_\ell = \mathsf{lenSeq}_\ell \cdot \PADSIZE$ for LSH-selected layers, and $R_0=N_0$ for dense layer~0), and the Neuron Fetcher operation mode (Read/Write).
\end{itemize}
These parameters fully determine all loop bounds, array dimensions, and traversal schedules throughout the protocol.

\noindent
\system protects the following \emph{private information}: the input training data values $D$;
all neuron weights, biases, activations, gradients, deltas, and optimizer moment estimates;
all data-dependent WTA/LSH signatures and neuron-to-bucket assignments;
all training label values~$Y_i$ and input feature values~$x_{d_j}$;
and all intermediate quantities arising during execution
(e.g., OHT contents, softmax normalization constants, and internal vectors used during feedforward and backpropagation).

\

\noindent\emph{Remark on input sparsity.}
\system focuses on protecting the access patterns of the wide output layer, where data-dependent neuron selection via LSH creates the primary side-channel threat (Section~\ref{sec:overview}, Figure~\ref{fig:overview}).
The non-zero feature indices $\{d_1,\ldots,d_{\mathsf{nnz}}\}$ of each input are treated as public because they affect the dense layer~0 feedforward scan and the first-layer gradient update (Algorithm~\ref{alg:updater-bp}, Phase~3), which operate on the narrow input layer and do not involve LSH-based selection.

Hiding sparse access patterns has been studied in both TEE-based systems~\cite{ohrimenko2016oblivious,kato2023olive, umar2025efficient} and MPC frameworks~\cite{schoppmann2019make,damie2025secure}.
These techniques address an orthogonal concern and can be composed with \system at an additional cost of $O(N_0 \cdot D \log D)$ per batch.
We consider this complementary and leave integration to future work.

For intuition, we provide a full-dimension scanning front-end that removes this exposure, together with its measured cost on the evaluated workloads. 
With that front end enabled, the indices $\{d_1, \dots, d_{\mathrm{nnz}}\}$ are removed from $P$, and Theorem~\ref{thm:system-final} holds against an adversary that additionally does not learn the sparse input's non-zero positions.

\subsection{Oblivious Building Blocks}
\label{sec:building-blocks}

We prove the security of \system in the \emph{hybrid model}~\cite{canetti2000security}.
In the hybrid model, the protocol may invoke a set of functionalities as black-box subroutines.
Each such subroutine correctly computes its specified function and produces a trace that depends only on public parameters, which is oblivious.
Security proven in the hybrid model extends to the fully instantiated protocol, since each functionality can be replaced by a concrete oblivious implementation while preserving security, via sequential composition~\cite{canetti2000security}.

Below we describe the functionalities that \system uses in the hybrid model.
Each corresponds to a standard oblivious primitive whose security has been established in prior work.

\begin{figure}[ht]
\centering
\fbox{\begin{minipage}{0.94\linewidth}
\textbf{Functionality $\F_{\OSort}$}
\\ \hrule
\begin{tabular}{@{}l@{}}
$A' \leftarrow \F_{\OSort}(A, n, \mathsf{key})$: \\[2pt]
\quad 1: Sort the $n$-element array $A$ by comparison key $\mathsf{key}$.\\
\quad 2: \textbf{return} sorted array $A'$.
\end{tabular}
\end{minipage}}
\caption{Oblivious sort functionality.}
\label{func:osort}
\end{figure}

\noindent
$\F_{\OSort}$ sorts an array of $n$ elements by a given comparison key.
The trace depends only on~$n$; a secure realization $\pi_{\mathsf{OSort}}$ exists via bitonic sort~\cite{batcher1968sorting}.

\begin{figure}[ht]
\centering
\fbox{\begin{minipage}{0.94\linewidth}
\textbf{Functionality $\F_{\OCmpSet}$}
\\ \hrule
\begin{tabular}{@{}l@{}}
$\F_{\OCmpSet}(a, b, \mathsf{pred}, \mathsf{dst})$: \\[2pt]
\quad 1: Evaluate predicate $\mathsf{pred}$ on $(a, b)$.\\
\quad 2: \textbf{if} $\mathsf{pred}$ is true, write $a$ to $\mathsf{dst}$;
         \textbf{else} write $b$ to $\mathsf{dst}$.
\end{tabular}
\end{minipage}}
\caption{Oblivious compare-and-set functionality.}
\label{func:ocmpset}
\end{figure}

\noindent
$\F_{\OCmpSet}$ is an umbrella functionality for the branch-free scalar steps used in \system.
It covers the compare and select steps behind \texttt{OblCompare}, \texttt{OblGt}, and \texttt{OblChoose}.
The public input $\mathsf{pred}$ says which test is used.
The values $a$ and $b$ are private.
The trace does not depend on the truth value of the test or on the values $a$ and $b$.
A secure realization $\pi_{\mathsf{OCmpSet}}$ exists via branch-free compare and conditional move instructions~\cite{ohrimenko2016oblivious, dauterman2021snoopy}.

\begin{figure}[ht]
\centering
\fbox{\begin{minipage}{0.94\linewidth}
\textbf{Functionality $\F_{\OHT}$}
\\ \hrule
\begin{tabular}{@{}l@{}}
$(T,\; \mathsf{H}_1,\; \mathsf{H}_2) \leftarrow \F_{\OHT}.\mathsf{Build}(S,\; \mathsf{numBin},\; \mathsf{binCap})$: \\[2pt]
\quad 1: Use the hash functions
  $\mathsf{H}_1, \mathsf{H}_2$ specified by $P$. \\
\quad 2: Construct a two-tier hash table $T$ from key-value pairs $S$,\\
\quad\quad\; using $\mathsf{H}_1, \mathsf{H}_2$,
  with $\mathsf{numBin}$ bins of capacity $\mathsf{binCap}$.\\
\quad 3: \textbf{return} $(T,\; \mathsf{H}_1,\; \mathsf{H}_2)$.
\end{tabular}
\end{minipage}}
\caption{Oblivious hash table functionality.}
\label{func:oht}
\end{figure}

\noindent
$\F_{\OHT}$ provides an oblivious build procedure for a fixed-capacity two-tier hash table.
The build trace depends only on $|S|$, $\mathsf{numBin}$, $\mathsf{binCap}$.
A secure realization $\pi_{\mathsf{OHT}}$ exists via~\cite{chan2017oblivious}.
In \system, $S$ is a multiset of request records: duplicate private \texttt{LSHbucketID} values are treated as distinct records (requested by distinct inputs in a batch), and the public OHT capacity parameters are chosen to avoid observable overflow or failure for all request multisets of public size $T$.

We assume that the OHT build procedure does not perform data-dependent retries or hash-function resampling.
The hash functions $\mathsf{H}_1, \mathsf{H}_2$ are specified by $P$ before examining the key-value pairs.
Items are placed into fixed-capacity bins via oblivious sort; no Cuckoo-style eviction loop is used. 
The standard capacity analysis of~\cite{chan2017oblivious} applies when its load assumptions hold; for TENNOR's duplicate-request case, the theorem relies on the explicit public-capacity contract above.

\subsection{Functionalities of \system's Components}
\label{sec:component-functionalities}

We organize one batch step of $\pi_{\system}$ into three sequential components, namely Neuron Requester, Neuron Fetcher, and Neuron Updater, each with its own functionality.
Each functionality is a black box that receives inputs (including the public parameters~$P$), computes the specified operation, and returns outputs.
All quantities associated with~$P$ (such as $T = B \cdot \mathsf{lenSeq}_\ell$ and $R_\ell = \mathsf{lenSeq}_\ell \cdot \PADSIZE$) are implicitly available to the functionality.

Because these three batch-step components execute \emph{strictly sequentially} (the nonconcurrency condition required by the composition theorem~\cite{canetti2000security}), their composed security follows from their individual security guarantees.

\begin{figure}[ht]
\centering
\fbox{\begin{minipage}{0.96\linewidth}
\textbf{Functionality $\F_{\mathsf{req}}$}
\\ \hrule
\begin{tabular}{@{}l@{}}
$\mathsf{ReqArray}_\ell \leftarrow \F_{\mathsf{req}}(\ell,\; \{q_1,\ldots,q_B\},\; P)$: \\[4pt]
\quad 1: For each input $q_i$ in the batch, compute its MP-WTA \\
\quad\quad\; probing sequence and assemble the corresponding \\ \quad\quad\; \texttt{Request-Array} entries, \\
\quad\quad\; each containing an \texttt{LSHbucketID}, a \texttt{Batch-Rank}, and a \\
\quad\quad\; dummy \texttt{Neuron-Buffer} of $\PADSIZE$ entries. \\[2pt]
\quad 2: \textbf{return} $\mathsf{ReqArray}_\ell$ of size $T = B \cdot \mathsf{lenSeq}_\ell$.
\end{tabular}
\end{minipage}}
\caption{Functionality for Neuron Requester.}
\label{func:req}
\end{figure}
 
\noindent
$\F_{\mathsf{req}}$ computes the MP-WTA probing sequence for each input in the batch and assembles the \texttt{Request-Array}.

\begin{figure}[ht]
\centering
\fbox{\begin{minipage}{0.96\linewidth}
\textbf{Functionality $\F_{\mathsf{fetch}}$}
\\ \hrule
\begin{tabular}{@{}l@{}}
$(\mathsf{ReqArray}_\ell',\; \mathsf{LSH}_\ell') \leftarrow
  \F_{\mathsf{fetch}}(\mathsf{ReqArray}_\ell,\;
  \mathsf{LSH}_\ell,\; \mathsf{mode},\; P)$: \\[4pt]
\quad 1: \textbf{if} $\mathsf{mode} = \textsc{Read}$, populate each entry's \\
\quad\quad\; \texttt{Neuron-Buffer} with the neurons stored in the \\
\quad\quad\; LSH bucket matching its \texttt{LSHbucketID}. \\
\quad\quad\; $\mathsf{LSH}_\ell' \leftarrow \mathsf{LSH}_\ell$. \\[2pt]
\quad 2: \textbf{if} $\mathsf{mode} = \textsc{Write}$, write each non-dummy entry's \\
\quad\quad\; \texttt{Neuron-Buffer} back to the matching LSH bucket. \\
\quad\quad\; $\mathsf{LSH}_\ell'$ reflects the updated table. \\[2pt]
\quad 3: \textbf{return} $(\mathsf{ReqArray}_\ell',\; \mathsf{LSH}_\ell')$.
\end{tabular}
\end{minipage}}
\caption{Functionality for Neuron Fetcher.}
\label{func:fetch}
\end{figure}

\noindent
$\F_{\mathsf{fetch}}$ transfers data between the \texttt{Request-Array} and the LSH table.
In \textsc{Read} mode it populates each entry's \texttt{Neuron-Buffer} from the matching LSH bucket and returns the LSH table unchanged.
In \textsc{Write} mode, it writes modified neurons back to the LSH table and returns both the updated $\mathsf{ReqArray}_\ell$ and $\mathsf{LSH}_\ell$.
All outputs are delivered to the User (the TD), the Hypervisor receives $\bot$.

\begin{figure}[ht]
\centering
\fbox{\begin{minipage}{0.96\linewidth}
\textbf{Functionality $\F_{\mathsf{upd}}$}
\\ \hrule
\begin{tabular}{@{}l@{}}
$\mathsf{ReqArray}_\ell \leftarrow
  \F_{\mathsf{upd}}(\mathsf{ReqArray}_\ell,\;
  \mathsf{ActiveNodes},\; \mathsf{ActiveVals},\;$\\
\quad\quad\quad\quad\quad\quad\quad\quad $\{Y_i\},\; \{\mathbf{x}_i\},\; P)$: \\[4pt]
\quad 1: Perform feedforward: compute activations for all inputs \\
\quad\quad\; across the populated \texttt{Request-Array}. \\[2pt]
\quad 2: Perform backpropagation: compute and accumulate \\
\quad\quad\; gradients for all neurons in the \texttt{Request-Array}. \\[2pt]
\quad 3: Merge duplicate neurons that are in the duplicate \\ \quad\quad\; entries sharing the same \texttt{LSHbucketID} \\
\quad\quad\; by consolidating their gradient accumulators. \\[2pt]
\quad 4: Apply the Adam optimizer to update all neuron weights \\
\quad\quad\; and biases. \\[2pt]
\quad 5: \textbf{return} updated $\mathsf{ReqArray}_\ell$.
\end{tabular}
\end{minipage}}
\caption{Functionality for Neuron Updater.}
\label{func:upd}
\end{figure}
 
\noindent
$\F_{\mathsf{upd}}$ performs the complete update cycle on the populated neurons in the \texttt{Request-Array} and on dense layer~0, including feedforward activation, backpropagation, gradient merge, and optimizer step.
For readability, the displayed interface suppresses the dense layer~0 state; it is passed through the updater, updated by the public-order dense scan, and returned with the same public shape.

\subsection{Security Statements}
\label{sec:security-statements}

We now state the security of the three components.
Each theorem below is a hybrid-model statement.
It proves security while the primitive calls inside the component are still ideal calls to the building blocks from Section~\ref{sec:building-blocks}.
The step that replaces those ideal primitive calls by the real primitive protocols is given later in Section~\ref{sec:primitive-replacement}.

\begin{theorem}[Hybrid-model security of Neuron Requester]
\label{thm:req}
Protocol~$\pi_{\mathsf{req}}$ (Algorithm~\ref{alg:neuron-requester})
securely realizes $\F_{\mathsf{req}}$ in the $\F_{\OCmpSet}$-hybrid model.
That is, there exists a simulator $\Sim_{\mathsf{req}}$ such that for all~$P$ and all~$D$ consistent with~$P$:
\[
    \bigl\{\Sim_{\mathsf{req}}(P)\bigr\}
    \;\approx_c\;
    \bigl\{\mathsf{view}_{\mathsf{H}}^{\pi_{\mathsf{req}}}(P, D)\bigr\}.
\]
\end{theorem}

\begin{theorem}[Hybrid-model security of Neuron Fetcher]
\label{thm:fetch}
Protocol~$\pi_{\mathsf{fetch}}$ (Algorithm~\ref{alg:neuron-fetcher})
securely realizes $\F_{\mathsf{fetch}}$ in the $(\F_{\OHT}, \F_{\OSort}, $ $ \F_{\OCmpSet})$-hybrid model.
That is, there exists a simulator $\Sim_{\mathsf{fetch}}$ such that for all~$P$ and all~$D$ consistent with~$P$:
\[
    \bigl\{\Sim_{\mathsf{fetch}}(P)\bigr\}
    \;\approx_c\;
    \bigl\{\mathsf{view}_{\mathsf{H}}^{\pi_{\mathsf{fetch}}}(P, D)\bigr\}.
\]
\end{theorem}

\begin{theorem}[Hybrid-model security of Neuron Updater]
\label{thm:upd}
Protocol~$\pi_{\mathsf{upd}}$ (Algorithms~\ref{alg:updater-ff}, \ref{alg:updater-bp}, \ref{alg:obl-merge}, and~\ref{alg:adam-update})
securely realizes $\F_{\mathsf{upd}}$ in the $(\F_{\OSort}, \F_{\OCmpSet})$-hybrid model.
This theorem is for the two-layer \system architecture analyzed in Appendix~\ref{app:neuron-updater}, where layer~0 is scanned densely for feedforward and optimizer updates, and the wide output layer is the LSH-selected layer.
That is, there exists a simulator $\Sim_{\mathsf{upd}}$ such that for all~$P$ and all~$D$ consistent with~$P$:
\[
    \bigl\{\Sim_{\mathsf{upd}}(P)\bigr\}
    \;\approx_c\;
    \bigl\{\mathsf{view}_{\mathsf{H}}^{\pi_{\mathsf{upd}}}(P, D)\bigr\}.
\]
\end{theorem}

The proofs of these three theorems are given next.
After that we prove one batch step, then the full training loop.

\subsection{Proof of Theorem~\ref{thm:req}: Security of Neuron Requester}
\label{sec:proof-req}

Figure~\ref{fig:pi-req} defines the protocol $\pi_{\mathsf{req}}$ and
Figure~\ref{fig:sim-req} defines the simulator $\Sim_{\mathsf{req}}$.

\begin{figure}[ht]
\centering
\fbox{\begin{minipage}{0.96\linewidth}
\textbf{Protocol $\pi_{\mathsf{req}}$}
\\ \hrule
\begin{tabular}{@{}l@{}}
$\mathsf{ReqArray}_\ell \leftarrow \pi_{\mathsf{req}}(\ell,\; \{q_1,\ldots,q_B\},\; P)$: \\[4pt]
\quad 1: Allocate $\mathsf{ReqArray}_\ell[0..T-1]$ with $T = B \cdot \mathsf{lenSeq}_\ell$. \\[2pt]
\quad 2: \textbf{for} $i = 1$ to $B$: \\
\quad\quad\quad $\texttt{Seq}_i \leftarrow$ \textsc{MP-WTA}$(q_i, h_1,\ldots,h_K)$ \\
    \quad\quad\quad (Algorithm~\ref{alg:oblivious_mpwta_mp}; internally invokes $\F_{\OCmpSet}$\\
\quad\quad\quad  via \textsc{$h^2$-WTA} (Algorithm~\ref{alg:abstract_max2wta})\\
\quad\quad\quad  and \textsc{$h^3$-WTA})\\[2pt]
\quad 3: \textbf{for} $i = 1$ to $B$, \textbf{for} $t = 0$ to $\mathsf{lenSeq}_\ell - 1$: \\
\quad\quad\quad populate $\mathsf{ReqArray}_\ell[(i-1)\!\cdot\!\mathsf{lenSeq}_\ell + t]$ \\
    \quad\quad\quad (Algorithm~\ref{alg:neuron-requester}, lines 6--12). \\[2pt]
\quad 4: \textbf{return} $\mathsf{ReqArray}_\ell$.
\end{tabular}
\end{minipage}}
\caption{Protocol for Neuron Requester
(in the $\F_{\OCmpSet}$-hybrid model).}
\label{fig:pi-req}
\end{figure}

\noindent\textbf{Protocol Description.}
The protocol $\pi_{\mathsf{req}}$ corresponds directly to Algorithm~\ref{alg:neuron-requester}.
It first allocates a \texttt{Request-Array} of fixed public size $T = B \cdot \mathsf{lenSeq}_\ell$ (step~1).
For each input $q_i$ in the batch, it computes the MP-WTA probing sequence (step~2)
by invoking Algorithm~\ref{alg:oblivious_mpwta_mp},
which internally calls Algorithm~\ref{alg:abstract_max2wta} to determine the maximum, second-maximum, and third-maximum indices within each WTA hash function's sampled subset.
Each such max/second-max/third-max computation uses $\F_{\OCmpSet}$ for branch-free comparison and selection.
Finally, the protocol writes one entry per probe into the \texttt{Request-Array} at a deterministic position (step~3).

\begin{figure}[ht]
\centering
\fbox{\begin{minipage}{0.94\linewidth}
\textbf{Simulator $\Sim_{\mathsf{req}}(P)$}
\\ \hrule
\footnotesize
\begin{tabular}{@{}l@{}}
$\mathsf{ReqArray}_\ell \leftarrow \Sim_{\mathsf{req}}(P)$: \\[4pt]
\quad 0: Generate random  $\{\hat{q}_1, \dots \hat{q}_B\}$ as dummy inputs \\[2pt]

\quad 1: Allocate $\mathsf{ReqArray}_\ell[0..T-1]$ with $T = B \cdot \mathsf{lenSeq}_\ell$;\\[2pt]
\quad 2: \textbf{for} $i = 1$ to $B$: \\
\quad\quad\quad Invoke $\F_{\OCmpSet}$ exactly as \textsc{MP-WTA} \\
\quad\quad\quad (Algorithm~\ref{alg:oblivious_mpwta_mp} via Algorithm~\ref{alg:abstract_max2wta}) \\
\quad\quad\quad would on an random input $\hat{q}_i$ of the same dimension, \\
\quad\quad\quad to produce a dummy probing sequence $\hat{\texttt{Seq}}_i$ \\
\quad\quad\quad of length $\mathsf{lenSeq}_\ell$. \\[2pt]
\quad 3: \textbf{for} $i = 1$ to $B$, \textbf{for} $t = 0$ to $\mathsf{lenSeq}_\ell - 1$: \\
\quad\quad\quad write a dummy entry to
    $\mathsf{ReqArray}_\ell[(i-1)\!\cdot\!\mathsf{lenSeq}_\ell + t]$. \\[2pt]
\quad 4: \textbf{return} $\mathsf{ReqArray}_\ell$.
\end{tabular}
\end{minipage}}
\caption{Simulator for Neuron Requester.}
\label{fig:sim-req}
\end{figure}

\noindent\textbf{Simulator description.}
The simulator $\Sim_{\mathsf{req}}$ receives only the public parameters~$P$ and has no access to the real inputs $\{q_1,\ldots,q_B\}$.
It mirrors the structure of $\pi_{\mathsf{req}}$ step by step,
replacing each real input $q_i$ with a random dummy vector $\hat{q}_i$ of the same public dimension.
In step~2, the simulator executes the same MP-WTA procedure on $\hat{q}_i$,
issuing the same sequence of $\F_{\OCmpSet}$ calls.
Since $\F_{\OCmpSet}$ produces a trace independent of its input values,
the trace generated by each $\F_{\OCmpSet}$ invocation on $\hat{q}_i$ is identical to that on $q_i$.
In step~3, the simulator writes dummy entries to the same deterministic positions in the \texttt{Request-Array} as the real protocol.

\begin{proof}[Proof of Theorem~\ref{thm:req}]
We compare the traces of $\pi_{\mathsf{req}}$ and $\Sim_{\mathsf{req}}$ step by step.

\emph{Step~1} (allocation):
Both the protocol and the simulator allocate an array of the same public size~$T$.
The trace (a sequence of memory writes to $T$ consecutive locations) is identical.

\emph{Step~2} (MP-WTA):
For each input $i \in [1\ldots B]$, Algorithm~\ref{alg:oblivious_mpwta_mp} executes the following:
\begin{itemize}[nosep]
    \item Algorithm~\ref{alg:abstract_max2wta} computes the max, second-max, and third-max indices across all $K$ hash functions.
    For each hash function, this requires a fixed number of $\F_{\OCmpSet}$ invocations determined by the window size~$M$, which is public.
    By the security of $\F_{\OCmpSet}$ (Figure~\ref{func:ocmpset}), each invocation produces a trace independent of the operand values.
    \item The perturbation loop (Algorithm~\ref{alg:oblivious_mpwta_mp}, lines 4--9) iterates over all subsets of size $n \in \{1,2,3\}$ among $K$ positions and all tuples with values in $\{2,3\}$.
    The loop bounds depend only on $K$, $r{=}3$, and $N{=}3$, all public.
    The body performs deterministic array copies with no $\F_{\OCmpSet}$ calls.
\end{itemize}
Therefore, the trace produced by step~2 on real input $q_i$ is indistinguishable from that on dummy input $\hat{q}_i$.

\emph{Step~3} (array population):
Both the protocol and the simulator write to position $(i-1)\cdot\mathsf{lenSeq}_\ell + t$ for each $(i, t)$.
The write addresses are deterministic functions of public parameters.
The values written differ (real vs.\ dummy), but values reside in TDX-encrypted memory and are not part of the trace.
The traces are therefore identical.

\emph{Step~4} (return):
No additional memory accesses.

\

\noindent
The next step is replacing each building-block functionality with its concrete oblivious implementation
(e.g., $\pi_{\mathsf{OCmpSet}}$ for $\F_{\OCmpSet}$)
preserves indistinguishability by the sequential composition theorem~\cite[Corollary~7]{canetti2000security},
since all building-block invocations within $\pi_{\mathsf{req}}$ execute sequentially.
(See Appendix~\ref{sec:primitive-replacement} for additional details.)

\

\noindent
Since the traces are identical in every step, we have:
\[
\{\Sim_{\mathsf{req}}(P)\} \approx_c \{\mathsf{view}_{\mathsf{H}}^{\pi_{\mathsf{req}}}(P, D)\}.
\]
\end{proof}

\subsection{Proof of Theorem~\ref{thm:fetch}: Security of Neuron Fetcher}
\label{sec:proof-fetch}

Figure~\ref{fig:pi-fetch} defines the protocol $\pi_{\mathsf{fetch}}$ and
Figure~\ref{fig:sim-fetch} defines the simulator $\Sim_{\mathsf{fetch}}$.

\begin{figure}[ht]
\centering
\fbox{\begin{minipage}{0.96\linewidth}
\textbf{Protocol $\pi_{\mathsf{fetch}}$}
\\ \hrule
\begin{tabular}{@{}l@{}}
$(\mathsf{ReqArray}_\ell',\; \mathsf{LSH}_\ell') \leftarrow \pi_{\mathsf{fetch}}(\mathsf{ReqArray}_\ell,\;
  \mathsf{LSH}_\ell,\; \mathsf{mode},\; P)$: \\[4pt]
\quad\textit{// Phase 1: OHT construction} \\
\quad 1: $(\mathsf{OHT},\; \mathsf{H}_1,\; \mathsf{H}_2) \leftarrow
  \F_{\OHT}.\mathsf{Build}(\mathsf{ReqArray}_\ell,\;$\\
  \quad\quad\quad $\mathsf{numBin},\; \mathsf{binCap})$ \\
\quad\quad\; (Algorithm~\ref{alg:neuron-fetcher}, line 2;
  internally invokes $\F_{\OSort}$). \\[2pt]
\quad\textit{// Phase 2: linear scan of LSH table} \\
\quad 2: \textbf{for} $b = 0$ to $\mathsf{numLshBucs} - 1$: \\
\quad\quad\quad compute $\mathsf{bin}_1 \leftarrow \mathsf{H}_1(b)$,\;
    $\mathsf{bin}_2 \leftarrow \mathsf{H}_2(b)$. \\
\quad\quad\quad \textbf{for} $\mathsf{tier} \in \{1,2\}$: \\
\quad\quad\quad\quad \textbf{for} each of $\mathsf{binCap}$ entries $e$
    in $\mathsf{OHT}.\mathsf{bins}[\mathsf{bin}_{\mathsf{tier}}]$: \\
\quad\quad\quad\quad\quad invoke $\F_{\OCmpSet}$ to test
    $e.\texttt{LSHbucketID} = b$ \\
\quad\quad\quad\quad\quad (and, in \textsc{Write} mode, that $e$ is non-dummy) \\
\quad\quad\quad\quad\quad and conditionally transfer neurons between \\
\quad\quad\quad\quad\quad $\mathsf{LSH}_\ell[b]$ and $e.\texttt{Neuron-Buffer}$ \\
\quad\quad\quad\quad\quad (direction determined by $\mathsf{mode}$). \\
\quad\quad\; (Algorithm~\ref{alg:neuron-fetcher}, lines 3--13.) \\[2pt]
\quad\textit{// Phase 3: restore Request-Array} \\
\quad 3: Collect all entries from $\mathsf{OHT}$ bins into a flat array $A$. \\
\quad 4: $\F_{\OSort}(A,\; \texttt{key} = (\texttt{isDummy},\, \texttt{Batch-Rank}))$. \\
\quad\quad\; (Algorithm~\ref{alg:neuron-fetcher}, lines 14--16.) \\[2pt]
\quad 5: $\mathsf{ReqArray}_\ell' \leftarrow A[0..T-1]$. \\[2pt]
\quad 6: \textbf{return} $(\mathsf{ReqArray}_\ell',\; \mathsf{LSH}_\ell')$.
\end{tabular}
\end{minipage}}
\caption{Protocol for Neuron Fetcher
(in the $(\F_{\OHT}, \F_{\OSort}, \F_{\OCmpSet})$-hybrid model).}
\label{fig:pi-fetch}
\end{figure}

\noindent\textbf{Protocol description.}
The protocol $\pi_{\mathsf{fetch}}$ corresponds directly to Algorithm~\ref{alg:neuron-fetcher}.
It operates in three phases.
In Phase~1 (step~1), the $T$ entries of $\mathsf{ReqArray}_\ell$ are transformed into an Oblivious Hash Table via $\F_{\OHT}.\mathsf{Build}$, which internally invokes $\F_{\OSort}$ to place entries into fixed-capacity bins keyed by \texttt{LSHbucketID}.
The Build procedure uses the two OHT hash functions $\mathsf{H}_1, \mathsf{H}_2$ specified by $P$.
In Phase~2 (step~2), the protocol performs a linear scan over all $\mathsf{numLshBucs}$ buckets of the LSH table.
For each LSH bucket~$b$, it computes the two candidate OHT bins via $\mathsf{H}_1(b)$ and $\mathsf{H}_2(b)$, then iterates over all $\mathsf{binCap}$ entries in each bin.
At every entry, $\F_{\OCmpSet}$ is invoked to test whether the entry's \texttt{LSHbucketID} matches~$b$ (and, in \textsc{Write} mode, that the entry is non-dummy), and to conditionally transfer neurons between $\mathsf{LSH}_\ell[b]$ and the entry's \texttt{Neuron-Buffer};
the direction of transfer is determined by $\mathsf{mode}$.
Because $\F_{\OCmpSet}$ always performs a write regardless of the match outcome, every iteration produces the same trace.
In Phase~3 (steps~3--5), all OHT entries are collected into a flat array and sorted via $\F_{\OSort}$ by $(\texttt{isDummy}, \texttt{Batch-Rank})$, placing non-dummy entries first, followed by dummy padding, to restore the fixed-length $\mathsf{ReqArray}_\ell$.
The protocol returns both the restored $\mathsf{ReqArray}_\ell'$ and $\mathsf{LSH}_\ell'$ (which is unchanged in \textsc{Read} mode and updated in place in \textsc{Write} mode).

\begin{figure}[ht]
\centering
\fbox{\begin{minipage}{0.96\linewidth}
\textbf{Simulator $\Sim_{\mathsf{fetch}}(P)$}
\\ \hrule
\begin{tabular}{@{}l@{}}
$(\mathsf{ReqArray}_\ell',\; \mathsf{LSH}_\ell') \leftarrow \Sim_{\mathsf{fetch}}(P)$: \\[4pt]
\quad 0: 
\; Initialize $T$ dummy entries, each with a fresh random \\
\quad\quad\; \texttt{LSHbucketID}, \texttt{Batch-Rank}, and dummy \texttt{Neuron-Buffer}. \\
\quad\quad\; Initialize a dummy LSH table with $\mathsf{numLshBucs}$ buckets, \\
\quad\quad\; each filled with $\mathsf{PADSIZE}$ dummy neurons. \\[2pt]
\quad\textit{// Phase 1: OHT construction} \\
\quad 1: $(\widehat{\mathsf{OHT}},\; \hat{\mathsf{H}}_1,\;
  \hat{\mathsf{H}}_2) \leftarrow
  \F_{\OHT}.\mathsf{Build}(T\ \text{dummy entries},\;$ \\
  \quad\quad\quad\quad $\mathsf{numBin},\; \mathsf{binCap})$. \\
\quad\textit{// Phase 2: linear scan of dummy LSH table} \\
\quad 2: \textbf{for} $b = 0$ to $\mathsf{numLshBucs} - 1$: \\
\quad\quad\quad compute $\mathsf{bin}_1 \leftarrow \hat{\mathsf{H}}_1(b)$,\;
    $\mathsf{bin}_2 \leftarrow \hat{\mathsf{H}}_2(b)$. \\
\quad\quad\quad \textbf{for} $\mathsf{tier} \in \{1,2\}$: \\
\quad\quad\quad\quad \textbf{for} each of $\mathsf{binCap}$ entries $e$
    in $\widehat{\mathsf{OHT}}.\mathsf{bins}[\mathsf{bin}_{\mathsf{tier}}]$: \\
\quad\quad\quad\quad\quad invoke $\F_{\OCmpSet}$ on dummy values. \\[2pt]
\quad\textit{// Phase 3: restore Request-Array} \\
\quad 3: Collect all entries from $\widehat{\mathsf{OHT}}$ bins into a flat array $A$. \\
\quad 4: $\F_{\OSort}(A,\; \texttt{key} = (\texttt{isDummy},\, \texttt{Batch-Rank}))$. \\[2pt]
\quad 5: $\mathsf{ReqArray}_\ell' \leftarrow A[0..T-1]$. \\[2pt]
\quad 6: \textbf{return} $(\mathsf{ReqArray}_\ell',\; \mathsf{LSH}_\ell')$.
\end{tabular}
\end{minipage}}
\caption{Simulator for Neuron Fetcher.}
\label{fig:sim-fetch}
\end{figure}

\noindent\textbf{Simulator description.}
The simulator $\Sim_{\mathsf{fetch}}$ receives only~$P$ and has no access to the real $\mathsf{ReqArray}_\ell$ or $\mathsf{LSH}_\ell$.
Its OHT build uses the same public hash functions specified by $P$.
It then mirrors the three-phase structure of $\pi_{\mathsf{fetch}}$, replacing all private data with dummy values.
In Phase~1, the simulator invokes $\F_{\OHT}.\mathsf{Build}$ on $T$ dummy entries and obtains $(\widehat{\mathsf{OHT}}, \hat{\mathsf{H}}_1, \hat{\mathsf{H}}_2)$;
the call to $\F_{\OHT}.\mathsf{Build}$ produces a trace that depends only on $T$, $\mathsf{numBin}$, $\mathsf{binCap}$, and $P$, not on the entry contents.
In Phase~2, the simulator iterates over all $\mathsf{numLshBucs}$ buckets and, for each LSH bucket~$b$, computes $\hat{\mathsf{H}}_1(b)$ and $\hat{\mathsf{H}}_2(b)$ to determine which OHT bins to visit.
The simulator then scans all $\mathsf{binCap}$ entries per bin, invoking $\F_{\OCmpSet}$ on dummy values at every entry.
In Phase~3, the simulator collects and sorts entries via $\F_{\OSort}$ on an array of the same public size.

\begin{proof}[Proof of Theorem~\ref{thm:fetch}]
We compare the traces of $\pi_{\mathsf{fetch}}$ and $\Sim_{\mathsf{fetch}}$ phase by phase.
 
\emph{Phase~1} (OHT construction).
Both the protocol and the simulator invoke $\F_{\OHT}.\mathsf{Build}$ on an array of $T$ entries.
By the security of $\F_{\OHT}$ (Figure~\ref{func:oht}), the build trace depends only on the number of entries~$T$, the number of bins~$\mathsf{numBin}$, and the bin capacity~$\mathsf{binCap}$, all of which are public.
The entry contents (real or dummy) do not affect the trace.
The traces are therefore identical.
 
\emph{Phase~2} (linear scan).
In both executions, the OHT hash functions are produced as part of
$\F_{\OHT}.\mathsf{Build}$'s output.
By the security of $\F_{\OHT}$ (Figure~\ref{func:oht}), the build procedure, using the public hash functions specified by $P$, produces a trace and outputs that depend only on $|S|$, $\mathsf{numBin}$, $\mathsf{binCap}$, and $P$, not on the key-value contents.
Since the real protocol and the simulator both use the same public hash functions from $P$, they evaluate the same $(\mathsf{H}_1, \mathsf{H}_2)$ on the same public sequence $b = 0, \ldots, \mathsf{numLshBucs}-1$.
Thus the bin access sequences are identical conditional on $P$.

Note that the hash functions $\mathsf{H}_1, \mathsf{H}_2$ are public.
This does not weaken security because, even if the OHT hash functions were kept private, the Hypervisor could observe from the trace which two OHT bins are accessed for each LSH bucket~$b$, and thereby reconstruct $\mathsf{H}_1(b)$ and $\mathsf{H}_2(b)$ for every~$b$.
Making the hash functions public therefore reveals no information beyond what is already present in the trace.

\emph{Phase~3} (Request-Array restoration).
Both executions collect all OHT bin entries into a flat array of the same public size ($\mathsf{numBin} \times \mathsf{binCap}$ per tier, summed over both tiers)
and invoke $\F_{\OSort}$ with key $(\texttt{isDummy}, \texttt{Batch-Rank})$.
By the security of $\F_{\OSort}$ (Figure~\ref{func:osort}), the sort trace depends only on the array length.
The subsequent extraction of the first $T$ entries accesses the same deterministic positions.
The traces are therefore identical.

\
 
\noindent
Combining all three phases, the traces produced by $\pi_{\mathsf{fetch}}$ and $\Sim_{\mathsf{fetch}}$ are identically distributed.
Replacing each ideal functionality with its concrete oblivious implementation (e.g., $\pi_{\OSort}$ for $\F_{\OSort}$, $\pi_{\mathsf{OHT}}$ for $\F_{\OHT}$, $\pi_{\mathsf{OCmpSet}}$ for $\F_{\OCmpSet}$) preserves indistinguishability by the sequential composition theorem~\cite[Corollary~7]{canetti2000security}, since all building-block invocations within $\pi_{\mathsf{fetch}}$ execute sequentially.
(Refer to Appendix~\ref{sec:primitive-replacement} for additional details.)

\
 
\noindent
We conclude:
\[
\{\Sim_{\mathsf{fetch}}(P)\}  \approx_c \{\mathsf{view}_{\mathsf{H}}^{\pi_{\mathsf{fetch}}}(P, D)\}.
\]
\end{proof}

\subsection{Proof of Theorem~\ref{thm:upd}: Security of Neuron Updater}
\label{sec:proof-upd}

Figure~\ref{fig:pi-upd} defines the protocol $\pi_{\mathsf{upd}}$ and
Figure~\ref{fig:sim-upd} defines the simulator $\Sim_{\mathsf{upd}}$.

\begin{figure}[htb]
\centering
\fbox{\begin{minipage}{0.96\linewidth}
\textbf{Protocol $\pi_{\mathsf{upd}}$}
\\ \hrule
\footnotesize
\begin{tabular}{@{}l@{}}
$\mathsf{ReqArray}_\ell \leftarrow
  \pi_{\mathsf{upd}}(\mathsf{ReqArray}_\ell,\;
  \mathsf{ActiveNodes},\; \mathsf{ActiveVals},\;$\\
  \quad\quad\quad\quad\quad\quad\quad\quad $\{Y_i\},\; \{\mathbf{x}_i\},\; P)$: \\[4pt]
\quad\textit{// Phase 1: Feed-Forward} (Algorithm~\ref{alg:updater-ff}.) \\
\quad 1: \textbf{for} $i = 1$ to $B$, \textbf{for} $r = 1$ to $R_\ell$: \\
\quad\quad\quad retrieve Request-Array neurons via \textsc{SlotNode}; \\
\quad\quad\quad retrieve dense layer~0 neurons in public order; \\
\quad\quad\quad compute activation from previous layer; \\
\quad\quad\quad apply ReLU or Softmax via $\F_{\OCmpSet}$. \\
\quad 2: \textbf{if} Softmax: \textbf{for} $i = 1$ to $B$, \textbf{for} $r = 1$ to $R_\ell$: \\
\quad\quad\quad compute stable exponential;\\
\quad\quad\quad set zero value to dummy slots via $\F_{\OCmpSet}$; \\
\quad\quad\quad accumulate normalization constant. \\[2pt]
\quad\textit{// Phase 2: Backpropagation} (Algorithm~\ref{alg:updater-bp}.) \\
\quad 3: \textbf{for} $i = 1$ to $B$, \textbf{for} $r = 1$ to $R_\ell$: \\
\quad\quad\quad compute output-layer delta;\\
\quad\quad\quad match labels via $\F_{\OCmpSet}$. \\
\quad 3a: \textbf{for} $i = 1$ to $B$, \textbf{for} $n = 1$ to $N_0$: \\
\quad\quad\quad reset dense-layer delta scratch state. \\
\quad 4: \textbf{for} $i = 1$ to $B$, \textbf{for} $r_u = 1$ to $R_{L-1}$, \\
\quad\quad\quad \textbf{for} $n = 1$ to $N_0$: \\
\quad\quad\quad\quad propagate delta (gated by ReLU via $\F_{\OCmpSet}$); \\
\quad\quad\quad\quad accumulate weight gradient. \\
\quad 5: \textbf{for} $i = 1$ to $B$, \textbf{for} $n = 1$ to $N_0$: \\
\quad\quad\quad accumulate first-layer gradients from input features. \\[2pt]
\quad\textit{// Phase 3: Oblivious Merge} (Algorithm~\ref{alg:obl-merge}.) \\
\quad 6: $\F_{\OSort}(\mathsf{ReqArray}_\ell,\; \texttt{key} = \texttt{LSHbucketID})$. \\
\quad 7: \textbf{for} $j = T-1$ down to $1$: \\
\quad\quad\quad invoke $\F_{\OCmpSet}$ to test adjacent entries and \\
\quad\quad\quad conditionally fold gradients; mark donor as dummy. \\
\quad 8: $\F_{\OSort}(\mathsf{ReqArray}_\ell,\; \texttt{key} = \texttt{isDummy})$. \\[2pt]
\quad\textit{// Phase 4: Adam Update} (Algorithm~\ref{alg:adam-update}.)  \\
\quad 9: \textbf{for} $k = 0$ to $T-1$, \textbf{for} $s = 0$ to $\PADSIZE-1$: \\
\quad\quad\quad update moments and weights (pure arithmetic). \\[2pt]
\quad\quad\quad also scan dense layer~0 in public order. \\[2pt]
\quad 10: \textbf{return} $\mathsf{ReqArray}_\ell$.
\end{tabular}
\end{minipage}}
\caption{Protocol for Neuron Updater
(in the $(\F_{\OSort}, \F_{\OCmpSet})$-hybrid model).}
\label{fig:pi-upd}
\end{figure}

\noindent\textbf{Protocol description.}
The protocol $\pi_{\mathsf{upd}}$ comprises four phases corresponding to four sub-algorithms.

Phase~1 (steps~1--2) performs feedforward (Algorithm~\ref{alg:updater-ff}).
For each input~$i$ and slot~$r$, it retrieves Request-Array neurons via \textsc{SlotNode} (Algorithm~\ref{alg:slotnode}) and dense layer~0 neurons in public order.
The activation is computed as a weighted sum over the previous layer's active state.
For ReLU layers, the activation is clamped to zero via $\F_{\OCmpSet}$ (an \texttt{OblGt} followed by an \texttt{OblChoose}).
For the Softmax output layer, a running oblivious maximum is maintained via $\F_{\OCmpSet}$,
followed by a second pass that computes stable exponentials, sets zero value to dummy slots via $\F_{\OCmpSet}$, and accumulates the normalization constant.
All loop bounds ($B$, $R_\ell$, $R_{\ell-1}$) are determined by~$P$.

Phase~2 (steps~3--5) performs backpropagation (Algorithm~\ref{alg:updater-bp}).
Step~3 computes output-layer deltas: for each neuron slot, a scan over $|Y_i|$ labels invokes $\F_{\OCmpSet}$ per label for oblivious matching, then $\F_{\OCmpSet}$ selects the appropriate delta.
Step~4 propagates deltas from the output layer to dense layer~0: the nested loop over $R_{L-1} \times N_0$ pairs invokes $\F_{\OCmpSet}$ to gate the delta propagation by the ReLU derivative and to accumulate weight gradients.
Step~5 accumulates first-layer gradients from input features over $N_0 \times \mathsf{nnz}$ iterations with no oblivious primitives (all indices are public).
All loop bounds are determined by~$P$.

Phase~3 (steps~6--8) performs oblivious merge (Algorithm~\ref{alg:obl-merge}).
Step~6 invokes $\F_{\OSort}$ to group request entries by \texttt{LSHbucketID}.
Step~7 performs a backward linear scan over all $T$ entries, invoking $\F_{\OCmpSet}$ at each pair to test for duplicate \texttt{LSHbucketID}s and conditionally fold gradients.
Step~8 invokes $\F_{\OSort}$ to push marked-dummy entries to the end.

Phase~4 (step~9) applies the Adam optimizer (Algorithm~\ref{alg:adam-update}).
It performs a flat linear scan over all $T \times \PADSIZE$ neuron slots, executing identical fixed arithmetic at every slot, and also scans dense layer~0 in public order.
No oblivious primitives are required; dummy neurons produce zero updates due to their zero gradient accumulators.

\begin{figure}[htb]
\centering
\fbox{\begin{minipage}{0.94\linewidth}
\textbf{Simulator $\Sim_{\mathsf{upd}}(P)$}
\\ \hrule
\footnotesize
\begin{tabular}{@{}l@{}}
$\mathsf{ReqArray}_\ell \leftarrow \Sim_{\mathsf{upd}}(P)$: \\[4pt]

\quad 0: Initialize a dummy $\mathsf{ReqArray}_\ell$ of $T$ entries, each with a \\
\quad\quad\; \texttt{Neuron-Buffer} of $\PADSIZE$ zero-valued neurons. \\
\quad\quad Initialize dummy $\mathsf{ActiveNodes}$, $\mathsf{ActiveVals}$ of public sizes; \\
\quad\quad Initialize dummy labels $\{\hat{Y}_i\}$ and dummy inputs $\{\hat{\mathbf{x}}_i\}$\\
\quad\quad \; of public size $|Y_i|$ and input dimensions.\\
\quad\textit{// Phase 1: Feed-Forward} \\
\quad 1: \textbf{for} $i = 1$ to $B$, \textbf{for} $r = 1$ to $R_\ell$: \\
\quad\quad\quad use \textsc{SlotNode} for Request-Array slots; \\
\quad\quad\quad use public-order access for dense layer~0; \\
\quad\quad\quad perform arithmetic on dummies; \\
\quad\quad\quad invoke $\F_{\OCmpSet}$ for ReLU or Softmax. \\
\quad 2: \textbf{if} Softmax: \textbf{for} $i = 1$ to $B$, \textbf{for} $r = 1$ to $R_\ell$: \\
\quad\quad\quad compute dummy exponential; \\
\quad\quad\quad invoke $\F_{\OCmpSet}$ to zero out dummy slots; \\
\quad\quad\quad accumulate dummy normalization constant. \\[2pt]
\quad\textit{// Phase 2: Backpropagation} \\
\quad 3: \textbf{for} $i = 1$ to $B$, \textbf{for} $r = 1$ to $R_\ell$: \\
\quad\quad\quad \textbf{for} each of $|Y_i|$ dummy labels: invoke $\F_{\OCmpSet}$. \\
\quad\quad\quad invoke $\F_{\OCmpSet}$ to select dummy delta. \\
\quad 3a: \textbf{for} $i = 1$ to $B$, \textbf{for} $n = 1$ to $N_0$: \\
\quad\quad\quad reset dummy dense-layer delta scratch state. \\
\quad 4: \textbf{for} $i = 1$ to $B$, \textbf{for} $r_u = 1$ to $R_{L-1}$, \\
\quad\quad\quad \textbf{for} $n = 1$ to $N_0$: \\
\quad\quad\quad\quad invoke $\F_{\OCmpSet}$ for dummy delta propagation; \\
\quad\quad\quad\quad perform dummy gradient accumulation. \\
\quad 5: \textbf{for} $i = 1$ to $B$, \textbf{for} $n = 1$ to $N_0$: \\
\quad\quad\quad \textbf{for} $j = 1$ to $\mathsf{nnz}$: perform dummy first-layer gradient. \\[2pt]
\quad\textit{// Phase 3: Oblivious Merge} \\
\quad 6: $\F_{\OSort}(\mathsf{ReqArray}_\ell,\; \texttt{key} = \texttt{LSHbucketID})$. \\
\quad 7: \textbf{for} $j = T-1$ down to $1$: \\
\quad\quad\quad invoke $\F_{\OCmpSet}$ on adjacent dummy entries. \\
\quad 8: $\F_{\OSort}(\mathsf{ReqArray}_\ell,\; \texttt{key} = \texttt{isDummy})$. \\[2pt]
\quad\textit{// Phase 4: Adam Update} \\
\quad 9: \textbf{for} $k = 0$ to $T-1$, \textbf{for} $s = 0$ to $\PADSIZE-1$: \\
\quad\quad\quad\; perform dummy arithmetic over $D_\ell$ weight dimensions. \\[2pt]
\quad\quad\quad\; also scan dummy dense layer~0 in public order. \\[2pt]
\quad 10: \textbf{return} $\mathsf{ReqArray}_\ell$.
\end{tabular}
\end{minipage}}
\caption{Simulator for Neuron Updater.}
\label{fig:sim-upd}
\end{figure}

\noindent\textbf{Simulator description.}
The simulator $\Sim_{\mathsf{upd}}$ receives only~$P$ and has no access to the real $\mathsf{ReqArray}_\ell$, activations, labels~$\{Y_i\}$, or input features~$\{\mathbf{x}_i\}$.
It mirrors the four-phase structure of $\pi_{\mathsf{upd}}$, replacing all private data with dummy values (e.g., zero-valued neurons, random labels, zero-valued features).

In Phase~1 (steps~1--2), the simulator initializes dummy data structures of the same public sizes and executes the same feedforward loop structure.
It also executes the dense layer~0 feedforward scan on dummy dense records in the same public order.
For Request-Array slots, it invokes \textsc{SlotNode}; for dense layer~0, it uses the same public-order scan.
The arithmetic (weighted sums, exponentials) is performed on dummy values;
$\F_{\OCmpSet}$ is invoked at exactly the same points as in $\pi_{\mathsf{upd}}$ for ReLU gating, oblivious maximum, and dummy zeroing.

In Phase~2 (steps~3--5), the simulator executes the same backpropagation loop structure.
Step~3 scans $|Y_i|$ dummy labels per slot, invoking $\F_{\OCmpSet}$ at each label and once for delta selection.
Step~4 iterates over the same $R_{L-1} \times N_0$ nested loop, invoking $\F_{\OCmpSet}$ for each pair.
Step~5 iterates over $N_0 \times \mathsf{nnz}$ pairs; since the non-zero feature indices are public (included in~$P$), the simulator accesses the same weight dimensions as the real protocol.

In Phase~3 (steps~6--8), the simulator invokes $\F_{\OSort}$ and $\F_{\OCmpSet}$ on the same-sized dummy arrays, following the identical merge procedure.

In Phase~4 (step~9), the simulator performs the same arithmetic loop over Request-Array slots and all $D_\ell$ weight dimensions, with no oblivious primitives.
For dense layer~0, it performs the same public-order Adam scan over dummy dense records.

\begin{proof}[Proof of Theorem~\ref{thm:upd}]
We compare the traces of $\pi_{\mathsf{upd}}$ and $\Sim_{\mathsf{upd}}$ phase by phase.

\emph{Phase~1} (feedforward):
Both executions perform identical loop structures with bounds $B \times R_\ell$ (Pass~1) and $B \times R_\ell$ (Pass~2, Softmax only), all determined by~$P$.
For Request-Array slots, \textsc{SlotNode} computes a deterministic index from $(i, r, P)$ alone (Algorithm~\ref{alg:slotnode}); dense layer~0 uses public-order addresses, so both executions access the same modeled positions.
The weighted-sum computation (Algorithm~\ref{alg:updater-ff}, lines 6--8) iterates over $R_{\ell-1}$ previous-layer slots, with $R_0=N_0$ for dense layer~0, using publicly determined addressing; the operand values differ but reside in encrypted memory and are not part of the trace.
Each invocation of $\F_{\OCmpSet}$ (for ReLU clamping, oblivious maximum, and dummy zeroing) produces a trace independent of the operand values by Figure~\ref{func:ocmpset}.
The traces are therefore identical.

\emph{Phase~2} (backpropagation):
Step~3 iterates over $B \times R_\ell$ slots.
Within each iteration, the label-matching loop runs $|Y_i|$ times, invoking $\F_{\OCmpSet}$ at each label and once for delta selection.
Both $R_\ell$ and $|Y_i|$ are public.
Before Step~4, both executions perform the same public-order $B \times N_0$ reset of dense-layer delta scratch state.
Step~4 iterates over the output-layer slots and dense layer-0 neurons, with a nested loop of $B \times R_{L-1} \times N_0$ iterations.
At each iteration, both executions invoke $\F_{\OCmpSet}$ for ReLU-gated delta propagation and perform arithmetic accumulation at addresses determined by public loop indices.
Step~5 iterates over $B \times N_0 \times \mathsf{nnz}$ entries.
The weight dimension accessed at each iteration is $d_j$, the $j$-th non-zero feature index, which is public (included in~$P$).
Therefore, both executions access the same addresses.
By the security of $\F_{\OCmpSet}$, all invocations produce value-independent traces.
The traces are therefore identical.

\emph{Phase~3} (oblivious merge):
Step~6 invokes $\F_{\OSort}$ on an array of size~$T$.
By Figure~\ref{func:osort}, the sort trace depends only on~$T$.
Step~7 performs a linear scan from $j = T-1$ down to $1$, invoking $\F_{\OCmpSet}$ at each of the $T-1$ pairs.
By Figure~\ref{func:ocmpset}, each invocation produces a value-independent trace.
Step~8 invokes $\F_{\OSort}$ on the same array of size~$T$.
The traces are therefore identical.

\emph{Phase~4} (Adam update):
Both executions iterate over $T \times \PADSIZE$ neuron slots, and within each slot over $D_\ell$ weight dimensions plus the bias.
For dense layer~0, both executions also perform the same public-order Adam scan over all $N_0$ dense records.
No oblivious primitives are invoked.
All loop bounds ($T$, $\PADSIZE$, $D_\ell$, and $N_0$ for the dense scan) and the bias-corrected learning rate (which depends on $\mathsf{step}$, $\alpha$, $\beta_1$, $\beta_2$, all public) are determined by~$P$.
Every Request-Array slot executes the same sequence of reads and writes at \textsc{SlotNode}-determined addresses, and every dense-layer slot is visited in public order.
The traces are therefore identical.

\

\noindent
The next step is replacing each building-block functionality with its concrete oblivious implementation
(e.g., $\pi_{\OSort}$ for $\F_{\OSort}$,  $\pi_{\mathsf{OCmpSet}}$ for $\F_{\OCmpSet}$)
preserves indistinguishability by the sequential composition theorem~\cite[Corollary~7]{canetti2000security},
since all building-block invocations within $\pi_{\mathsf{upd}}$ execute sequentially.
(Refer to Appendix~\ref{sec:primitive-replacement} for additional details.)

\

\noindent
Since the traces are identical in every phase, we have:
\[
\{\Sim_{\mathsf{upd}}(P)\} \approx_c \{\mathsf{view}_{\mathsf{H}}^{\pi_{\mathsf{upd}}}(P, D)\}.
\]
\end{proof}

\subsection{Replacing Ideal Primitive Calls by Real Primitive Protocols}
\label{sec:primitive-replacement}

We now show the missing step that turns the three hybrid-model theorems into concrete component statements.
For a component $C \in \{\mathsf{req}, \mathsf{fetch}, \mathsf{upd}\}$, let $\pi_C^{\mathcal{F}}$ be the version of that component in which every primitive call is still an ideal call to one of $\F_{\OSort}$, $\F_{\OCmpSet}$, or $\F_{\OHT}$.
Let $g_1, \ldots, g_m$ be these primitive calls listed in program order.
A call list is in program order if $g_i$ happens before $g_{i+1}$ in every execution.
This is well defined here because all primitive calls inside each component run one after another.

\begin{lemma}[Primitive replacement inside one component]
\label{lem:primitive-replacement}
Assume that every real primitive protocol securely realizes its matching functionality.
Then, for each component $C \in \{\mathsf{req}, \mathsf{fetch}, \mathsf{upd}\}$, replacing the ideal calls $g_1, \ldots, g_m$ in $\pi_C^{\mathcal{F}}$ one by one by their real primitive protocols preserves computational indistinguishability.
\end{lemma}

\begin{proof}
Fix one component $C$.
Define a hybrid sequence $H_C^0, H_C^1,$ $ \ldots, H_C^m$ as follows.
In $H_C^0$, all primitive calls are ideal calls.
In $H_C^j$, the first $j$ primitive calls are replaced by their real primitive protocols, and the remaining calls stay ideal.
Thus $H_C^0$ is $\pi_C^{\mathcal{F}}$, and $H_C^m$ is the fully concrete component protocol.

Suppose for contradiction that there is an index $j \in [1..m]$ and a distinguisher $\mathcal{D}$ that tells $H_C^{j-1}$ from $H_C^j$ with non-negligible advantage.
We build an adversary $\mathcal A_j$ against the security of the $j$-th primitive protocol as follows.
Adversary $\mathcal A_j$ simulates the entire component execution for $\mathcal D$.
For the first $j-1$ primitive calls it uses the real primitive protocols.
For the last $m-j$ primitive calls it uses the functionalities.
For the $j$-th call it forwards the inputs of that call to its challenge oracle.
If the challenge oracle is the functionality, then the view generated for $\mathcal D$ is distributed exactly as in $H_C^{j-1}$.
If the challenge oracle is the real primitive protocol, then the view generated for $\mathcal D$ is distributed exactly as in $H_C^j$.
Therefore any non-negligible advantage of $\mathcal{D}$ gives the same non-negligible advantage to $\mathcal{A}_j$.
This contradicts the assumed security of that primitive protocol.
Hence every adjacent pair $H_C^{j-1}$ and $H_C^j$ is computationally indistinguishable.
By chaining these pairs, we get $H_C^0 \approx_c H_C^m$.
\end{proof}

\begin{corollary}[Concrete security of the three components]
\label{cor:components-real}
Assume that $\pi_{\mathsf{OCmpSet}}$, $\pi_{\OSort}$, and $\pi_{\mathsf{OHT}}$ securely realize $\F_{\OCmpSet}$, $\F_{\OSort}$, and $\F_{\OHT}$, respectively.
Then the concrete protocols $\pi_{\mathsf{req}}$, $\pi_{\mathsf{fetch}}$, and $\pi_{\mathsf{upd}}$ securely realize $\F_{\mathsf{req}}$, $\F_{\mathsf{fetch}}$, and $\F_{\mathsf{upd}}$, respectively.
\end{corollary}

\begin{proof}
Theorems~\ref{thm:req}, \ref{thm:fetch}, and~\ref{thm:upd} prove the hybrid-model claims for the three components.
Lemma~\ref{lem:primitive-replacement} then replaces the ideal primitive calls inside each component by the real primitive protocols.
\end{proof}

\paragraph{Optimized Fetcher schedules.}
The Fetcher proof above is written for the logical fixed-schedule Fetcher.
The implementation optimizations from Section~\ref{sec:neuron-fetcher} preserve this claim because they are public refinements of the same schedule.
O1 defers dummy-payload materialization but, before any candidate-slot scan, restores fixed-size \texttt{Neuron-Buffer}s for every OHT slot that can be scanned, including padding dummy slots.
O2 partitions public OHT-bin work across threads according to public scheduler metadata; the proof treats this as a public per-worker schedule, not an arbitrary OS-chosen interleaving.
O3 reuses the feedforward OHT layout only through oblivious sorting or fixed scans, not through direct secret-dependent indexing.
Thus these optimizations add only public scheduling metadata and do not change the simulator's input beyond~$P$.

\subsection{Security of One Batch Step}
\label{sec:proof-batch}

A batch step is one execution of Neuron Requester, Neuron Fetcher in \textsc{Read} mode, Neuron Updater, and Neuron Fetcher in \textsc{Write} mode on one training batch.
The current internal structures of a batch step are the current LSH table, dense layer~0, and the optimizer values stored in their neuron records.
For readability, the displayed batch functionality writes only the LSH-table input/output; dense layer~0 is carried as implicit state across batch steps and updated by $\F_{\mathsf{upd}}$.

\begin{figure}[ht]
\centering
\fbox{\begin{minipage}{0.94\linewidth}
\textbf{Functionality $\F_{\mathsf{batch}}$}
\\ \hrule
\footnotesize
\begin{tabular}{@{}l@{}}
$\mathsf{LSH}_\ell' \leftarrow \F_{\mathsf{batch}}(\ell,\; \{q_i\}_{i=1}^B,\; \{Y_i\}_{i=1}^B,\; \{\mathbf{x}_i\}_{i=1}^B,\; \mathsf{LSH}_\ell,\; P)$: \\[4pt]
\quad 1: $\mathsf{ReqArray}_\ell \leftarrow \F_{\mathsf{req}}(\ell,\; \{q_i\}_{i=1}^B,\; P)$. \\[2pt]
\quad 2: $\mathsf{ReqArray}_\ell, \mathsf{LSH}_\ell \leftarrow \F_{\mathsf{fetch}}(\mathsf{ReqArray}_\ell,\; \mathsf{LSH}_\ell,\; \textsc{Read},\; P)$. \\[2pt]
\quad 3: Allocate fresh $\mathsf{ActiveNodes}$ and $\mathsf{ActiveVals}$ of \\
\quad \; the public sizes fixed by $P$. \\[2pt]
\quad 4: $\mathsf{ReqArray}_\ell \leftarrow \F_{\mathsf{upd}}(\mathsf{ReqArray}_\ell,\; \mathsf{ActiveNodes},\; \mathsf{ActiveVals},\;$ \\
\quad\quad\quad\quad\quad\quad\quad\quad $\{Y_i\}_{i=1}^B,\; \{\mathbf{x}_i\}_{i=1}^B,\; P)$. \\[2pt]
\quad 5: $\mathsf{ReqArray}_\ell, \mathsf{LSH}_\ell' \leftarrow \F_{\mathsf{fetch}}(\mathsf{ReqArray}_\ell,\; \mathsf{LSH}_\ell,\; \textsc{Write},\; P)$. \\[2pt]
\quad 6: \textbf{return} $\mathsf{LSH}_\ell'$.
\end{tabular}
\end{minipage}}
\caption{Ideal functionality for one batch step.}
\label{func:batch}
\end{figure}

\begin{figure}[ht]
\centering
\fbox{\begin{minipage}{0.94\linewidth}
\textbf{Protocol $\pi_{\mathsf{batch}}$}
\\ \hrule
\begin{tabular}{@{}l@{}}
$\mathsf{LSH}_\ell' \leftarrow \pi_{\mathsf{batch}}(\ell,\; \{q_i\}_{i=1}^B,\; \{Y_i\}_{i=1}^B,\; \{\mathbf{x}_i\}_{i=1}^B,\; \mathsf{LSH}_\ell,\; P)$: \\[4pt]
\quad 1: Call $\F_{\mathsf{req}}$ exactly as in Figure~\ref{func:batch}, step~1. \\[2pt]
\quad 2: Call $\F_{\mathsf{fetch}}$ in \textsc{Read} mode exactly as in Figure~\ref{func:batch}, step~2. \\[2pt]
\quad 3: Allocate fresh $\mathsf{ActiveNodes}$ and $\mathsf{ActiveVals}$ of \\
\quad \; the public sizes fixed by $P$. \\[2pt]
\quad 4: Call $\F_{\mathsf{upd}}$ exactly as in Figure~\ref{func:batch}, step~4. \\[2pt]
\quad 5: Call $\F_{\mathsf{fetch}}$ in \textsc{Write} mode exactly as in Figure~\ref{func:batch}, step~5. \\[2pt]
\quad 6: \textbf{return} $\mathsf{LSH}_\ell'$.
\end{tabular}
\end{minipage}}
\caption{One batch step in the component hybrid model.}
\label{fig:pi-batch}
\end{figure}

\begin{figure}[ht]
\centering
\fbox{\begin{minipage}{0.94\linewidth}
\textbf{Simulator $\Sim_{\mathsf{batch}}(P)$}
\\ \hrule
\begin{tabular}{@{}l@{}}
$\widehat{\mathsf{LSH}}_\ell' \leftarrow \Sim_{\mathsf{batch}}(P)$: \\[4pt]
\quad 0: Generate dummy batch inputs $\{\hat{q}_i\}_{i=1}^B$, \\ 
\quad \; dummy labels $\{\hat{Y}_i\}_{i=1}^B$,  and dummy features $\{\hat{\mathbf{x}}_i\}_{i=1}^B$ \\
 \quad \; of the same public shapes as the real batch. \\[2pt]
\quad \; Generate a dummy current LSH table $\widehat{\mathsf{LSH}}_\ell$ \\
\quad \; with the same public shape as the real current LSH table. \\[2pt]
\quad 1: Call $\F_{\mathsf{req}}$, $\F_{\mathsf{fetch}}$ in \textsc{Read} mode, $\F_{\mathsf{upd}}$, and \\
\quad \; $\F_{\mathsf{fetch}}$ in \textsc{Write} mode in the same order as $\pi_{\mathsf{batch}}$. \\[2pt]
\quad 2: \textbf{return} $\widehat{\mathsf{LSH}}_\ell'$.
\end{tabular}
\end{minipage}}
\caption{Simulator for one batch step.}
\label{fig:sim-batch}
\end{figure}

\begin{theorem}[Hybrid-model security of one batch step]
\label{thm:batch-hybrid}
Protocol $\pi_{\mathsf{batch}}$ securely realizes $\F_{\mathsf{batch}}$ in the $(\F_{\mathsf{req}}, \F_{\mathsf{fetch}}, \F_{\mathsf{upd}})$-hybrid model.
\end{theorem}

\begin{proof}
The protocol and the simulator make the same four component calls in the same order.
The only work done outside these four calls is the allocation of $\mathsf{ActiveNodes}$ and $\mathsf{ActiveVals}$, whose sizes are fixed by the public parameters.
The simulator uses dummy inputs, a dummy current LSH table, and dummy dense layer~0 state of the same public shapes as the real ones.
By the definitions of $\F_{\mathsf{req}}$, $\F_{\mathsf{fetch}}$, and $\F_{\mathsf{upd}}$, each call reveals only the fixed public call schedule and returns an object of the same public shape.
Hence the trace of $\pi_{\mathsf{batch}}$ is computationally indistinguishable from the trace produced by $\Sim_{\mathsf{batch}}(P)$.
\end{proof}

\begin{theorem}[Concrete security of one batch step]
\label{thm:batch-real}
Assume the premise of Corollary~\ref{cor:components-real}.
Then the concrete batch protocol $\pi_{\mathsf{batch}}^{\pi_{\mathsf{req}},\pi_{ \mathsf{fetch}}, \pi_{\mathsf{upd}}}$ obtained by replacing the four component calls in $\pi_{\mathsf{batch}}^{\mathcal{F}_{\mathsf{req}},\mathcal{F}_{ \mathsf{fetch}}, \mathcal{F}_{\mathsf{upd}}}$ with $\pi_{\mathsf{req}}$, $\pi_{\mathsf{fetch}}$, $\pi_{\mathsf{upd}}$, and $\pi_{\mathsf{fetch}}$ securely realizes $\F_{\mathsf{batch}}$.
\end{theorem}

\begin{proof}
Define hybrids $H_0, H_1, H_2, H_3, H_4$.
In $H_0$, all four component calls are functionalities, so $H_0$ is exactly $\pi_{\mathsf{batch}}^{\mathcal{F}_{\mathsf{req}},\mathcal{F}_{ \mathsf{fetch}}, \mathcal{F}_{\mathsf{upd}}}$.
In $H_1$, replace only the first call, namely $\F_{\mathsf{req}}$, by $\pi_{\mathsf{req}}$, which is $\pi_{\mathsf{batch}}^{\pi_{\mathsf{req}},\mathcal{F}_{ \mathsf{fetch}}, \mathcal{F}_{\mathsf{upd}}}$.
In $H_2$, also replace the \textsc{Read} call to $\F_{\mathsf{fetch}}$ by $\pi_{\mathsf{fetch}}$, which is $\pi_{\mathsf{batch}}^{\pi_{\mathsf{req}},\pi_{ \mathsf{fetch}}, \mathcal{F}_{\mathsf{upd}}}$.
In $H_3$, also replace $\F_{\mathsf{upd}}$ by $\pi_{\mathsf{upd}}$, which is $\pi_{\mathsf{batch}}^{\pi_{\mathsf{req}},\pi_{ \mathsf{fetch}}, \pi_{\mathsf{upd}}}$.
In $H_4$, also replace the \textsc{Write} call to $\F_{\mathsf{fetch}}$ by $\pi_{\mathsf{fetch}}$.
Thus $H_4$ is the concrete protocol $\pi_{\mathsf{batch}}^{\pi_{\mathsf{req}},\pi_{ \mathsf{fetch}}, \pi_{\mathsf{upd}}}$.

Suppose some adjacent pair $H_{j-1}$ and $H_j$ can be distinguished with non-negligible advantage.
We build an adversary against the concrete security of the component replaced at step $j$.
That adversary runs the full batch step internally.
For the first $j-1$ component calls it uses the real component protocols.
For the last $4-j$ component calls it uses the component functionalities.
For the $j$-th component call it forwards the inputs of that call to its challenge oracle.
If the oracle is the functionality, the internal execution is distributed exactly as $H_{j-1}$.
If the oracle is the real component protocol, the internal execution is distributed exactly as $H_j$.
Hence, any distinguisher for $H_{j-1}$ and $H_j$ can be used to build an adversary for the security of that component, which contradicts Corollary~\ref{cor:components-real}.
Thus, every adjacent pair is computationally indistinguishable.
By chaining the four pairs, $H_0 \approx_c H_4$.
Together with Theorem~\ref{thm:batch-hybrid}, this proves the theorem.
\end{proof}

\subsection{Security of the Full Training Protocol}
\label{sec:proof-system}

\noindent\textbf{Ideal Functionality for LSH Refresh.}
For clarity, we define the ideal functionality $\mathcal{F}_{\mathsf{refresh}}$ for LSH initialization and refresh.
Parameters such as LSH size, LSH bucket capacity, and array lengths below are fixed by the public parameter set $P$.
For simplicity, the functionality specifies only the logical operations of the LSH.
Its detailed oblivious realization, including fixed-size oblivious sorting and branch-free capacity enforcement, is given by \textsc{OblRefreshLSH} in Algorithm~\ref{alg:lsh-refresh-obl} and \textsc{EnforceCapacity} in Algorithm~\ref{alg:enforce-capacity} shown in Appendix~\ref{app:lsh-refresh}.

\begin{figure}[ht]
\centering
\fbox{\begin{minipage}{0.96\linewidth}
\textbf{Functionality $\F_{\mathsf{refresh}}$}
\\ \hrule
\begin{tabular}{@{}l@{}}
$\mathsf{LSHTable}'_\ell \leftarrow
  \F_{\mathsf{refresh}}(\mathsf{LSHTable}_\ell,\;
  \mathsf{isInit},\; P)$: \\[4pt]
\quad 1: Collect all neuron records from the main and overflow \\
\quad\quad\; regions into a working array $A$. \\[2pt]
\quad 2: If $\mathsf{isInit}=1$, append the fixed dummy reservoir \\
\quad\quad\; determined by $P$ to $A$. \\[2pt]
\quad 3: Recompute the WTA signature for every real neuron; \\
\quad\quad\; each dummy neuron retains its predetermined signature. \\[2pt]
\quad 4: Group all neurons by their assigned buckets, placing real \\
\quad\quad\; neurons before dummy neurons within each bucket. \\[2pt]
\quad 5: Retain up to \texttt{PADSIZE} neurons in each main-region \\
\quad\quad\; bucket and place all remaining neurons in the fixed \\
\quad\quad\; overflow region determined by $P$. \\[2pt]
\quad 6: \textbf{return} refreshed $\mathsf{LSHTable}'_\ell$.
\end{tabular}
\end{minipage}}
\caption{Functionality for LSH initialization and refresh.}
\label{func:refresh}
\end{figure}

We now prove the top-level statement.
The full training protocol follows a public event schedule: one initial LSH construction event, a public-length sequence of batch steps, and periodic LSH refresh events at public batch indices determined by~$P$.
Let $N_{\mathsf{batch}}$ be the total number of batch steps, which is determined by the public training configuration in $P$.
Let $N_{\mathsf{event}}$ be the total number of public initialization, batch, and refresh events.
At the end of the last batch step, the trained model $M$ is read from the final internal structures and returned to the User.
The Hypervisor still sees only the trace.

\begin{theorem}[Security of \system]
\label{thm:system-final}
Assume the premise of Corollary~\ref{cor:components-real}. 
Also assume that the LSH initialization and refresh procedure of Appendix~\ref{app:lsh-refresh} operates on a main region (LSH table) of fixed size $\mathsf{numLshBucs}\times\texttt{PADSIZE}$ and an overflow region of fixed capacity $C_{\mathsf{ovf}}$, where these capacities, the total number of neurons, the refresh schedule, and any rule used to select \texttt{PADSIZE} are included in the public parameters $P$. 
Then, for the two-layer \system architecture formalized above, protocol $\pi_{\system}$ securely realizes $\mathcal{F}_{\mathsf{train}}$ in the sense of Definition~\ref{def:security-working}, regardless of whether the overflow region contains only dummy entries or also real neurons.
\end{theorem}

\begin{proof}
Define hybrids $G_0, G_1, \ldots, G_{N_{\mathsf{event}}}$ over the public event schedule.
In $G_t$, the first $t$ events are concrete events, and the remaining $N_{\mathsf{event}}-t$ events are simulated events driven by dummy state of the same public shape.
Batch events use the concrete protocol $\pi_{\mathsf{batch}}$ or the ideal functionality $\F_{\mathsf{batch}}$ as above.

Initialization and refresh events use either \textsc{OblRefreshLSH} from Algorithm~\ref{alg:lsh-refresh-obl} or the ideal functionality $\mathcal{F}_{\mathsf{refresh}}$ defined above in Figure~\ref{func:refresh}.

Thus $G_0$ is the fully simulated training loop, and $G_{N_{\mathsf{event}}}$ is the real execution of $\pi_{\system}$.

Fix any $t \in [1..N_{\mathsf{event}}]$.
The only difference between $G_{t-1}$ and $G_t$ is the $t$-th public event.
All earlier events are identical in the hybrids, so they produce the same public LSH-table shape before event $t$.
All later events use the same public schedule on state with the same public shape, so their traces are identical.
If event $t$ is a batch step, any distinguisher gives a distinguisher for one batch step, contradicting Theorem~\ref{thm:batch-real}.

If event $t$ is an initialization or refresh event, the real execution and the simulation process arrays of the same fixed length determined by $P$.
Both executions compute a WTA candidate bucket (i.e., WTA signature) for every neuron using the same fixed loop structure, invoke the same fixed-size oblivious sorts, execute algorithm \textsc{EnforceCapacity} as a branch-free linear scan over all neurons, and write the same fixed numbers of neurons to the main (LSH table) and overflow regions.
The values of \texttt{newID}, \texttt{isDummy}, and \texttt{isOverflow} may differ, but they affect only encrypted values and oblivious sort keys, not the access trace.
Therefore, based on \textsc{OblRefreshLSH} (Algorithm~\ref{alg:lsh-refresh-obl}) and \textsc{EnforceCapacity} (Algorithm~\ref{alg:enforce-capacity}), the existing overflow area handles overflowing neurons obliviously.
Whether an overflow entry contains a real or a dummy neuron affects only the contents of the refreshed LSH, not its observable memory access trace.
Consequently, the concrete refresh procedure realizes $\mathcal{F}_{\mathsf{refresh}}$ as shown in Figure~\ref{func:refresh}, regardless of the specifics of the run and its overflow, i.e., how many neurons overflow, if any.

The refreshed LSH has the same fixed main-region and overflow-region sizes in both executions.
All subsequent batch and refresh events therefore continue to operate on the same fixed sizes and event schedule, so replacing this event preserves the validity of the remaining hybrid steps.
Hence every adjacent pair $G_{t-1}$ and $G_t$ is computationally indistinguishable.
By chaining these pairs, $G_0\approx_c G_{N_{\mathsf{event}}}$.

In $G_0$, the simulator sees only $P$ and produces a trace by running the dummy initialization, batch-step, and refresh simulations according to the public event schedule.
In $G_{N_{\mathsf{event}}}$, the Hypervisor sees the real trace of $\pi_{\system}$.
The hybrid argument concerns only the Hypervisor's visible trace.
The model returned to the User is handled separately by correctness.
In the ideal world, the output model is produced by $\F_{\mathsf{train}}$, while in the real world it is produced by the execution of $\pi_{\system}$.
By correctness, the honest-party output of $\pi_{\system}$ matches the output specified by $\F_{\mathsf{train}}$.
Thus, the only part that needs simulation is the Hypervisor's trace.
We conclude
\[
    \bigl\{\Sim(P)\bigr\}
    \;\approx_c\;
    \bigl\{\mathsf{view}_{\mathsf{H}}^{\pi_{\system}}(P, D)\bigr\},
\]
which is exactly Definition~\ref{def:security-working}.
\end{proof}

\section{Cost of Hiding the Sparse Input}
\label{app:input-support}

As stated in Section~\ref{sec:threat-model}, \system treats the number and positions of non-zero input features as public parameters.
This appendix describes a compatible full-dimensional scan that removes this exposure and reports its measured cost on our evaluated workloads.

\noindent\textbf{Construction.}
The front-end hides the sparse input by replacing the sparse traversal of only non-zero features with a fixed-order scan over the entire original input space.
Let $D$ denote the original input dimensionality.
For every input, the front-end visits each dimension $d\in[0,D)$, including dimensions whose values are zero, and performs the same first-layer computation.
No feature is removed or combined; the front-end changes only which input positions are processed.
Concretely, for each input, \system\ processes all $D$ feature positions in a fixed order: 
the feed-forward pass of Algorithm~\ref{alg:updater-ff} iterates over $d\in[0,D)$ rather than only the $\mathrm{nnz}$ non-zero positions, and the
first-layer gradient accumulation in Phase~3 of Algorithm~\ref{alg:updater-bp} likewise updates $v.\mathtt{t}[d]$ for every $d\in[0,D)$.
Feature and neuron ranges are assigned deterministically to different threads.
Because the loop bounds and memory addresses are fixed by batch size $B$, the number of the first layer's neurons$N_0$, and input dimension $D$, feature values affect only arithmetic operands and not the access pattern.
Each position executes the same multiply-accumulate operations regardless of whether the value is zero or non-zero, so an oblivious comparison or oblivious selection primitive is not required.
The security argument of Appendix~\ref{app:proof} therefore still holds with this front-end, since $B, N_0, D$ are already included in the public parameter set $P$.

\noindent\textbf{Complexity and Measurement.}
For a batch of $B$ inputs, the full-dimensional \textit{feed-forward} pass performs $B \cdot N_0 \cdot D$ multiplication operations, where $N_0$ is the number of neurons in the first layer, and $D$ is the dimensionality of the input.
As for \textit{back-propagation}, first-layer gradient accumulation performs another $B \cdot N_0 \cdot D$ operations.
No input-gradient pass is needed because backpropagation does not continue past the first layer to compute gradients with respect to the input.
The resulting total cost is $2B\cdot N_0 \cdot D$ operations, or $O(B \cdot N_0 \cdot D)$ total work and $O(\frac{B \cdot N_0 \cdot D}{p})$ ideal parallel time with $p$ threads.
Table~\ref{tab:input-support} reports the measured server-side computation time of both passes.

\begin{table}[tbh]
\centering
\caption{Measured server-side computation cost of the
full-dimensional input scan at $B{=}32$, relative to \system's per-batch
runtime at 44 threads (Table~\ref{tab:eval-runtime-32}).
The scan time includes feed-forward computation and first-layer gradient
accumulation, but excludes client-to-\system\ transmission.
Measurements use the same Google Cloud Confidential C3 instance and
\texttt{-O0} build as Section~\ref{sec:evaluation}; we report the median of
10 repetitions after 3 warm-up rounds, with a standard deviation below
$0.4\%$ of the median in every configuration.}
\label{tab:input-support}
\small
\begin{tabular}{lrrrrr}
\toprule
Dataset & $D$ & $N_0$ & Front-end & \system & Overhead \\
\midrule
Wiki10-31K & 101,938   & 256 & 0.124\,s & 248.3\,s  & 0.050\% \\
Amz-Sub    & 135,909   & 128 & 0.089\,s & 695.7\,s  & 0.013\% \\
Wiki-325K  & 1,617,899 & 128 & 1.083\,s & 1578.3\,s & 0.069\% \\
\bottomrule
\end{tabular}
\end{table}

As shown in Table~\ref{tab:input-support}, the full-dimensional scan adds at most $0.07\%$ to \system's server-side per-batch runtime on the evaluated workloads with batch size $B=32$ and thread number $p=44$.\\

\noindent\textbf{Factors Contributing to the Low Overhead.}
Several properties contribute to this low measured overhead.
First, the fixed-order scan has high memory locality and can be partitioned deterministically across threads.
Second, each iteration consists only of multiplication operations; unlike the Neuron Fetcher, it requires no oblivious sorting, comparison, or selection.
Finally, the relative overhead is particularly small because \system's per-batch runtime is dominated by the oblivious sorting in the Neuron Fetcher (Figure~\ref{fig:breakdown-bs32} in Section~\ref{sec:eval-performance}), rather than by first-layer arithmetic.

\section{Effect of \texttt{PADSIZE} on Accuracy}
\label{app:padsize-accuracy}

Following the methodology of Section~\ref{sec:acc-eval}, we evaluate the effect of \texttt{PADSIZE} on the NN accuracy. 
We highlight here that for this experiment we use the baseline \texttt{SLIDE} implementation configured with our newly introduced \texttt{MP-WTA}. 
That is, the tested pipeline is configured without \system's TEE or oblivious execution mechanisms, since they do not affect NN training accuracy.
This experiment therefore isolates the effect of \texttt{PADSIZE} on predictive accuracy and does not concern itself with the time performance of \system and its oblivious mechanisms.

We evaluate small and medium-sized datasets,  Wiki10-31K and Amazon-Sub, with \texttt{PADSIZE} set to 128, 256, and 512.
The default value of \texttt{PADSIZE}=128 used in our main experiments follows the configuration adopted by prior work~\cite{chen2020slide}.
Except for \texttt{PADSIZE}, all model architectures, training hyperparameters, hashing configurations, batch sizes, training schedules, and evaluation procedures are identical to those used in Section~\ref{sec:acc-eval}.
As in that section, accuracy is measured using Precision at 1 (P@1).
Each curve reports the average of three independent runs.

\begin{figure}[h]
  \centering

  \begin{subfigure}[b]{0.49\linewidth}
    \centering
    \includegraphics[width=\linewidth]{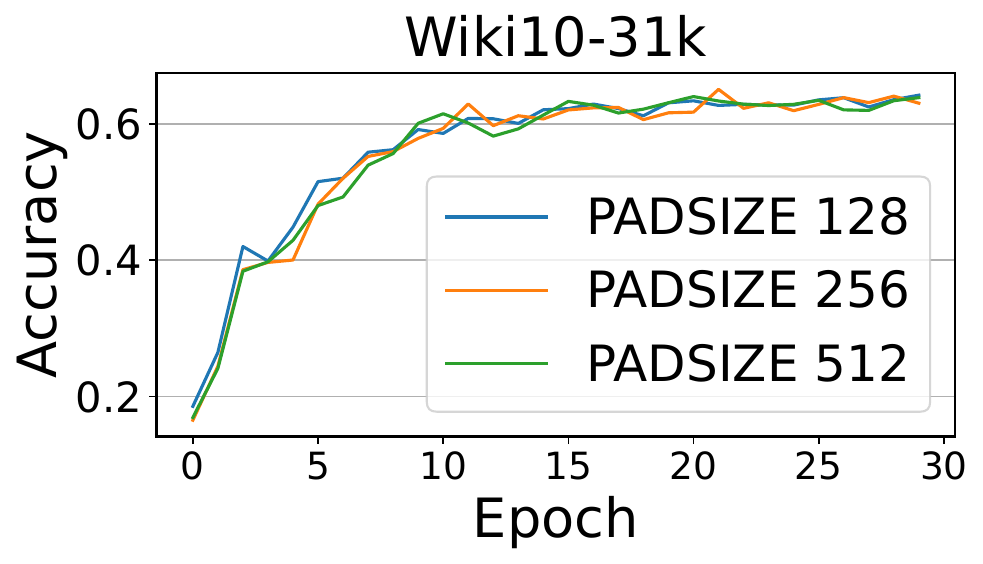}
    \caption{Wiki10-31K}
    \label{fig:padsize-wiki10}
  \end{subfigure}
  \hfill
  \begin{subfigure}[b]{0.49\linewidth}
    \centering
    \includegraphics[width=\linewidth]{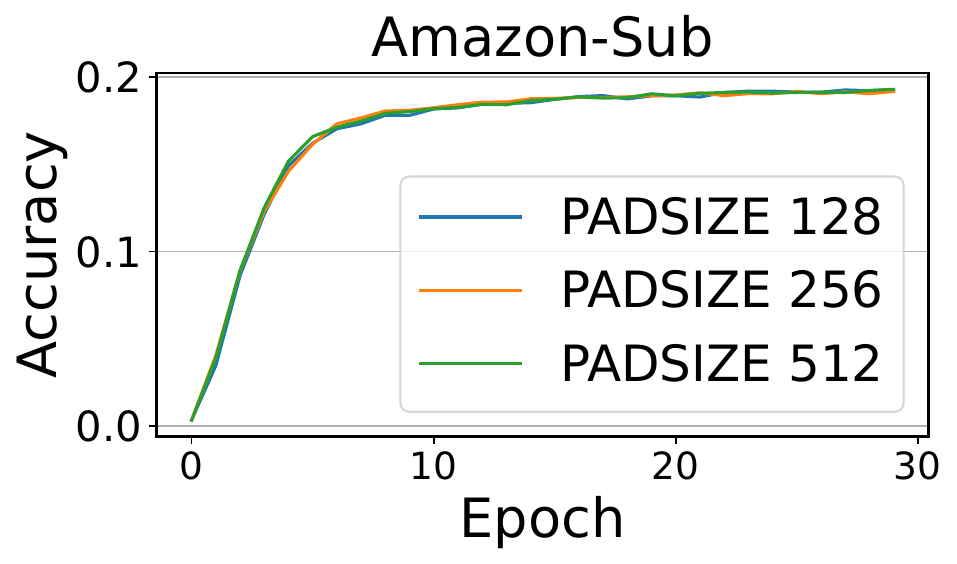}
    \caption{Amazon-Sub}
    \label{fig:padsize-amazon-sub}
  \end{subfigure}

  \caption{Classification accuracy over training epochs with $\texttt{PADSIZE}\in\{128,256,512\}$ on Wiki10-31K and Amazon-Sub. All other experimental parameters are identical to those used in Section~\ref{sec:acc-eval}.}
  \label{fig:padsize-accuracy}
\end{figure}

Figure~\ref{fig:padsize-accuracy} shows that all three configurations follow closely aligned accuracy curves and converge to comparable final P@1 on the tested datasets.

Figure~\ref{fig:padsize-accuracy} shows that all three configurations follow closely aligned accuracy trajectories and converge to comparable final P@1 on both datasets.
The maximum relative final-accuracy difference is approximately 0.24\% on Wiki10-31K and 0.33\% on Amazon-Sub.
Although larger values provide additional bucket capacity and can therefore reduce overflow caused by capacity limits, increasing \texttt{PADSIZE} from 128 to 512 does not produce a consistent accuracy improvement.
These results indicate that, for the evaluated workloads, using the default \texttt{PADSIZE}=128 instead of the larger capacities has no practically significant effect on predictive accuracy.

Although larger values provide additional LSH bucket capacity and may therefore reduce neuron overflow caused by capacity limits, increasing \texttt{PADSIZE} from 128 to 512 does not produce a consistent accuracy improvement.
These results indicate that, for the evaluated workloads, using the default \texttt{PADSIZE}=128, rather than larger capacities, has no practically significant effect on the neural network's predictive accuracy.

\end{document}